\documentclass[11pt]{article}
\usepackage{geometry}                
\usepackage[pdftex]{graphicx,color}
\usepackage{amssymb,amsmath,amsthm}
\usepackage{epstopdf}
\usepackage{wrapfig}
\usepackage{bbold}

\usepackage{amssymb,amsmath,amsthm}
\usepackage{epstopdf}
\usepackage{wrapfig}
\usepackage{graphicx}
\usepackage{arydshln}
\usepackage{setspace}
\usepackage[normalem]{ulem}
\usepackage{cite}

\usepackage{mathtools}
\usepackage{blkarray, bigstrut}

\usepackage{marvosym}
\usepackage{tikz}
\usepackage{verbatim}

\usepackage{authblk}

\usepackage{tikz}

\usepackage{tikz}
\usetikzlibrary{arrows}
\usetikzlibrary{decorations.pathmorphing}
\usetikzlibrary{decorations.markings}
\usetikzlibrary{patterns}
\usetikzlibrary{automata}
\usetikzlibrary{positioning}
\usepackage{tikz-cd}
\tikzset{->-/.style={decoration={
			markings,
			mark=at position #1 with {\arrow{latex}}},postaction={decorate}}}

\tikzset{-<-/.style={decoration={
			markings,
			mark=at position #1 with {\arrowreversed{latex}}},postaction={decorate}}}

\usetikzlibrary{shapes.misc}\tikzset{cross/.style={cross out, draw, 
		minimum size=2*(#1-\pgflinewidth), 
		inner sep=0pt, outer sep=0pt}}

\usepackage{pgfplots}

\usepackage[pdftex,bookmarks,colorlinks,breaklinks]{hyperref}
\definecolor{dullmagenta}{rgb}{0.4,0,0.4}   
\definecolor{darkblue}{rgb}{0,0,0.4}
\hypersetup{linkcolor=red,citecolor=blue,filecolor=dullmagenta,urlcolor=darkblue}

\def\dblue#1{\textcolor[rgb]{0,0,0.7}{#1}}

\newcommand{\ii}{{\rm i}}
\newcommand{\ee}{{\rm e}}
\newcommand{\dd}{{\rm d}}
\newcommand{\CC}{{\mathbb C}}
\newcommand{\RR}{{\mathbb R}}
\newcommand{\DD}{{\mathbb D}}

\def\frak{\mathfrak}

\newtheorem{thm}{Theorem}[section]
\newtheorem{prop}[thm]{Proposition}
\newtheorem{defn}[thm]{Definition}
\newtheorem{lemma}[thm]{Lemma}

\newtheorem{cor}[thm]{Corollary}
\newtheorem{remark}{Remark}[section]

\hypersetup{colorlinks=true, linkcolor=black}

\renewcommand{\descriptionlabel}[1]%
         {\dblue{#1:}\\}

\title{Strong Asymptotics of Planar Orthogonal Polynomials: \\Gaussian Weight Perturbed by Finite Number of Point Charges}

\author[1]{Seung-Yeop Lee}
\author[2]{Meng Yang}
\affil[1]{Department of Mathematics and Statistics, University of South Florida}
\affil[2]{Institut de Recherche en Math\'ematique et Physique, Universit\'e catholique de Louvain}

\date{}                                           

\begin{document}
\maketitle

\begin{abstract}
We consider the orthogonal polynomial $p_{n}(z)$ with respect to the planar measure supported on the whole complex plane 
\begin{equation*}
\ee^{-N|z|^2} \prod_{j=1}^\nu |z-a_j|^{2c_j}\,\dd
A(z)\end{equation*} 
where $\dd A$ is the Lebesgue measure of the plane, $N$ is a positive constant, 
$\{c_1,\dots,c_\nu\}$ are nonzero real numbers greater than $-1$ and
$\{a_1,\dots,a_\nu\}\subset{\mathbb D}\setminus\{0\}$ are distinct points inside the unit disk. In the scaling limit when $n/N = 1$ and $n\to \infty$ we obtain the strong asymptotics of the polynomial $p_n(z)$.   We show that the support of the roots converges to what we call the ``multiple Szeg\H o curve," a certain connected curve having $\nu+1$ components in its complement.  We apply the nonlinear steepest descent method \cite{Deift 1999,DKMVZ 1999} on the matrix Riemann-Hilbert problem of size $(\nu+1)\times(\nu+1)$ posed in \cite{Lee 2017}.      
\end{abstract}

\tableofcontents

\section{Introduction and Main Result}
Let $\{c_1,\dots,c_\nu\}$ be a set of nonzero real numbers greater than $-1$ and
$\{a_1,\dots,a_\nu\}$ be a set of distinct points inside the unit disk.
Let $p_{n,N}(z)$ be the monic polynomial of degree $n$ satisfying the
orthogonality relation
\begin{equation}\label{eq1}
\int_\CC p_{n,N}(z)\,\overline{p_{m,N}(z)}\,\ee^{-N|z|^2} \prod_{j=1}^\nu|z-a_j|^{2c_j}\,\dd
A(z)=h_n\delta_{nm},\quad n,m\geq 0.\end{equation} Here $\dd A$ is
the Lebesgue area measure on the complex plane, $N$ is a positive constant, $h_n$ is the positive norming constant and $\delta_{nm}=1$ when $n=m$ and $\delta_{nm}=0$ when $n\neq m$.  

We consider the scaling limit where $n$ and $N$ both go to $\infty$ while $\lim_{n\to\infty} n/N=1$. 
We will set $N=n$ without losing generality since the
orthogonality gives the relation
$$ p_{n,\, N}(z; a_1,\dots,a_\nu) = \left(\frac{n}{N}\right)^{n/2} p_{n,\, n}\left(\sqrt{\frac{N}{n}}z; \sqrt{\frac{N}{n}}a_1,\dots,\sqrt{\frac{N}{n}}a_\nu\right).$$

\begin{remark}
The asymptotic behavior of the orthogonal polynomials for the planar measure given by $\exp(-N Q(z))\dd A(z)$ for a general external field $Q:{\mathbb C}\to {\mathbb R}$ has been an open problem in relation to the normal matrix model, two dimensional Coulomb gas and Hele-Shaw problems \cite{Zabrodin 2006, Zabrodin 2011}.
The asymptotic behaviors are known only for special choices of $Q$ \cite{Ba 2015,Ba1 2017, Bertola 2018,ku94 2015,ku103 2015, ku104 2015,ku105 2015,ku106 2015,HM 2013,Ra 2005,Lee 2017, Martinez 2019}.  For a general class of $Q$ Hedenmalm and Wennman \cite{Hedenmalm 2017} have found the asymptotic behavior of the orthogonal polynomials outside the ``droplet".  This general result still does not identify the limiting support of the roots, because the roots are mostly found --- except a finite number of them --- inside the droplet as their results have reassured. Our main goal is to find the strong asymptotics for the new class of planar orthogonal polynomials given in \eqref{eq1}.  
\end{remark}

\begin{remark}
We do not consider the case when (some) $a_j$'s are outside ${\mathbb D}$.   The reason is partly because we have been motivated by the results of \cite{Webb 2018} and \cite{Alfredo 2019}, where the main question is the asymptotic behavior of the partition function of the Coulomb gas ensemble as the function of $\{a_j\}_{j=1}^{\nu}\subset{\mathbb D}$ and $\{c_j\}_{j=1}^{\nu}$.  This problem will be studied in our subsequent publication based on the results of this paper. Another application of our results can be the universal behavior of the Coulomb gas in the vicinity of a point singularity (such as $a_j$) which has been studied in \cite{Ameur 2001} using Ward's equation.
\end{remark}

\noindent {\bf Notations.}
We set $N=n$. Though $N=n$ we will keep both $N$ and $n$, preserving their separate roles --- $N$ as a real-valued parameter and $n$ as the integer-valued degree of polynomial --- as much as possible. 
We define $p_n(z)=p_{n,n}(z).$ We denote $\DD=\{z:|z|<1\}$ and $\sum c=\sum_{j=1}^\nu c_j$. We use both the bar and the superscript $*$ for the complex conjugation, e.g., $\bar{z}$ and $z^*$.
\bigskip

\noindent{{\bf Asymptotics for $\nu=1$:}}
When $\nu=1$ the full asymptotic behavior has been found \cite{Lee 2017}; the roots of the polynomial converge towards the generalized Szeg\H{o} curve that depends only on $a_1$ but not on $c_1$. See Figure \ref{fig single}.   The limiting support of the roots is given by the simple closed curve (which is exactly the {\em Szeg\H o curve} when $a_1=1$)
\begin{equation}\nonumber
\Gamma= \{z\in\DD:    \log|z| - {\rm Re}(\overline{a}_1z) = \log|a_1| - |a_1|^2\}.
\end{equation} 
The curve divides the plane into the unbounded domain $\Omega_0$ and the bounded domain $\Omega_1$
such that $\CC= \Omega_0\cup \Omega_1\cup\Gamma$.  
The strong asymptotics of the polynomial $p_n$ is given by
\begin{equation}
p_n(z)=\begin{cases}
\displaystyle\frac{z^{n+c_1}}{(z-a_1)^{c_1}}\left(1+{\cal O}\left(\frac{1}{N^{\infty}}\right)\right),& z\in \Omega_0,\vspace{0.1cm}\\
\displaystyle -\frac{a_1(1-|a_1|^2)^{c_1-1}}{N^{1-c_1}\Gamma(c_1)}\frac{\ee^{N(\overline{a}_1z+\log a_1-|a_1|^2)}}{z-a_1}\left(1+\mathcal
{O}\left(\frac{1}{N}\right)\right),& z\in\Omega_{1},
\end{cases}
\end{equation}
where ${\cal O}(1/N^\infty)$ stands for ${\cal O}(1/N^m)$ for an arbitrary $m>0$. In $\Omega_0$ the branch is chosen such that $z^{n+c_1}/(z-a_1)^{c_1}\sim z^n$ as $|z|\to \infty$ with the branch cut $[0,a_1]$.

\begin{figure}\label{fig single}
\begin{center}
\includegraphics[width=0.49\textwidth]{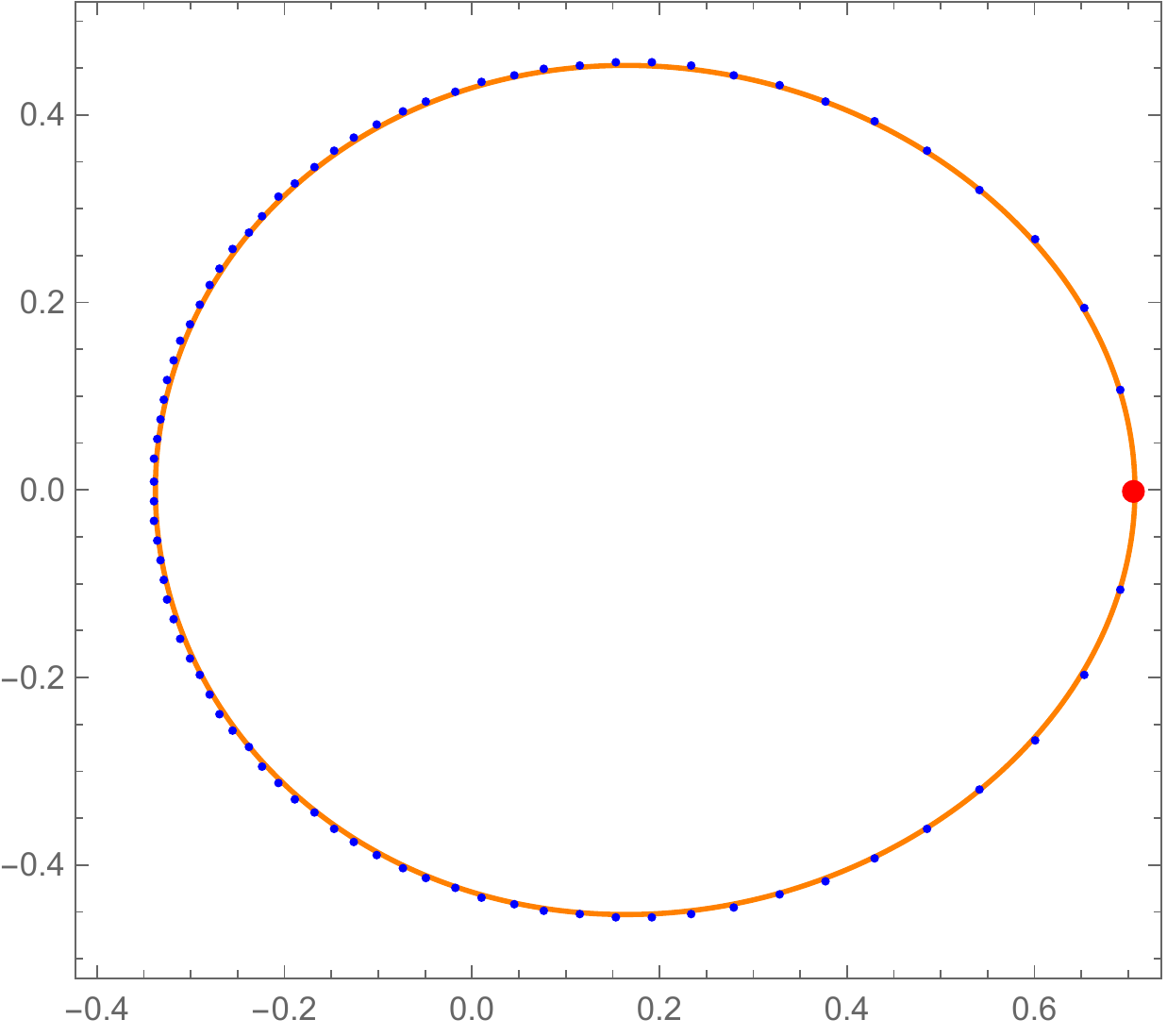}
\end{center}
\caption{The zeros (blue dots) of orthogonal polynomials for $\nu=1$ with $n=80$,
$c_1=1$ and $a_1=1/\sqrt{2}$. Zeros are close to
the generalized Szeg\"o curve (orange).} 
\end{figure}

When $z$ is near $\Gamma$ but away from $a_1$ the strong asymptotics is given by the sum of the two asymptotic expressions given above, hence zeros of $p_n(z)$ line up along $\Gamma$ with the inter-distance of order ${\cal O}(1/N)$.

When $z$ is near $a_1$ we define the local zooming coordinate 
\begin{equation}\nonumber
   \zeta(z)= -N(\overline{a}_1z-\log z+\log a_1-|a_1|^2),
\end{equation}
which maps $[0,a_1]$ to the negative real axis. We have
\begin{equation}\label{eq7}
p_n(z)=\frac{z^{n+c_1}}{(z-a_1)^{c_1}}\frac{\zeta(z)^{c_1}}{\ee^{\zeta(z)}}\left(\frac{\ee^{\zeta(z)}}{\zeta(z)^{c_1}}-f_{c_1}(\zeta(z))+{\cal O}\left(\frac{1}{N}\right)\right),  
\end{equation}
where the multivalued function $\zeta^c$ is defined with the principal branch and $f_{c}(\zeta)$ is defined by the two conditions: $f_{c}(\zeta)\to 0$ as $|\zeta|\to 0$, and $\ee^\zeta/\zeta^c-f_{c}(\zeta)$ is entire. See Appendix \ref{appendix} for more details about $f_c(\zeta)$.
Zeros of the above entire function are shown in Figure \ref{localzeros}.

\begin{figure}
\begin{center}
\includegraphics[width=0.32\textwidth]{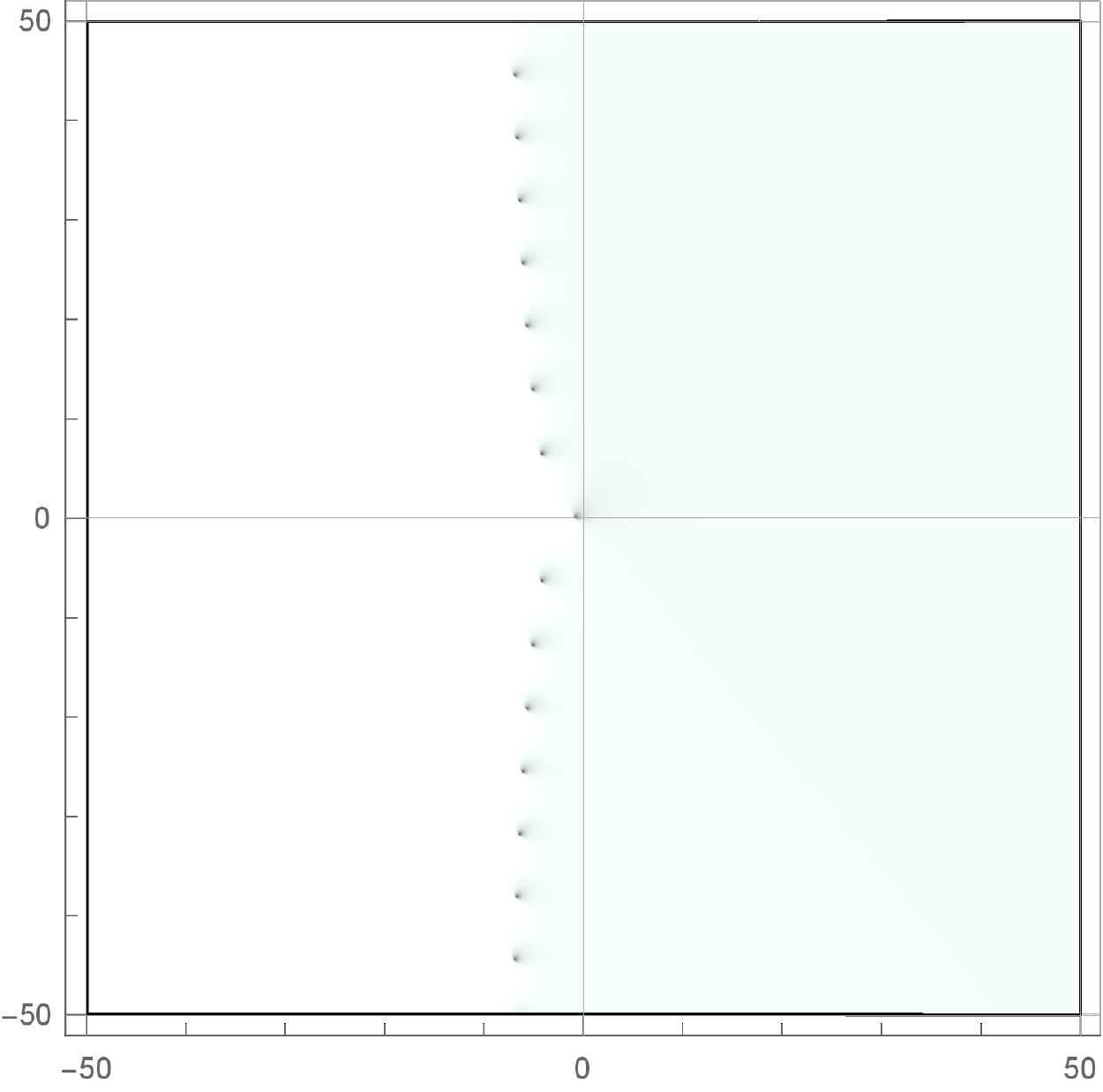}
\includegraphics[width=0.32\textwidth]{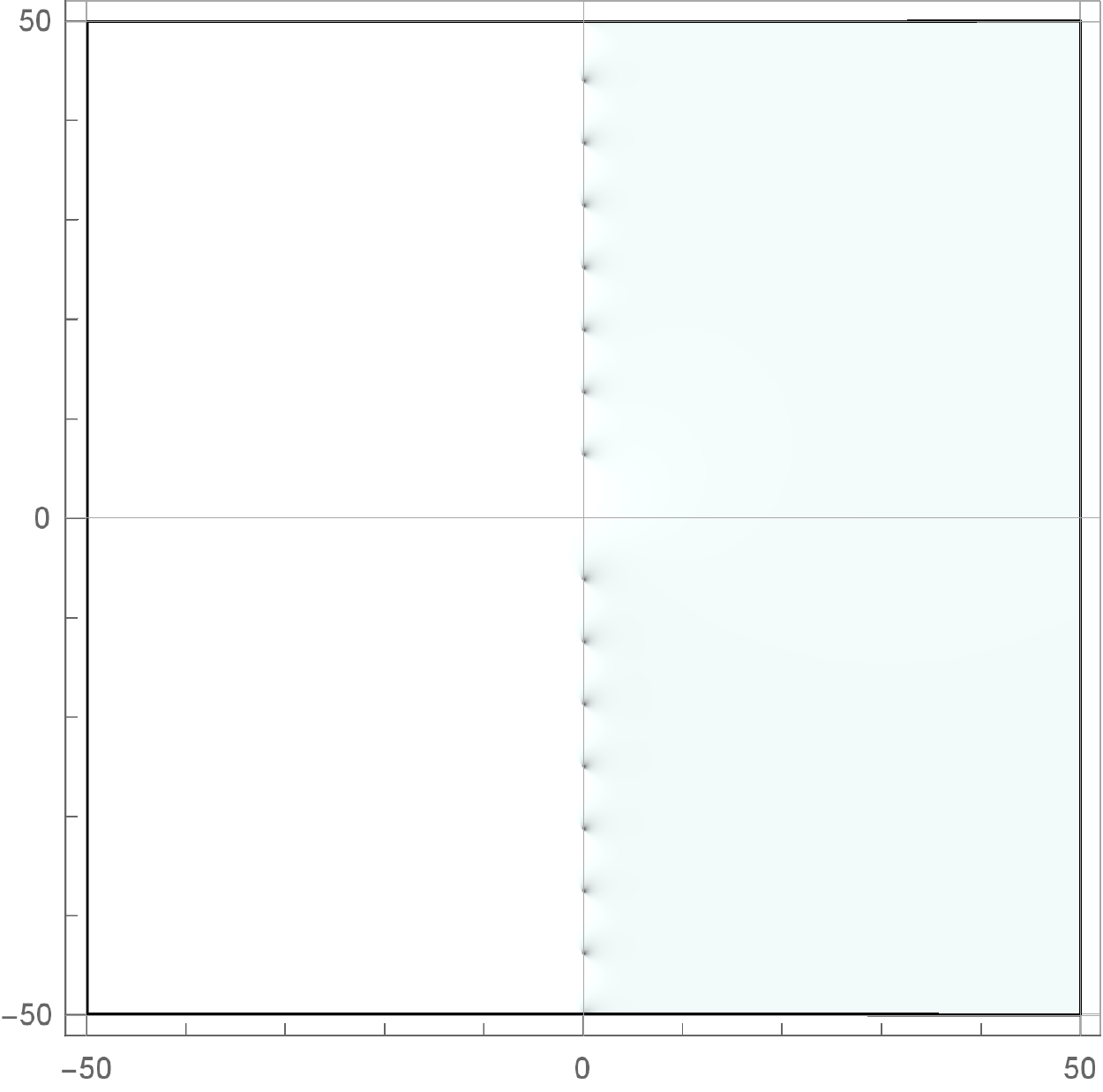}
\includegraphics[width=0.32\textwidth]{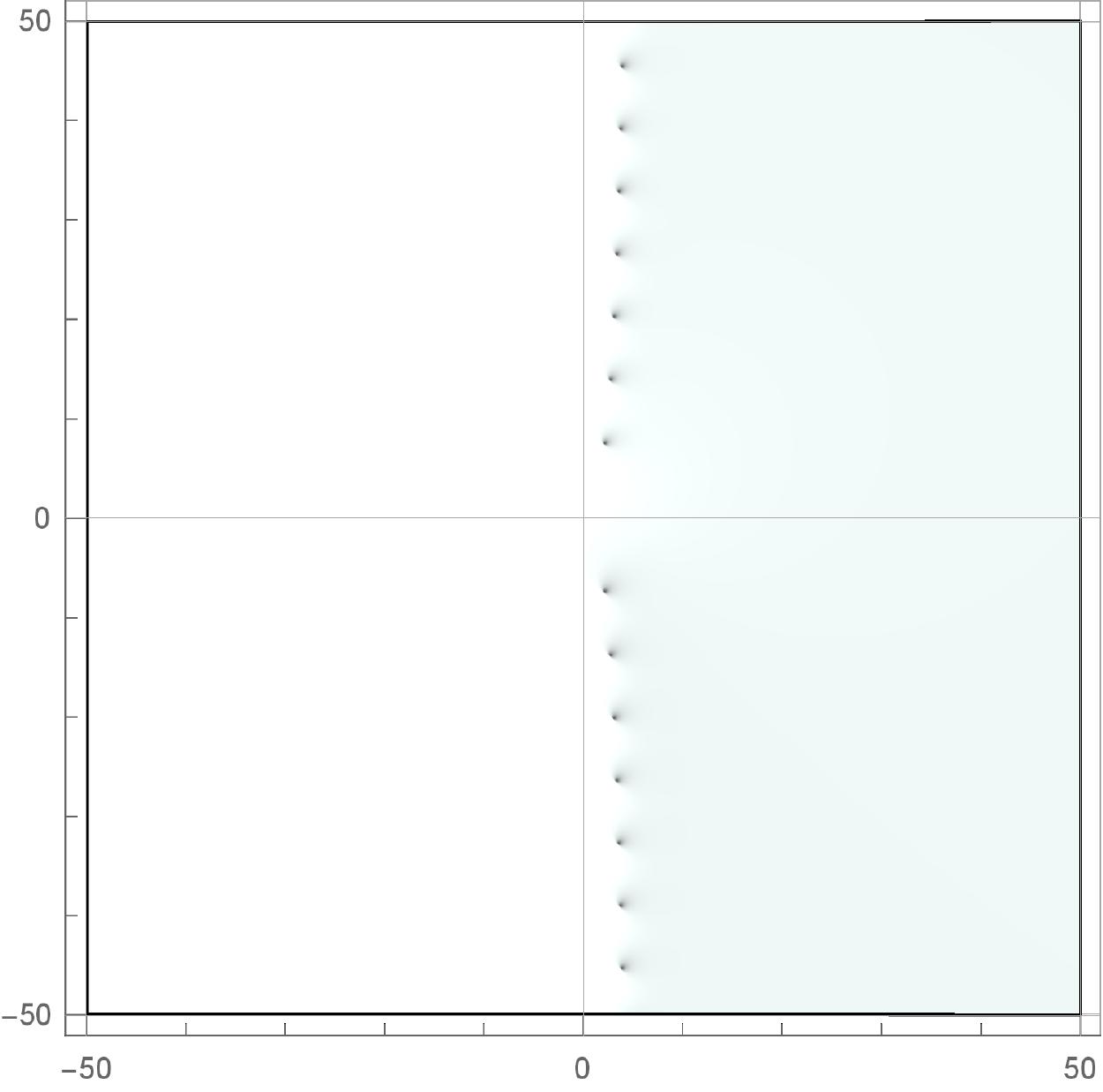}
\end{center}
 \caption{Zeros (dots) of $\ee^{\zeta}/\zeta^c-f_{c}(\zeta)$ for $c=-0.5,1,2$ (from left).
} \label{localzeros} 
\end{figure}


\bigskip
\noindent{{\bf Multiple Szeg\H{o} curve:}}
The goal of this paper is to generalize these results to the case of $\nu>1$.   
We obtain that the roots of the polynomial converge towards what we call the {\it multiple Szeg\H{o} curve}, a certain merger of $\nu$ number of the generalized Szeg\H{o} curves.  The multiple Szeg\H o curve, that we will denote by $\Gamma$, is determined in terms of $\{a_1,\dots,a_\nu\}\subset{\mathbb D}$ and it divides the plane into $\nu+1$ domains, the unbounded domain $\Omega_0$, and the $\nu$ number of bounded domains: $\Omega_1,\dots,\Omega_\nu$ such that ${\mathbb C}=\Gamma\cup\Omega_0\cup\Omega_1\cup\dots\cup\Omega_\nu$.  See Figure \ref{szego0} for an example when $\nu=4$.

\begin{figure}
\begin{center}
\includegraphics[width=0.5\textwidth]{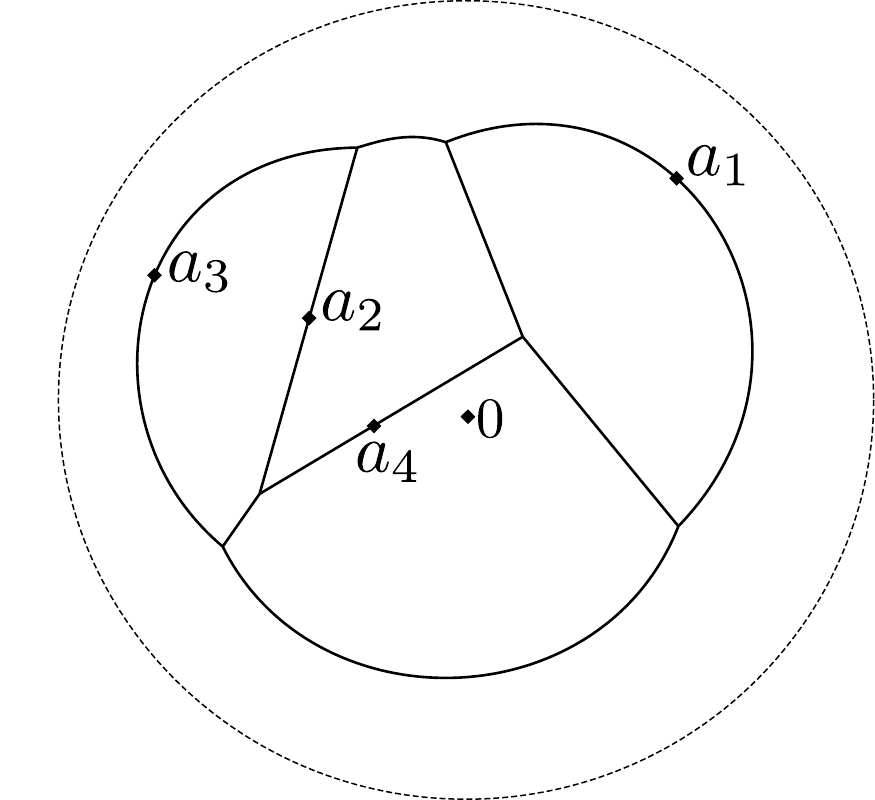}
\end{center}
 \caption{The multiple Szeg\H{o} curve with $\nu=4$.  Dotted line is the unit circle.  The plane is divided into five domains by the curve; the bounded region adjacent to $a_1$ is $\Omega_1$, the bounded region adjacent to $a_3$ is $\Omega_3$, the region containing the origin is $\Omega_4$ and the last remaining bounded region is $\Omega_2$.
} \label{szego0}
\end{figure}

To define the multiple Szeg\H o curve, let $L=(l_1,\dots,l_\nu)$ be a set of real numbers. We define a continuous function $\Phi^L(z):\DD\to\RR$ by
\begin{equation}\label{eq:PhiL} \Phi^L(z) = \max \{\log|z|, {\rm Re}(\overline a_1 z)+l_1,\dots,{\rm Re}(\overline a_\nu z)+l_\nu\}.
\end{equation}
We define the bounded domains that depend on $L$ by
\begin{equation}\label{def omegaj}
\Omega_j = {\rm Int}\{z\in \DD\,|\,\Phi^L(z)={\rm Re}(\overline a_j z)+l_j\},\quad j=1,\dots,\nu,
\end{equation}
where ${\rm Int}A$ stands for the interior of $A$, the largest open subset of $A$.
We also define the unbounded domain $\Omega_0$ by
\begin{equation}\label{def omega0}
    \Omega_0 = {\mathbb D}^c\cup{\rm Int}\{ z\in {\mathbb D}| \Phi^L(z)=\log|z|\}.
\end{equation}
Note that $\Omega_j$ for $j\neq 0$ can be empty in some cases.

\begin{thm}\label{Thm1}
For a given $\{a_1,\dots,a_\nu\}$ there exists the unique set of real numbers, $L=(l_1,\dots,l_\nu)$, such that 
\begin{equation}\nonumber
    a_j\in\partial \Omega_j \text{~~ for $j=1,\dots,\nu$}.
\end{equation}
\end{thm}

Given $\{a\text{'s}\}$ the above theorem uniquely determines $L$ and, in turn, $\Phi^L$ and $\Omega_j$'s.  It allows us to define the following.  
\begin{defn}
For a given $\{a_j\text{'s}\}$ we define {\em the multiple Szeg\H o curve} $\Gamma$ by
\begin{equation}\label{def-gamma}
    \Gamma=\bigcup_{j=1}^\nu\partial\Omega_j,
\end{equation}
where $\Omega_j$'s are defined by \eqref{def omegaj} in terms of the unique $L$ that is given by Theorem \ref{Thm1}.
\end{defn}

Theorem \ref{Thm1} says that  $a_j\in\partial\Omega_j$.  It means that $a_j$ is adjacent to another domain $\Omega_k$ for some $k\neq j$. In such case we define the following notation:
\begin{equation}\label{def-arrow}
    j\to k   \quad\Longleftrightarrow\quad a_j\in\partial\Omega_k \mbox{ and $j\neq k$}.
\end{equation}

\begin{defn}\label{def chain of aj}
Let {\em the chain of $a_j$} be the ordered subset $(k_s,k_{s-1},\dots,k_1)\subset\{1,\dots,\nu\}$ such that $k_s=j$ and
\begin{equation}\label{chain13}
k_s\to k_{s-1}\to\dots\to k_1\to 0.
\end{equation}
\end{defn}

\begin{remark}\label{remark1}
In this paper we consider only generic cases when the multiple Szeg\H o curve is smooth at every $a_j$.  This means that the point $a_j$ is on the boundary of exactly two domains $\Omega_j$ and $\Omega_k$ for $k\neq j$, which implies that the chain of $a_j$ is unique.    It is possible that $a_j$ belongs to the boundary of three or more domains. See Figure \ref{nongeneric case} (the left picture). Though we omit such cases for brevity our method still applies to such non-generic cases and it modifies only Theorem \ref{thm02}.  It is also possible that $\Omega_j$ is empty for some $j$. See Figure \ref{nongeneric case} (the middle and the right pictures). In this case we expect new type of local behavior showing up near $a_j$. 
\end{remark}

\begin{figure}
\begin{center}
\includegraphics[width=0.32\textwidth]{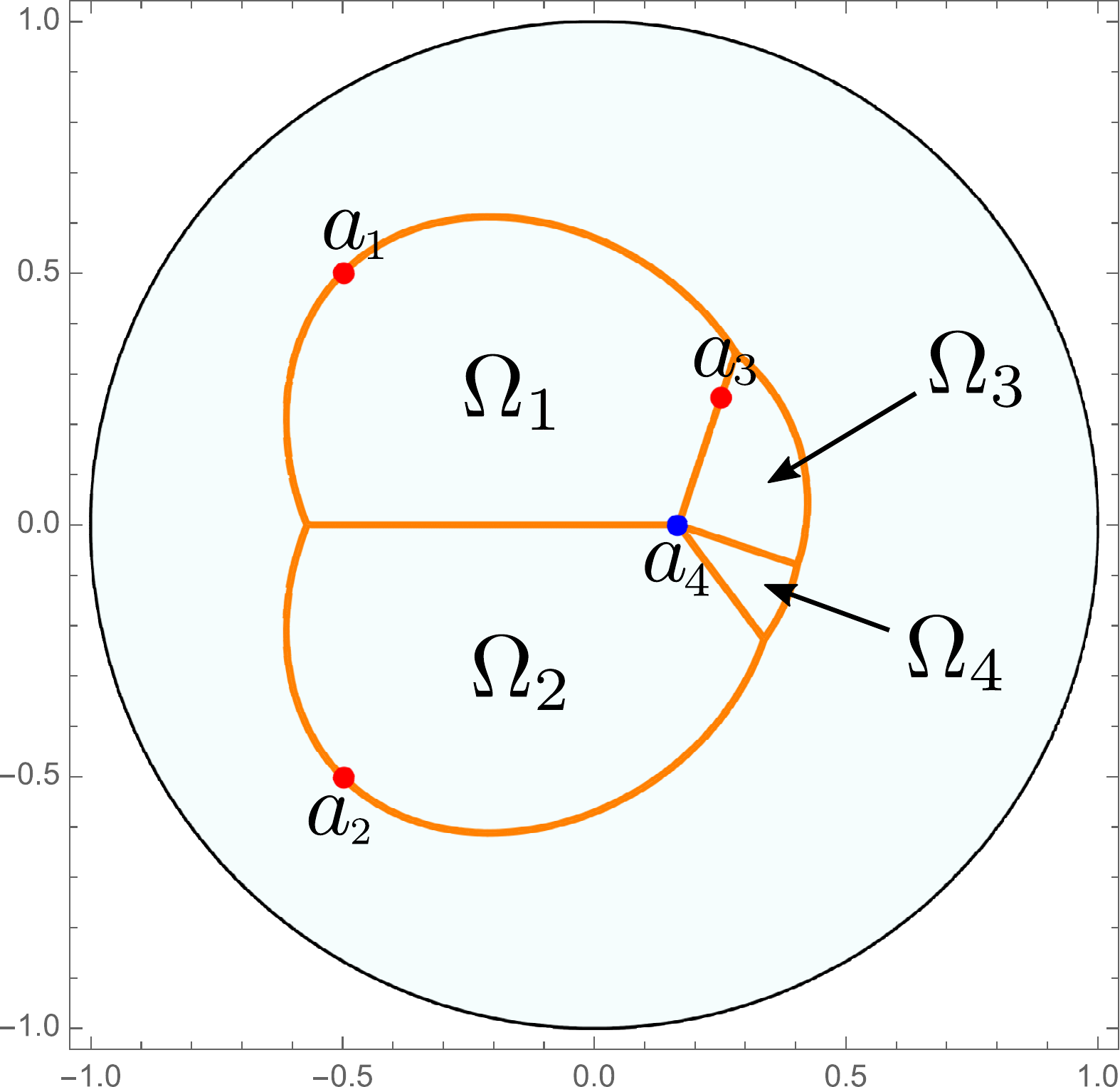}
\includegraphics[width=0.32\textwidth]{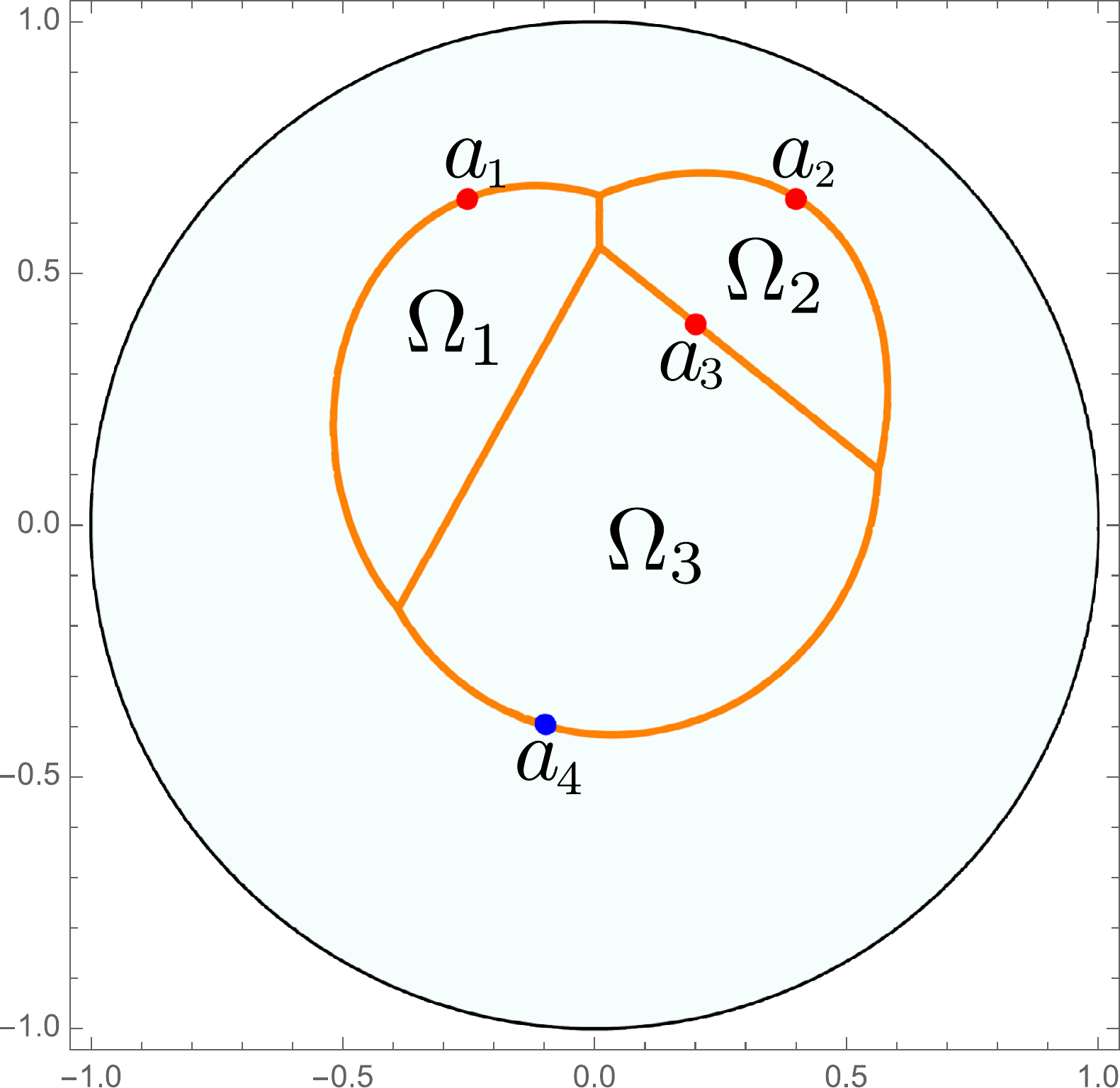}
\includegraphics[width=0.32\textwidth]{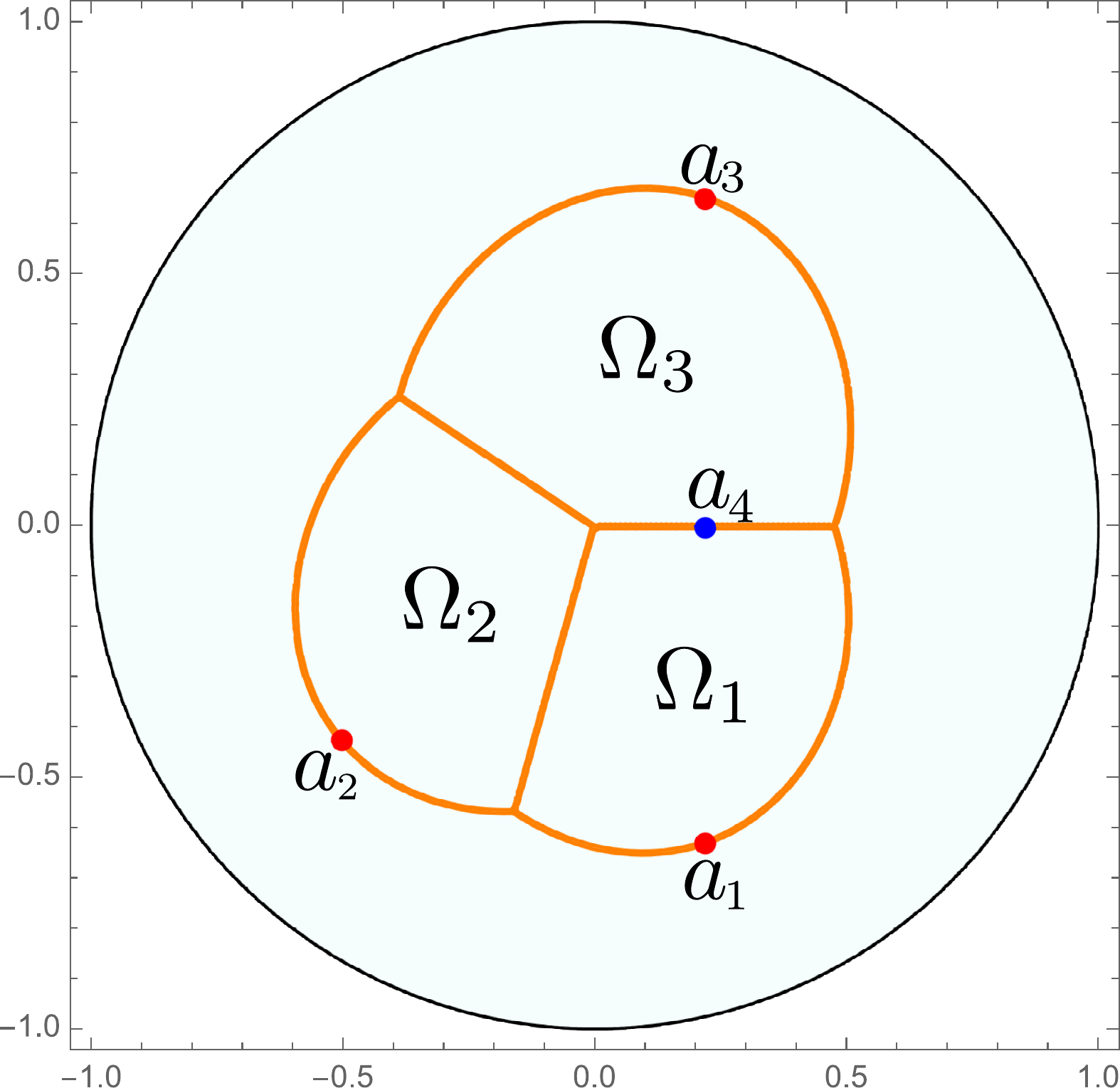}
\end{center}
 \caption{The nongeneric multiple Szeg\H{o} curves with $\nu=4$. In the left picture the point $a_4$ (blue) is at the boundary of more than two (four) regions. In the middle and the right ones the region $\Omega_4$ is empty. } \label{nongeneric case}
\end{figure}

When $j\to k$, the continuity of $\Phi^L(z)$ at $a_j\in \partial\Omega_j\cap \partial\Omega_k$, $$\lim_{\substack{z\to a_j\\z\in\Omega_j}}\Phi^L(z)=\lim_{\substack{z\to a_j\\z\in\Omega_k}}\Phi^L(z),$$ gives
\begin{equation}\label{l relation}
|a_j|^2+ l_j=\begin{cases}
     {\rm Re}(\overline a_ka_j) + l_k &\mbox{ if $k\neq 0$},
\\
\log |a_j|  &\mbox{ if $k = 0$}.
\end{cases}
\end{equation}
When $j\to k$ we define the complex numbers $\ell_j$'s such that
\begin{equation}\label{ell relation}
|a_j|^2+ \ell_j=\begin{cases}
     \overline a_ka_j + \ell_k &\mbox{ if $k\neq 0$},
\\
\log a_j  &\mbox{ if $k = 0$}.
\end{cases}
\end{equation}
These relations uniquely determine all the $\ell_j$'s inductively for a given chain; for example, the chain in \eqref{chain13} gives $\ell_{k_1}=\log a_{k_1}-|a_{k_1}|^2$ and $\ell_{k_2}= -|a_{k_2}|^2 +   \overline a_{k_1}a_{k_2} + \ell_{k_1}$ and so on.  Solving the relations inductively for the chain in \eqref{chain13} we get
\begin{equation}\nonumber
    \ell_j=\log a_{k_1}-|a_{k_1}|^2 + \sum_{i=2}^{s}\left(\overline a_{k_{i-1}}a_{k_i}-|a_{k_i}|^2   \right).
\end{equation}
By \eqref{l relation} and \eqref{ell relation}, one can also observe that ${\rm Re}\ell_j=l_j.$

\bigskip
\noindent{\bf Branch cuts for non-integer $c_j$'s:}
Whenever there is a non-integer exponent $c_j$ one must be careful about the branch of the multivalued function.  For example, an expression like $z^{c_1}$ has infinitely many branches when $c_1$ is irrational.  The precise definition of the branches of the multivalued functions are needed to state the main results.

First we define various branch cuts; See Figure \ref{fig-branch}.
\begin{align}\label{def bhat}
   \widehat{\bf B} &=\bigcup_{j=1}^\nu\widehat{\bf B}_j \mbox{ where } \widehat{\bf B}_j=\{a_j t:0\leq t\leq 1\},
   \\\label{def B}
   {\bf B}&=\bigcup_{j=1}^\nu {\bf B}_j \mbox{ where }  {\bf B}_j=\{a_j t: t\geq 1\},
   \\\label{def bk}
   {\bf B}[k]&=\bigg(\bigcup_{\substack{j=1\\j\neq k}}^\nu {\bf B}_{jk}\bigg)\cup {\bf B}_k \mbox{ where }  {\bf B}_{jk}=\{a_j+(a_j-a_k) t: t\geq 0\} .
\end{align}
In all these branch cuts, we define the  orientations of the branch cuts by the directions of increasing $t$.

\begin{figure} 
\centering
\begin{tikzpicture}[scale=0.6]

\draw[->] (2,1) -- (12,6);
\draw[->] (-2,4) -- (-6,12);
\draw[->] (5,7) -- (7.798,10.9174);
\draw[->] (1,9) -- (1.48159,13.3343);

\draw[red, postaction={decorate, decoration={markings, mark = at position 0.5 with {\arrow{>}}}} ]
				(0,0) -- (-2,4);
\draw[red, postaction={decorate, decoration={markings, mark = at position 0.5 with {\arrow{>}}}} ]
				(0,0) -- (2,1);
\draw[red, postaction={decorate, decoration={markings, mark = at position 0.5 with {\arrow{>}}}} ]
				(0,0) -- (5,7);
\draw[red, postaction={decorate, decoration={markings, mark = at position 0.5 with {\arrow{>}}}} ]
				(0,0) -- (1,9);			

\draw[blue,->]  (-2,4) -- (-10,10);
\draw[blue,->]  (5,7) -- (7.2,11.4);
\draw[blue,->]  (1,9) -- (0.45,13.4);

\draw[blue,dotted]  (2,1) -- (5,7);
\draw[blue,dotted]  (2,1) -- (1,9);
\draw[blue,dotted]  (2,1) -- (-2,4);

\foreach \Point/\PointLabel in {(0,0)/0, (2,1)/a_1, (5,7)/a_2, (1,9)/a_3,(-2,4)/ }
\draw[fill=black] \Point circle (0.07) node[below right] {$\PointLabel$};

\foreach \Point/\PointLabel in {(-1,4.8)/a_4, (12.5,5.5)/{\bf B}_{1}, (9,10.5)/{\bf B}_{2},(2.8,13)/{\bf B}_{3}, (-4.3,12)/{\bf B}_{4}, (6.7,11.5)/{\bf B}_{21}, (0.5,13)/{\bf B}_{31}, (-9,9)/{\bf B}_{41}, (2,0.5)/{\widehat{\bf B}}_{1}, (3,5.5)/{\widehat{\bf B}}_{2}, (0.5,6)/{\widehat{\bf B}}_{3}, (-1.2,2.7)/{\widehat{\bf B}}_{4} }
\draw[fill=black]  \Point
 node[below left] {$\PointLabel$};

 \end{tikzpicture}
 \caption{Various branch cuts for $\nu=4$. ${\bf B}=\{{\bf B}_1,{\bf B}_2,{\bf B}_3,{\bf B}_4\}$ (black rays), $\widehat{\bf B}=\{\widehat{\bf B}_1,\widehat{\bf B}_2,\widehat{\bf B}_3,\widehat{\bf B}_4\}$ (red rays), and  ${\bf B}[1]={\bf B}_1\cup\{{\bf B}_{21},{\bf B}_{31},{\bf B}_{41}\}$ (blue rays). The branch cuts of $W(z)$ are ${\bf B}$. The branch cuts of $W_1(z)$ are ${\bf B}[1]$. }\label{fig-branch}
 \end{figure}
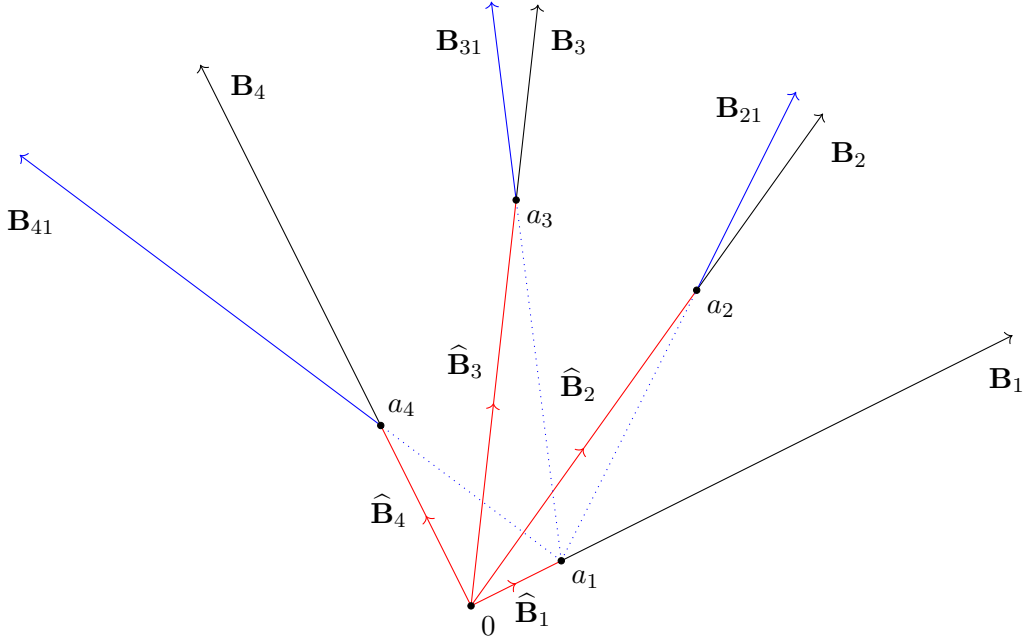

\noindent For the sake of presentation we will assume that 
no three points from  $\{0,a_1,\dots,a_\nu\}$ are collinear, which implies that the branch cuts do not overlap with each other. Such assumption can be disposed of with a perturbation argument.

We now define the exact branches of the multivalued functions that appear in this paper. 
\begin{itemize}
  \setlength\itemsep{-0.2em}
    \item[(i)] $(z-a_j)^{c_j}$ is analytic away from ${\bf B}_j$. One may choose {\em any branch} for this function but one should stick to the choice throughout the paper.
    \item[(ii)] For $j\neq k$ we define $\big[(z-a_j)^{c_j}\big]_{{\bf B}[k]}$ to be analytic away from ${\bf B}_{jk}$ and \begin{equation}\label{def subbk}
        \big[(z-a_j)^{c_j}\big]_{{\bf B}[k]}=(z-a_j)^{c_j}\mbox{ when }z\in {\bf B}_k\cup\widehat{\bf B}_k \mbox{ and $j\neq k$}.
    \end{equation}   We also define $\big[(z-a_j)^{c_j}\big]_{{\bf B}[j]}=(z-a_j)^{c_j}$.
    \item[(iii)] We define 
    \begin{equation}\label{def w}
        W(z)=\prod_{j=1}^\nu (z-a_j)^{c_j}.
    \end{equation}
    \item[(iv)]  We define 
    \begin{equation}\label{def wk}
        W_k(z)=\prod_{j=1}^\nu \big[(z-a_j)^{c_j}\big]_{{\bf B}[k]}.
    \end{equation}  It follows that $W_k(z) = W(z)$ when $z$ is in a neighborhood of ${\bf B}_k\cup\widehat{\bf B}_k$.  Note that $W_k(z)$ has the branch cut on ${\bf B}[k]$.
        \item[(v)] $z^{c_j}$ is analytic away from ${\bf B}_j\cup \widehat{\bf B}_j$. We select the branch such that $(z-a_j)^{c_j}/z^{c_j}\to 1$ as $z$ goes to $\infty$ along ${\bf B}_j$. 
    \item[(vi)] $\big[z^{c_j}\big]_{{\bf B}[k]}$ has the branch cut on ${\bf B}_{jk}\cup \widehat{\bf B}_j$. We select the branch such that \begin{equation}\label{eq24}
            \big[z^{c_j}\big]_{{\bf B}[k]}=z^{c_j}\mbox{ when }z\in{\bf B}_k\cup\widehat {\bf B}_k.
        \end{equation}
    \item[(vii)] We define $z^{\sum c}=\prod_{j=1}^\nu z^{c_j}$. We use the shortened notation $\sum c=\sum_{j=1}^\nu c_j.$
\item[(viii)] We define $\big[z^{\sum c}\big]_{{\bf B}[k]}=\prod_{j=1}^\nu \big[z^{c_j}\big]_{{\bf B}[k]}$.
\end{itemize}

When $c_j$'s are all integer-valued there is no ambiguity in the choice of branches. We have $z^{c_j}=\big[z^{c_j}\big]_{{\bf B}[k]}$ and $(z-a_j)^{c_j}=\big[(z-a_j)^{c_j}\big]_{{\bf B}[k]}$ and therefore $W_k(z)=W(z)$. Also the final results should be independent of the choice of the branch made in (i).

 \begin{defn}
 Let us define the phase factor $\widetilde{\eta}_{kj}$ by
\begin{align}\label{def etakjtilde}
 \widetilde{\eta}_{kj}:&=\frac{[(z-a_j)^{c_j}]_{{\bf B}[k]}}{(z-a_j)^{c_j}}\frac{W_j(z)}{W_k(z)},\quad z\in{\bf B}_{jk}.
\end{align}
Let $(k_s,k_{s-1},\dots,k_1)$ be the chain of $a_j$. Then we define the constant ${\sf chain}(j)$  by
\begin{equation}\label{chain15}
 {\sf chain}(j)= \frac{a_{k_1}^{1+\sum_{i\neq k_1} c_i}N^{\sum_{i=1}^s(c_{k_i}-1)}}{\Gamma(c_{k_1})(1-|a_{k_1}|^2)^{1-c_{k_1}}}\prod_{i=1}^{s-1} \frac{\widetilde{\eta}_{k_i,k_{i+1}}(a_{k_{i+1}}-a_{k_i})^{c_{k_i}}|a_{k_i}-a_{k_{i+1}}|^{2(c_{k_{i+1}}-1)}}{\Gamma(c_{k_{i+1}})
 (a_{k_i}-a_{k_{i+1}})^{c_{k_{i+1}}} }. 
\end{equation}
 \end{defn}
Above $z^{n+\sum c} = z^n\cdot z^{\sum c}$ and $(a_{k_i}-a_{k_{i+1}})^{c_{k_{i+1}}}$ that appears in ${\sf chain}[j]$ is the evaluation of $(z-a_{k_{i+1}})^{c_{k_{i+1}}}$ at $z=a_{k_i}$.

\bigskip
\noindent{\bf Strong asymptotics of $p_n$:} We now state the main results.  See Figure 4 for numerical support.

\begin{figure}
\begin{center}
\includegraphics[width=0.4\textwidth]{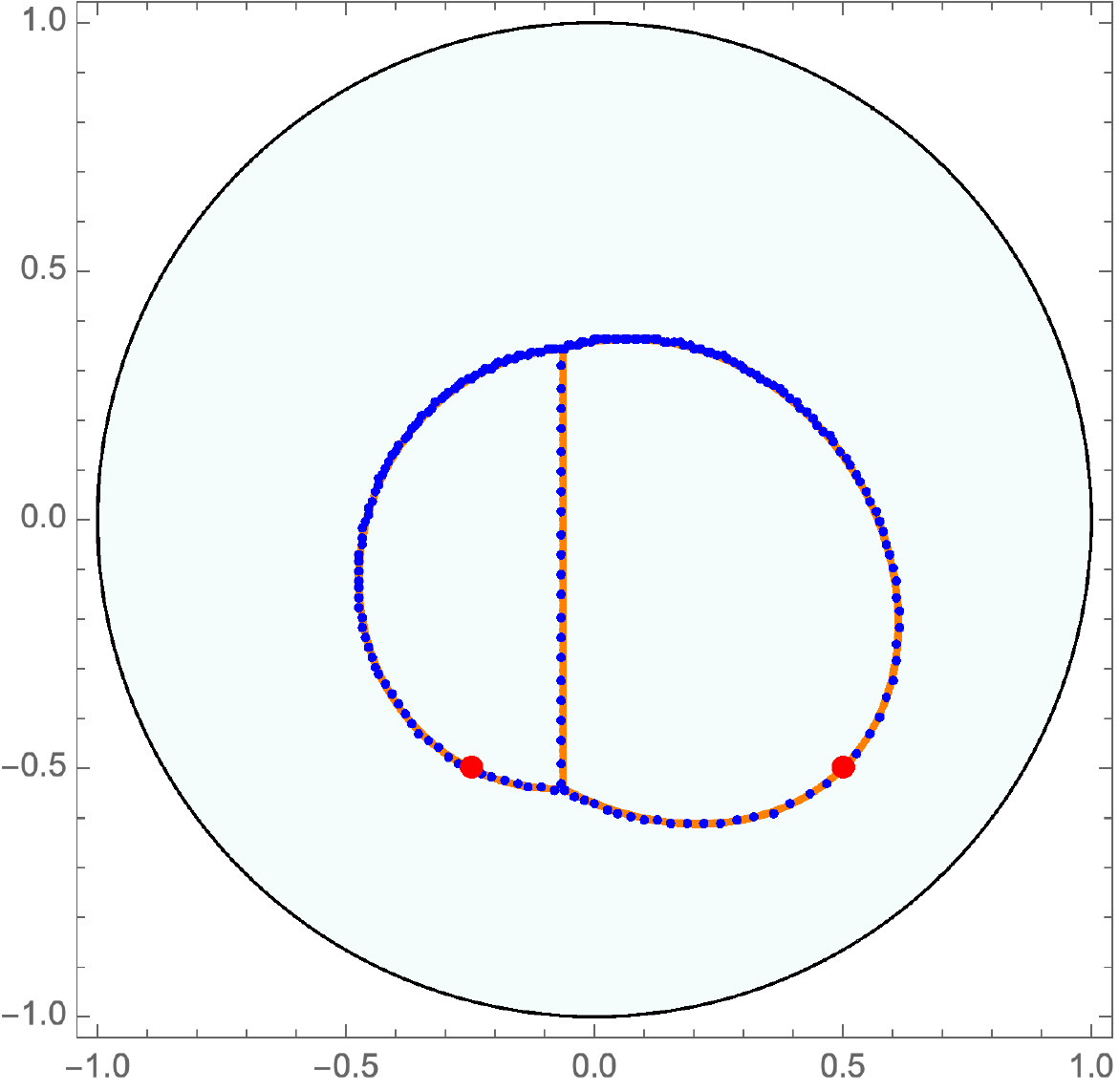}
\includegraphics[width=0.4\textwidth]{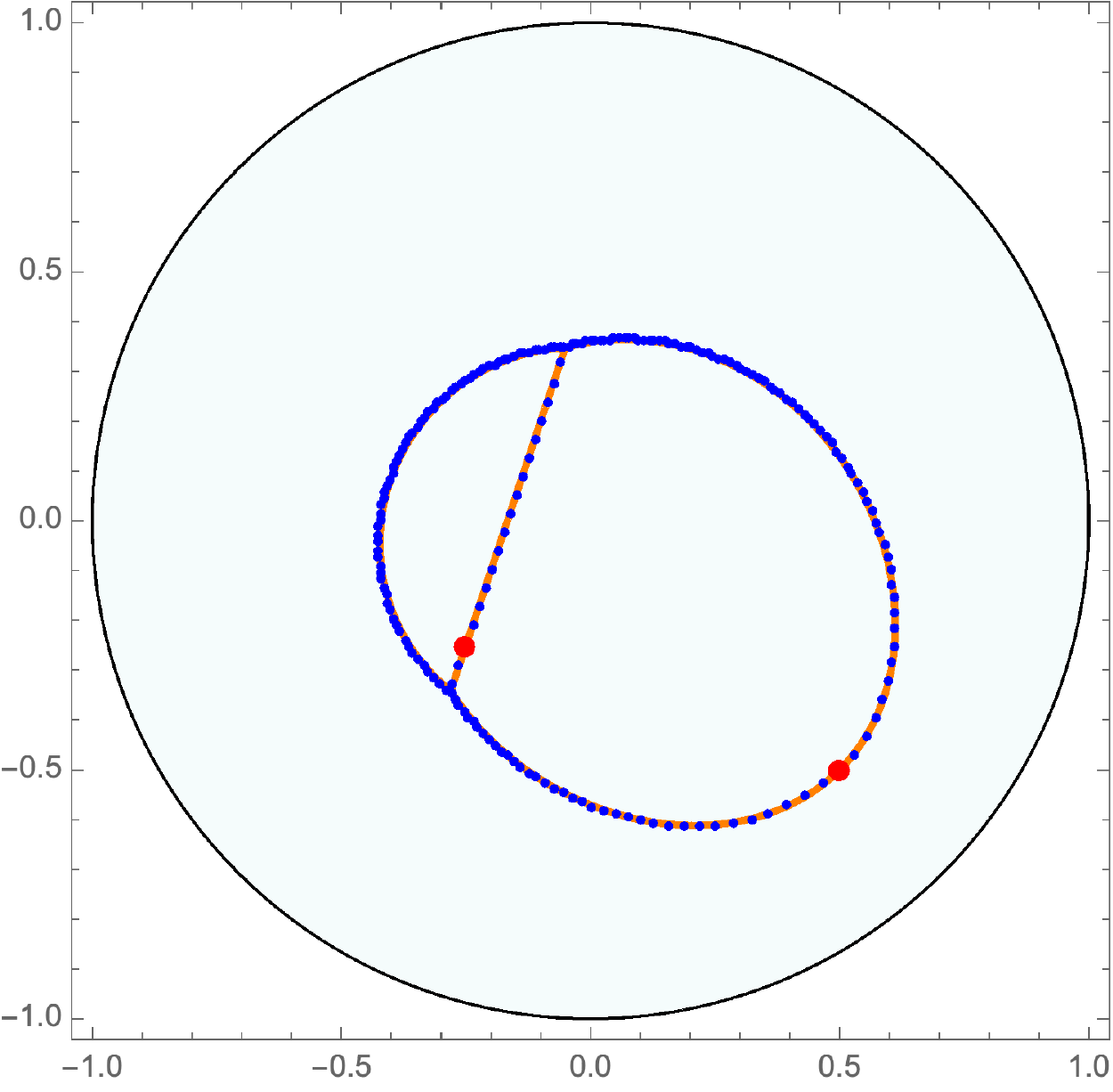}
\\
\includegraphics[width=0.4\textwidth]{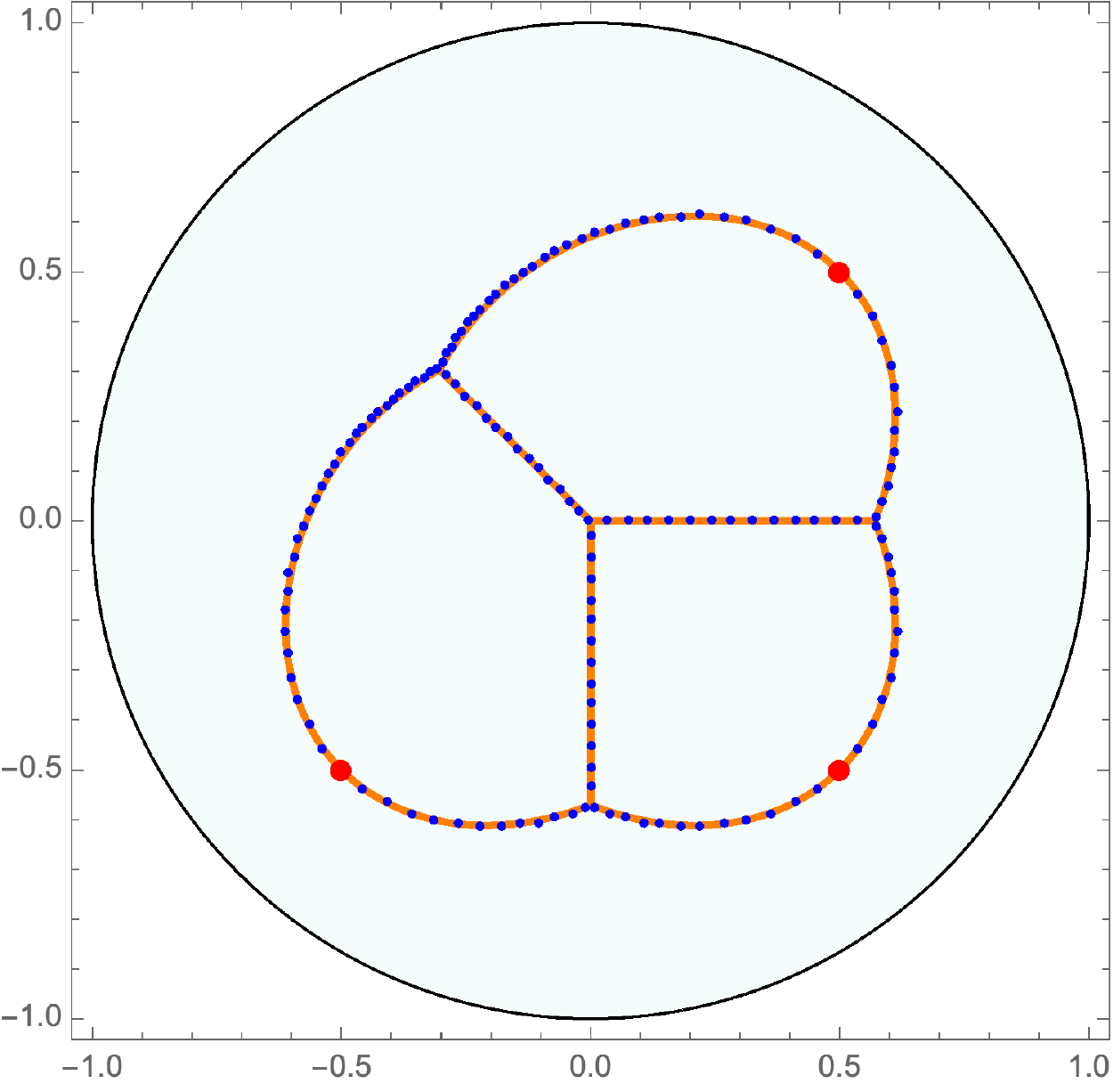}
\includegraphics[width=0.4\textwidth]{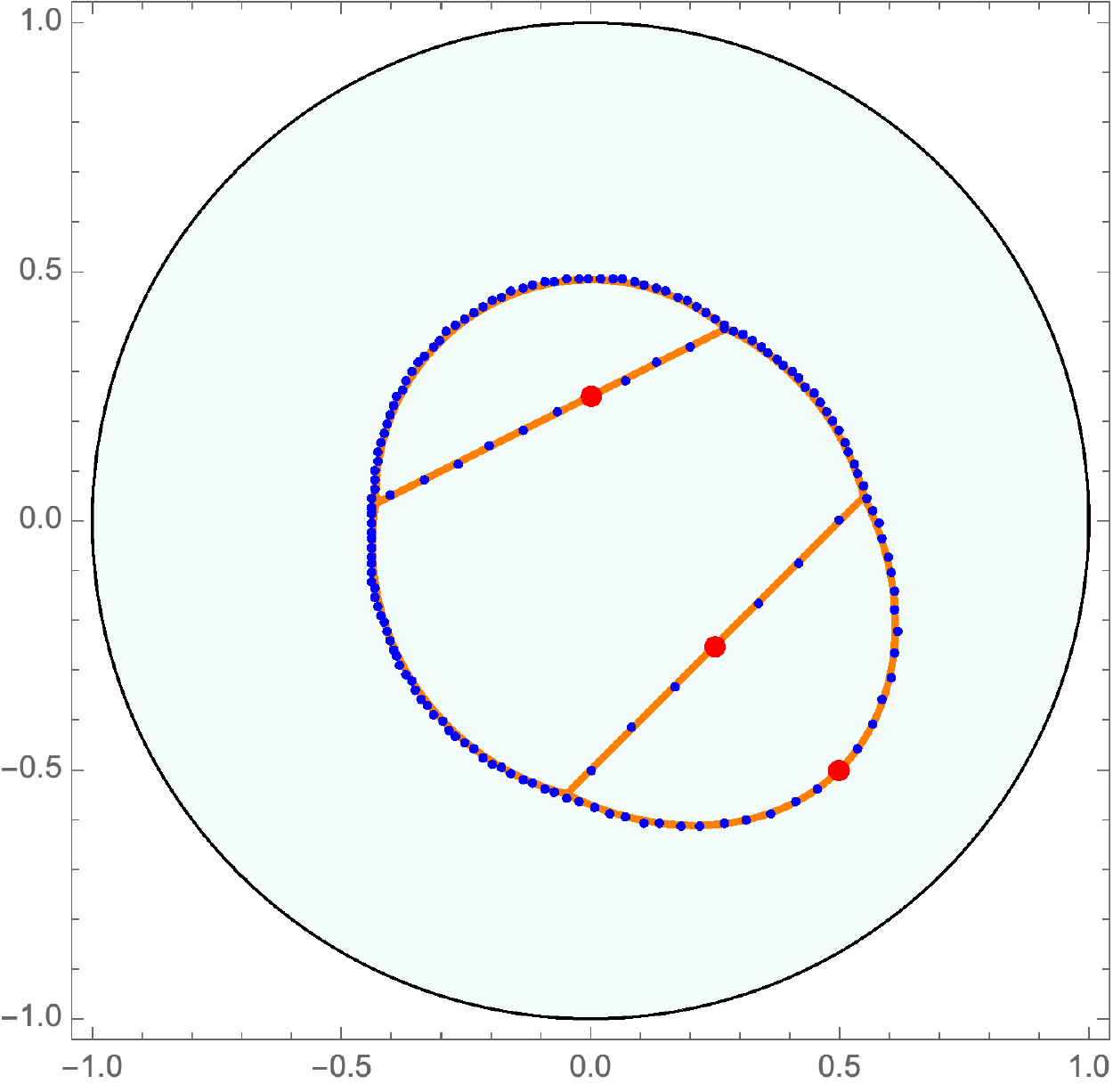}
\end{center}
 \caption{The roots of $p_n(z)$ (blue dots) with the multiple Szeg\H o curves (orange lines) for $n=200$ (first row) and $n=150$ (second row).  From the top left in clockwise order: $(a_1,a_2)=(0.5- 0.5\ii,-0.25- 0.5\ii)$, $(a_1,a_2)=(0.5-0.5\ii,-0.25-0.25\ii)$, $(a_1,a_2, a_3)=(0.5+0.5\ii,-0.5-0.5\ii,0.5-0.5\ii)$, $(a_1,a_2, a_3)=(0.5-0.5\ii,0.25-0.25\ii,0.25\ii)$.  For all cases, $c_j=1$ for all $j$.  With the current resolution of the pictures, the roots are observed right on top of the theoretical limiting curves.
} \label{szegozeros}
\end{figure}

\begin{thm}\label{thm01}
 Let $\{a_1,\dots,a_\nu\}\subset{\mathbb D}\setminus\{0\}$.
In a generic case (see Remark \ref{remark1} above), if $z$ is away from $\Gamma$,
as $N\to\infty$ such that $n/N= 1$,  the polynomial $p_n$ satisfies
\begin{equation}\label{result1}
p_n(z)=\begin{cases}
\displaystyle\frac{z^{n+\sum c}}{W(z)}\left(1+{\cal O}\left(\frac{1}{N^{\infty}}\right)\right),& z\in\Omega_{0},\vspace{0.1cm}\\
\displaystyle-\frac{\exp\big[N(\overline a_jz+\ell_j)\big](z-a_j)^{c_j}}{W_j(z)}\frac{{\sf chain}(j)}{z-a_j}\left(1+{\cal O}\left(\frac{1}{N}\right)\right),& z\in\Omega_{j}.
\end{cases}
\end{equation}
The error bounds are uniform over a compact subset in the corresponding regions.

When $z$ is near $\Gamma$ but away from $a_j$'s the strong asymptotics of $p_n$ is given by the sum of the two asymptotic expressions from the adjacent domains as below.  
\begin{equation}\label{result}
p_n(z)=\begin{cases}
\displaystyle \frac{z^{n+\sum c}}{W(z)}\left(1+{\cal O}\left(\frac{1}{N^{\infty}}\right)\right)\\ \displaystyle
\qquad\qquad -\frac{\exp\big[N(\overline a_jz+\ell_j)\big](z-a_j)^{c_j}}{W_j(z)}\frac{{\sf chain}(j)}{z-a_j}\left(1+{\cal O}\left(\frac{1}{N}\right)\right), \\
{\qquad\qquad\qquad\qquad\qquad\qquad\qquad\qquad\qquad\qquad\qquad\qquad\mbox{when $z$ near $\Gamma_{j0}$}},\vspace{0.1cm}
\\
\displaystyle
-\frac{\exp\big[N(\overline a_jz+\ell_j)\big](z-a_j)^{c_j}}{W_j(z)}\frac{{\sf chain}(j)}{z-a_j}\left(1+{\cal O}\left(\frac{1}{N}\right)\right)\vspace{0.1cm}
\\\displaystyle\qquad\qquad 
-\frac{\exp\big[N(\overline a_kz+\ell_k)\big](z-a_k)^{c_k}}{W_k(z)}\frac{{\sf chain}(k)}{z-a_k}\left(1+{\cal O}\left(\frac{1}{N}\right)\right),
\\{\qquad\qquad\qquad\qquad\qquad\qquad\qquad\qquad\qquad\qquad\qquad\qquad\mbox{when $z$ near $\Gamma_{jk}$}}.
\end{cases}
\end{equation}
The error bounds are uniform over a compact subset of $\Omega_0\cup \Omega_j\cup \Gamma_{j0}\setminus\{a\text{'s}\}$ for the former and $\Omega_k\cup \Omega_j\cup \Gamma_{jk}\setminus\{a\text{'s}\}$ for the latter.  The error bound ${\cal O}(1/N^\infty)$ stands for ${\cal O}(1/N^m)$ for all $m>0$. 
\end{thm}

In a neighborhood of $a_j\in \partial\Omega_k\cap\partial\Omega_j$, $k\neq j$, we define the local zooming coordinate $\zeta$ by
\begin{align}
\zeta(z)&=-N(\overline{a}_jz-\log z+\log a_j-|a_j|^2)&\text{when $k=0$},
\\
\zeta(z)&=-N((\overline{a}_j-\overline{a}_k)z-(\overline{a}_j-\overline{a}_k)a_j)&\text{when $k\neq 0$}.
\end{align}

\begin{thm}\label{thm02}
In a generic case (see Remark \ref{remark1} above), if $a_j\in \partial\Omega_j\cap\partial\Omega_k$, $k\neq j$, we have the following asysmptotic behavior when $N\to\infty$  such that $n/N= 1$
\begin{equation}
p_n(z)={\cal A}_k(z) \frac{\zeta^{c_j}}{\ee^{\zeta}}\left(\frac{\ee^{\zeta}}{\zeta^{c_j}}-f_{c_j}(\zeta)+{\cal O}\left(\frac{1}{N}\right)\right),\quad z\in D_{a_j},
\end{equation}
where $D_{a_j}$ is a sufficiently small disk centered at $a_j$ with a fixed radius and ${\cal A}_k(z)$ is the leading asymptotics in $\Omega_k$ as written in Theorem \ref{thm01}, i.e.,
\begin{align}
{\cal A}_k(z)=\begin{cases}
\displaystyle\frac{z^{n+\sum c}}{W(z)},& k=0,\\
\displaystyle-\frac{\exp\big[N(\overline a_kz+\ell_k)\big](z-a_k)^{c_k}}{W_k(z)}\frac{{\sf chain}(k)}{z-a_k},& k\neq 0.
\end{cases}
\end{align}
The error bound is uniform over $D_{a_j}$.
\end{thm}

\begin{remark}
We did not pursue the case when one of the $a_j$'s is at the origin.  We believe that such case yields the same results. 
\end{remark}

\noindent {\bf Plan of the paper:} In Section \ref{section msc}, we prove Theorem \ref{Thm1} which essentially states that the multiple Szeg\H o curve is uniquely given in terms of $\{a_j\}_{j=1}^\nu$.  The basic idea is to determine the constants $l_j$'s such that the maximal function $\Phi^L$ \eqref{eq:PhiL}, which is in fact the logarithmic potential of the limiting zeros, has the discontinuities exactly at each $a_j$'s.   In the proof, we present an algorithm to determine $l_j$'s in finite steps, which we also use to generate the figures of multiple Szeg\H o curves in this paper.

In Section \ref{section noninteger}, we construct the $\nu\times\nu$ matrix function $\Psi(z)$ that will be used in the subsequent Riemann-Hilbert analysis.  
This construction is needed mostly to handle the non-integer values of $c$'s. 
We need a systematic placement of all the branch cuts, especially ${\bf B}[k]$'s, such that the jumps along the branch cuts decay properly at the end. 
One can skip this section and simply put $\Psi(z)=\widetilde{\Psi}(z)$ \eqref{def psitilde and hat} if $c$'s are all integers.

In Section \ref{steepest} we apply the nonlinear steepest decent method \cite{DKMVZ 1999} on the corresponding Riemann-Hilbert problem of size $(\nu+1)\times(\nu+1)$ and perform successive transformations, $Y\to \widetilde{Y}$ \eqref{tildeY}, $\widetilde{Y}\to T$ \eqref{T} and $T\to S$ \eqref{def S}.  We define the global parametrix $\Phi$ \eqref{def global phi} that satisfies the approximate Riemann-Hilbert problem of $S$.

In Section \ref{section local} we construct the local parametrices near each $a_j$'s to match the global parametrix. We also construct a rational function ${\cal R}$ to improve the global parametrix $\Phi$ into ${\cal R}\Phi$ such that to match the local parametrices better. This construction, called ``partial Schlesinger transform'' has been introduced in \cite{Bertola 2008} and also used in \cite{Lee 2017}.

In Section \ref{section strong}, based on the global and the local parametrices in the previous sections, we define $S^\infty$ \eqref{Sfinal}, which is the strong asymptotics of $S$.  By applying the small norm theorem to $S^\infty S^{-1}$ we prove Theorem \ref{thm01} and Theorem \ref{thm02}. 

In Appendix \ref{appendix} we explain the special function, a certain truncation of the exponential function, that appears in the local parametrix.

\noindent {\bf Acknowledgement.}  We thank Tom Claeys, Alfredo Dea\~no, Arno Kuijlaars and Nick Simm for the discussions and their interests in this project. The second author is supported by the Fonds de la Recherche Scientifique-FNRS under EOS project O013018F.  

\section{Multiple \texorpdfstring{Szeg\H{o}}{szego} curve}\label{section msc}

In this section we define the multiple Szeg\H{o} curve that depends on the set of points:
\begin{equation}\nonumber
\{a_1,\dots,a_\nu\} \subset \DD, \quad\nu\geq 2.
\end{equation}
Let $\Lambda$ be an $\nu$-dimensional vector with real entries $$ \Lambda=(\lambda_1,\dots,\lambda_\nu),$$ 
and $\Phi^\Lambda$ be the continuous and piecewise smooth function given by
\begin{equation}\label{eq:Phi} \Phi^\Lambda(z) = \max \{\log|z|, {\rm Re}(\overline a_1 z)+\lambda_1,\dots,{\rm Re}(\overline a_\nu z)+\lambda_\nu\}. 
\end{equation}
Then we define the regions,
\begin{equation}\label{eq:K}
\begin{split}
K^\Lambda_j &= \{z\in \DD|\Phi^\Lambda(z)={\rm Re}(\overline a_j z)+\lambda_j\},\quad j=1,\dots,\nu,
\\
K^\Lambda_0 &=\{z\in \DD|\Phi^\Lambda(z)=\log|z|\}\cup \DD^c.
\end{split}
\end{equation}
One can visualize the function $\Phi^\Lambda$ as follows. The graphs of the $\nu+1$ functions inside the $\max$ function in \eqref{eq:Phi} give $\nu+1$ surfaces among which the $\nu$ of them are planes.  The aerial view of these mutually intersecting surfaces selects the maximal function in each corresponding region defined in \eqref{eq:K}.   As one increases or decreases $\lambda_j$ the corresponding surface moves up or down respectively, and the regions $K^\Lambda_j$ expands or shrinks as well.
\bigskip 

\begin{thm}\label{thm1}
There exists a {\em maximal vector} $L = (l_1,\dots,l_\nu)$ such that 
\begin{equation}\label{eq:aj} a_j \notin {\rm Int}(K^L_j) \mbox{ for all } j=1,\dots, \nu.
\end{equation}
By {\em maximal vector} we mean that, for any other vector $L'>L$, the property \eqref{eq:aj} does not hold for some $j$.   We say $L'>L$ if the the vector $L'-L$ is a nonzero vector without any negative entry.
\end{thm}

The property \eqref{eq:aj} defines a closed set in the parameter space of $\Lambda$, i.e., the set $${\cal S}=\{ \Lambda : a_j \notin {\rm Int}(K^\Lambda_j) \mbox{ for all } j=1,\dots, \nu \}\subset{\mathbb R}^\nu $$ is closed. The set ${\cal S}$ is also nonempty.
When $a_j\neq 0$ we can choose $\lambda_j$ to be sufficiently small such that $\log|a_j|>|a_j|^2+\lambda_j$ which leads to $a_j\notin K^\Lambda_j$. Then $\lambda_k$ for $k\neq j$ can be chosen sufficiently small such that ${\rm Re}(\overline{a}_j a_k)+\lambda_j>|a_k|^2+\lambda_k$ hence $a_k\notin K^\Lambda_k$.  In this way, we can find $\Lambda$ such that $a_j \notin K^L_j$ for all $j$'s.

\begin{lemma}\label{lemma1}
If $\Lambda\in{\cal S}$ there exists some $j$ such that $\Phi^\Lambda(a_j)=\log|a_j|$.  
\end{lemma}

The lemma gives the upper bound on ${\cal S}$ since, if $\Lambda\in{\cal S}$, ${\rm Re}(\overline a_k a_j)+\lambda_k\leq \Phi^\Lambda(a_j)=\log|a_j|$ and, therefore,
$\lambda_k\leq \log|a_j|-{\rm Re}(\overline a_k a_j)\leq \log|a_j|+|a_j|$. 

To prove the {\em existence of a maximal vector} let $l_1=\sup_{\cal S}\lambda_1=\max_{\cal S}\lambda_1$ where we use that ${\cal S}$ is closed.  Inductively we define $$l_j = \sup\{\lambda_j|\Lambda\in{\cal S}, \lambda_1=l_1,\dots,\lambda_{j-1}=l_{j-1}\}.$$
Then $L=(l_1,\dots,l_\nu)\in{\cal S}$ and it is a maximal vector.  Hence Theorem \ref{thm1} is proven with the following proof of Lemma \ref{lemma1}.

\begin{proof}{(Proof of Lemma \ref{lemma1})}
Assuming otherwise, for $\Lambda\in{\cal S}$ and for any $j$, there exists $k\neq j$ such that
 $$\Phi^\Lambda(a_j)={\rm Re}(\overline a_k a_j) + \lambda_k .  $$   Since the condition \eqref{eq:aj} means that $|a_j|^2 + \lambda_j \leq \Phi^\Lambda(a_j)$ we get that
$$ |a_j|^2 -{\rm Re}(\overline a_k a_j)  \leq  \lambda_k-\lambda_j\quad \mbox{ for some $k\neq j$}. $$
Let us use the notation $j\leadsto k$ to represent the above inequality.   Repeating the argument, there exists some $\ell\neq k$ such that $k\leadsto\ell$. Since the index set is finite, the chain of arrows (i.e. the chain of inequalities), $j\leadsto k\leadsto\ell\leadsto\dots$, must eventually repeat some entry and form a closed loop.  Without losing generality let the closed loop be $1\leadsto 2\leadsto \dots \leadsto s\leadsto 1$.  Adding up the corresponding inequalities, we get
\begin{equation}\label{eq:loop}
\begin{split}
 \big(|a_1|^2 -{\rm Re}(\overline a_2 a_1)\big) + \big(|a_2|^2 -{\rm Re}(\overline a_3 a_2)\big) + \dots +\big(|a_s|^2 -{\rm Re}(\overline a_1 a_s)\big) 
 \\
\leq (\lambda_1 - \lambda_2) + (\lambda_2 - \lambda_3)+\dots (\lambda_s - \lambda_1) =0.
\end{split}
\end{equation}
The left hand side is the half the inner product, $({\bf A} - \widehat{\bf A})\cdot({\bf A} - \widehat{\bf A})^*$ (with the complex conjugation denoted by $*$), where the vectors ${\bf A}$ and $\widehat{\bf A}$ are given by
\begin{align}\nonumber
{\bf A} = ( a_1,\dots, a_s),\quad \widehat{\bf A} = ( a_2, a_3,\dots, a_s, a_1). 
\end{align}
This leads to $a_1= a_2=\dots = a_s$, a contradiction.
\end{proof}

If $L$ be maximal in the sense of Theorem \ref{thm1} and if $a_j\notin K^L_j$ for some $j$ then one can increase $l_j$ slightly without breaking the condition \eqref{eq:aj}.
This shows that the maximal $L$ occurs only if every $a_j$ is in $\partial K^L_j$.
Then there are two possibilities. 
\begin{align}
\label{eq:0} &\Phi^L(a_j)=|a_j|^2+l_j = \log|a_j|  \quad &&\mbox{if $a_j\in\partial K^L_0$}, 
\\
\label{eq:arrow}
&\Phi^L(a_j)=|a_j|^2 + l_j  =  {\rm Re}(\overline a_k a_j) + l_k \quad &&\mbox{if $a_j\in\partial K^L_k$ for some $k\neq j$.}
\end{align}
According to the notation defined in \eqref{def-arrow} we note that the former case corresponds to $j\to 0$ and the latter case corresponds to $j\to k$.

\begin{lemma}\label{lem:chain}
Given $j\neq 0$, the chain of arrows, $j\to k\to \ell\to \dots$, eventually leads to $\dots \to 0$ without repeating any entry. 
\end{lemma}
\begin{proof}
Idea of the proof is similar to that of Lemma \ref{lemma1}.
Given successive relations, $j\to k \to \ell \to \dots$, it is enough to show that the chain of arrows never visits any nonzero number twice.  To prove this statement, assume that we have a loop $j_1\to j_2\to \dots \to j_s\to j_1$ with all $j$'s being nonzero.  We get
\begin{align}\nonumber
\begin{split}
 \big(|a_{j_1}|^2 -{\rm Re}(\overline a_{j_2} a_{j_1})\big) + \big(|a_{j_2}|^2 -{\rm Re}(\overline a_{j_3} a_{j_2})\big) + \dots +\big(|a_{j_s}|^2 -{\rm Re}(\overline a_{j_1} a_{j_s})\big) 
 \\
 =  (l_1 - l_2) + (l_2 - l_3)+\dots (l_s - l_1) =0.
 \end{split}
\end{align}
By the argument after the equation \eqref{eq:loop}, we obtain $a_{j_1}=a_{j_2}=\dots=a_{j_s}$, a contradiction.
\end{proof}

\begin{defn}\label{def chain}
From Lemma \ref{lem:chain}, for each $a_j$, there exists chains $j\to\dots\to 0$.   We define {\em the level of $a_j$} by the smallest number of arrows among all the chains that starts with $j$.  For example, the level of $a_j$ is one if $j\to0$.  Lemma \ref{lem:chain} says that the level of $a_j$ should be $\leq \nu$.
\end{defn}

\begin{remark}
For a generic choice of $\{a_1,\dots,a_m\}\subset {\mathbb D}$, an $a_j$ is adjacent to exactly two regions, i.e., $a_j\in K^L_j\cap K^L_k$ for exactly one $k\neq j$ among $0\leq k\leq m$.  It means that $j\to k$.  For a generic case there exists a unique chain $j\to \dots \to 0$ for each $a_j$.   For a non-generic case $a_j$ can be at the boundary of three or more regions. To avoid too much technicality we do not consider such case.   See Figure \ref{nongeneric case} for a non-generic case.    
\end{remark}

\begin{proof}(Proof of Theorem \ref{Thm1})  Since the existence part is proven in Theorem \ref{thm1} we only prove the uniqueness.
We claim that the following iterative steps finds $L=(l_1,\dots,l_\nu)$ in Theorem \ref{Thm1}, hence $L$ is unique.   
\\\noindent{\bf Algorithm to find $L$.}
\begin{itemize}
  \setlength\itemsep{-0.2em}
\item[1.] Set $\lambda_j = \log|a_j|-|a_j|^2$ for all $j=1,\dots,\nu$.  (If $a_j=0$ then set $\lambda_j= \log|a_k|-|a_k|^2$ for some $k\neq j$.)
\item[2.] Define $\widetilde \lambda_j= \Phi^\Lambda(a_j)-|a_j|^2$ for all $j$. Note that $\widetilde \lambda_j\geq \lambda_j$ since $\Phi^\Lambda(a_j)\geq |a_j|^2+\lambda_j$.
\item[3.] Redefine $\lambda_j = \widetilde\lambda_j$ and, accordingly, $\Lambda$. 
\item[4.] Repeat the above two steps $\nu$ times.  
\item[5.] Set $L=\Lambda$.
\end{itemize}
If $a_j$ is of level one, i.e. $j\to 0$, $\lambda_j=l_j$ is obtained by the step 1.  Since, for all $j$, 
$|a_j|^2+l_j =\Phi^L(a_j)\geq \log|a_j|$, we get $\Lambda\leq L$ after the step 1, i.e. $\lambda_j\leq l_j$ for all $j$.    
In the prospect of using induction, let us assume that $\lambda_j=l_j$ for all the $a_j$'s up to the $k$th level while $\Lambda\leq L$.  Let $a_q$ be of level $k+1$, i.e. $q\to q'\to\dots 0$ where $a_{q'}$ is of level $k$.
Since we already have $\lambda_{q'}=l_{q'}$ by the assumption, the step 2 gives $\widetilde\lambda_q=\Phi^\Lambda(a_q)-|a_q|^2\geq {\rm Re}(\overline{a}_{q'} a_q)+l_{q'}-|a_q|^2=l_q$  where the last equality is from $q\to q'$.  On the other hand, we have, for all $j$, $\widetilde\lambda_j=\Phi^\Lambda(a_j)-|a_j|^2\leq \Phi^L(a_j)-|a_j|^2=l_j$ since $\Lambda\leq L$.   The last two sentences lead to $\widetilde\lambda_q=l_q$ and $\widetilde\Lambda\leq L$.  The step 3 will then give $\lambda_q=l_q$ and $\Lambda\leq L$.   This shows that the step 2 and 3 can be repeated inductively.   
 
Since the largest level of $a_j$ is $\leq\nu$, the induction will give $\Lambda = L$ after $\nu$ iterations.

This ends the proof of Theorem \ref{Thm1}.
\end{proof}

Let us repeat the definition of $\Gamma$ in \eqref{def-gamma} using that $\Omega_j={\rm Int} K_j^L$ as can be seen from \eqref{def omegaj}.  
\begin{defn}\label{def msc}
(Multiple Szeg\H{o} curve)  Given $\{a_j\}_{j=1}^\nu\subset \DD$, we define the {\em multiple Szeg\H{o} curve}
$$ \Gamma = \bigcup_{j=1}^\nu\partial K^L_j,$$
where $K^L_j$ are defined by \eqref{eq:K} and $L$ is the unique vector that is asserted in Theorem \ref{Thm1} and can be obtained explicitly  by the iterative algorithm.  We define the oriented arc,
\begin{equation}\label{def gammajk}
   \Gamma_{jk} = K^L_j\cap K^L_k , \quad j\neq k,\quad  \{j,k\} \subset  \{0,1,\dots,\nu\},  
\end{equation}
whose orientation is such that $K^L_j$ sits to the left with respect to the traveller who follows the orientation along the arc, i.e., $K^L_j$ is at the + side of $\Gamma_{jk}$. 
\end{defn}

From the definition, the contours, $\Gamma_{jk}$ and $\Gamma_{kj}$, have the opposite orientations.  Also, $\Gamma_{jk}$ will be a straight line segment only if both indices are nonzero.

\begin{lemma}
For $\{a_j\}_{j=1}^\nu\subset \DD$, the corresponding multiple Szeg\H{o} curve $\Gamma$ is in $\DD$, i.e. $\partial K_0\subset \DD$.
\end{lemma}

\begin{proof}
Let $L$ be the unique vector in Theorem \ref{Thm1}. We first show that 
\begin{equation}\label{eq:sandwich} \log|z|\leq\Phi^L(z)\leq \frac{1}{2}(|z|^2-1).
\end{equation}
The first inequality is trivial from the definition of $\Phi^L$.  
Let us prove the latter inequality. 
Let $1\to 2\to 0$ \eqref{def-arrow} be a chain. Since $a_2\in K_0$ we have $\Phi^L(a_2)=\log|a_2|\leq \frac{1}{2}(|a_2|^2-1)$, i.e. the inequality holds at $a_2$.  In fact, the inequality holds for all $z\in {\rm Int} K_2$ because
\begin{equation}\label{eq:z2-1}
    \frac{1}{2}(|z|^2-1) - {\rm Re}(\overline a_2 z)-l_2 = \frac{1}{2}(|z-a_2|^2-|a_2|^2-1)-l_2
\end{equation}
has the global minimum at $z=a_2$.
Since $\Phi^L$ is continuous, the inequality holds at $a_1\in \partial K_2$.  By the same argument the inequality holds for all $z\in {\rm Int} K_1$.  By induction, the argument applies to any chain and, therefore, the inequality holds for all $K_j$, $j=1,\dots,\nu$.

Now we assume that $K_j$ intersects $\partial \DD$.  Then there exists a point $p\in K_j\cap \partial \DD$.  By the squeezing inequalities \eqref{eq:sandwich} we have $\Phi^L(p)=\frac{1}{2}(|p|^2-1)=0$.  Since $a_j$ is the unique minimum of $\frac{1}{2}(|z|^2-1) - \Phi^L(z)$ in $K_j$ by \eqref{eq:z2-1} applied to $a_j$, we have $\frac{1}{2}(|a_j|^2-1)- \Phi^L(a_j)<\frac{1}{2}(|p|^2-1)- \Phi^L(p)=0$, a contradiction to the established inequality \eqref{eq:sandwich}.
\end{proof}

\section{Dealing with non-integer \texorpdfstring{$c$}{c}'s}\label{section noninteger}

We use Theorem 1 in \cite{Lee 2018} which states that the polynomial $p_n$ is a multiple orthogonal polynomial of type II \cite{Fi 2017,Assche 2001,Ku 2010} and Theorem 2 in \cite{Lee 2018} which relates such multiple orthogonal polynomial to a $(\nu+1)\times(\nu+1)$ Riemann-Hilbert problem. Let us recapture both theorems by Theorem \ref{thm 3.1} and Theorem \ref{thm 3.2} below.

Let us assume that $a_j$'s are all distinct for a simple presentation.  
Without losing generality we set
\begin{equation}\nonumber
0\leq \arg a_1<\dots<\arg a_\nu<2\pi   
\end{equation}
and let $\gamma$ be a simple closed curve given by $\overrightarrow{a_1a_2}\cup \overrightarrow{a_2a_3}\cup \dots \cup \overrightarrow{a_\nu a_1}$ where $\overrightarrow{AB}$ stands for the line segment connecting $A$ and $B$. 
We assign the orientation to $\gamma$ such that the curve encloses the origin in counterclockwise direction. 
\begin{thm}\label{thm 3.1}
Let $\kappa=\lfloor{n/\nu}\rfloor$ and set $(n_1,\dots,n_\nu)$ by
\begin{align}\nonumber
 n_j = \begin{cases}
 \kappa + 1 \quad &\text{if}\quad j\leq n-\kappa \nu,
 \\ \kappa \quad & \text{otherwise.}
 \end{cases}
\end{align}
For a fixed $1\leq k\leq \nu$ the polynomial $p_n(z)$ satisfies the orthogonality,
\begin{equation}\label{orth}
\int_\gamma p_n(z)z^i \chi_j^{(1)}(z)\dd z=0\quad \text{for}\quad 0\leq i\leq n_j-1\quad \mbox{and}\quad 1\leq j\leq \nu, \end{equation}
with respect to the $\nu$ multiple measures $\{\chi_1^{(1)}dz,\dots,\chi_\nu^{(1)}dz\}$ given by
\begin{equation*}
    \chi_j^{(k)}(z)=W(z)\int_{\overline{a}_k}^{{\overline z}\times\infty} \overline{W(\overline s)}\prod_{i=1}^\nu{(s-\overline{a}_i)^{n_i-\delta_{ij}}}\, \ee^{-N z s } \dd s.
\end{equation*}
Above the integration contour starts at $\overline{a}_k$ and extends to $\infty$ in the angular direction of $\arg \bar z$ while avoiding $\bigcup_j \{ \overline{a}_j t: t\geq 1\}$.  We will use an alternative but equivalent expression 
\begin{equation}\label{chi}
    \chi_j^{(k)}(z)=W(z)\bigg(\int_{a_k}^{z\times\infty} W( s)\prod_{i=1}^\nu{(s-a_i)^{n_i-\delta_{ij}}}\, \ee^{-N \overline z s } \dd s\bigg)^*,
\end{equation}
where the superscript $*$ stands for the complex conjugation,
and the integration contour starts from $a_k$ and escapes to $z\times\infty$, the infinity in the angular direction of $\arg z$, while avoiding ${\bf B}$ \eqref{def B}.  We also note that the above definitions make sense only for $z\notin {\bf B}$.\end{thm}

\begin{thm}\label{thm 3.2}Let $\gamma$ and $ \chi_j^{(k)}(z)$ be given above, then the Riemann-Hilbert problem,
\begin{equation}
\begin{cases}
Y(z)\,\,\mbox{has the continuous limit values from each side of  $\gamma$ and is holomorphic in $\CC\setminus\gamma$},
\\
Y_+(z)=Y_-(z)\begin{bmatrix}
1&\chi_1^{(1)}(z)&\cdots&\chi_{\nu}^{(1)}(z) \\
0&1&\cdots&0\\
\vdots&\vdots&\ddots&\vdots\\
0&0&\cdots&1
\end{bmatrix},& z\in\gamma,
\\
Y(z)=\displaystyle\left(I_{\nu+1}+{\cal O}\left(\frac{1}{z}\right)\right){\rm diag}\left(z^{n},z^{-n_1},\dots,z^{-n_\nu}\right),& z\to\infty,
\end{cases}
\end{equation}
has the unique solution given by
\[
Y(z) =\begin{blockarray}{ccccc}
\begin{block}{[cccc]c}
p_n(z) &\displaystyle \frac{1}{2\pi
\mathrm{i}}\int_\gamma\frac{p_n(w)\chi_1^{(1)}(w)}{w-z}\,\dd w
 &\cdots &\displaystyle \frac{1}{2\pi \mathrm{i}}\int_\gamma\frac{p_n(w)\chi_\nu^{(1)}(w)}{w-z}\,\dd w&  \vspace{0.1cm}\\
\vdots&\vdots&\vdots&\vdots &  \vspace{0.1cm}\\
q_{n-1}^{(j)}(z)&\displaystyle \frac{1}{2\pi
\mathrm{i}} \int_\gamma\frac{q_{n-1}^{(j)}(w)\chi_1^{(1)}(w)}{w-z}\,\dd w&\cdots&\displaystyle \frac{1}{2\pi
\mathrm{i}} \int_\gamma\frac{q_{n-1}^{(j)}(w)\chi_\nu^{(1)}(w)}{w-z}\,\dd w & \leftarrow(j+1)\text{th}\,\, \text{row}, \vspace{0.1cm}\\
\vdots&\vdots&\vdots&\vdots &  \vspace{0.1cm}\\
\end{block}
\end{blockarray}
 \] where $q_{n-1}^{(j)}(z)$ is
the polynomial of degree $n-1$ satisfying the
orthogonality condition:
\begin{equation}\nonumber
\displaystyle \int_\gamma q_{n-1}^{(j)}(z)\,z^m\chi_i^{(1)}(z)\,\dd z=0, \quad 0\leq
m\leq n_i-1-\delta_{ij},\quad 1\leq i,\, j\leq \nu,\end{equation} and the leading coefficient of $q_{n-1}^{(j)}(z)$ is given by
$\left(\frac{-1}{2\pi \mathrm{i}} \int_\gamma
q_{n-1}^{(j)}(w)w^{n_j-1}\chi_j^{(1)}(w)\,\dd
w\right)^{-1}.  $
\end{thm}
\begin{remark}
In all the Riemann-Hilbert problems that we subsequently write, we will state only the jump conditions and the boundary behaviors.  We assume the analyticity of the solution away from the jump contours and the existence of the continuous boundary values, that we denote by the subscripts $+$ and $-$, alongside the contours.\end{remark}

The goal of the next four subsections is to define the $\nu\times\nu$ matrix function $\Psi(z)$ that we will use for the subsequent Riemann-Hilbert analysis.  
These sections are needed mostly to be able to handle non-integer values of $c$'s.  We first define $\widetilde\Psi$ which has its branch cut on ${\bf B}$.  The method of Riemann-Hilbert analysis is to construct, by a succession of transformations, the Riemann-Hilbert problem with a desirable jump condition.  The jump on ${\bf B}$ that is originated from the non-integer $c$'s turns out not desirable -- the jump matrix could increase exponentially in $N$.  The cure is to deform the jump conditions on ${\bf B}$ into the jump conditions on ${\bf B}[k]$'s.  This is exactly done in \hyperref[sec 3.3]{Section 3.3} by defining the transformation matrix ${\bf V}$ and $\Psi(z)={\bf V}(z) \widetilde \Psi(z)$.  

This section, unfortunately, involves rather lengthy manipulation of intertwined branch cuts.  If one is interested in only integer $c_j$'s then one can skip ahead to Section 4, simply by using $\widetilde\Psi$ \eqref{def psitilde and hat} in the place of $\Psi$.

\subsection{Construction of \texorpdfstring{$\widehat\Psi$}{widehatpsi} and \texorpdfstring{$\widetilde\Psi$}{widetildepsi}}


Let us define a shorthand notation
\begin{equation}\nonumber
  \eta_j = \ee^{-2\pi i c_j},\quad j=1,\dots,\nu.
\end{equation}
It will be convenient to define piecewise analytic row vectors for $j=1,\dots,\nu$ by
\begin{equation}\label{psi tilde}
    {\widetilde \psi}_j(z)=W(z)^{-1}\Big[\chi_1^{(j)}(z),\dots,\chi_\nu^{(j)}(z) \Big],  \qquad z\notin{\bf B}\cup\widehat{\bf B}. 
\end{equation}
Using the above definitions, let us define $\nu\times\nu$ matrix:
\begin{equation}\label{def psitilde and hat}
    {\widetilde \Psi}(z)=\begin{bmatrix}
 \widetilde\psi_1(z)\\
\vdots\\
 \widetilde\psi_\nu(z)
\end{bmatrix}.
\end{equation}

\begin{lemma}\label{lem 3.3}
The matrix $\widetilde \Psi$ satisfies a jump discontinuity,
\begin{equation}\label{eq-branch}
{\widetilde \Psi}_+(z)=\widetilde{J}_{j}{\widetilde \Psi}_-(z), \quad z\in \widehat{\mathbf{B}}_j\cup{\bf B}_j,
\end{equation}
where
\begin{equation}\label{eq: Jpsi}
\widetilde{J}_{j}=\left[{\begin{array}{c:c:c}
\begin{matrix}
\vspace{-0.4cm}\\I_{j-1}
\end{matrix}
&\begin{matrix}\vspace{-0.4cm}\\\eta_j^{-1}-1\\\vdots\\\eta_j^{-1}-1\end{matrix} &\begin{matrix}\vspace{-0.4cm}\\{\bf
0}\end{matrix}\vspace{0.1cm}\\\hdashline
\begin{matrix}
\vspace{-0.4cm}\\{\bf 0}
\end{matrix}
&\begin{matrix}\vspace{-0.4cm}\\\eta_j^{-1}\end{matrix}
&\begin{matrix}\vspace{-0.4cm}\\{\bf 0}\end{matrix}\vspace{0.1cm}\\\hdashline
\begin{matrix}\vspace{-0.4cm}\\{\bf 0}\end{matrix}
&\begin{matrix}\vspace{-0.4cm}\\\eta_j^{-1}-1\\\vdots\\\eta_j^{-1}-1\end{matrix} &\begin{matrix}
\vspace{-0.4cm}\\I_{\nu-j}
\end{matrix}\end{array}}\right],\quad \text{for}\quad j=1,\dots,\nu.    
\end{equation}
\end{lemma}  

Let us remark about the notations that will be used throughout the paper. The subscript $\pm$ will be used for the boundary values on the sides of the said contour.  In \eqref{eq-branch} for example, $\widetilde\Psi_+(z)$ refers to the boundary value of $\widetilde\Psi(z)$ on the $+$ side of $\widehat{\mathbf{B}}_j\cup{\bf B}_j$, i.e., the left side from the point of view of the traveler that walks along the oriented contour.  Note that the orientations of $\widehat{\bf B}_j$ and ${\bf B}_j$ are given by \eqref{def bhat} and \eqref{def B}.

Here and below, we will express matrices in block form as in $\widetilde{J}_{j}$ \eqref{eq: Jpsi}. We will use $I_k$ for an identity matrix (block) of size $k\times k$, and ${\bf 0}$ for zero matrix (block) of appropriate size that is determined by the size of its neighboring blocks. For example the four ${\bf 0}$'s in $\widetilde{J}_{j}$ are of sizes $(j-1)\times(\nu-j)$, $1\times(j-1)$, $1\times(\nu-j)$, and $(\nu-j)\times (j-1)$, respectively.  

\begin{proof}{(Proof of Lemma \ref{lem 3.3})}
From the jump conditions
\begin{equation}\nonumber
{W_+(z)}=\eta_j\,W_-(z),\quad z\in{\bf B}_j,\quad j=1,\dots,\nu,\end{equation} 
we obtain
\begin{equation}\nonumber
\big[\widetilde \psi_j(z)\big]_+ =\eta_j^{-1}\,\big[\widetilde \psi_j(z)\big]_-  ,\quad z\in{\bf B}_j\cup\widehat{\bf B}_j.
\end{equation}

Using that $\widetilde\psi_j(z)-\widetilde\psi_k(z)$ is analytic in the whole complex plane we get 
\begin{align*}
\big[\widetilde \psi_k(z)\big]_+
&=\big[\widetilde \psi_k(z)-\widetilde \psi_j(z)\big]_++\big[\widetilde \psi_j(z)\big]_+
\\
&=\big[\widetilde \psi_k(z)-\widetilde \psi_j(z)\big]_-+\eta_j^{-1}\,\big[\widetilde \psi_j(z)\big]_-\\
&=\big[\widetilde \psi_k(z)\big]_- +\big(\eta_j^{-1}-1\big) \big[\widetilde \psi_j(z)\big]_-,  \quad z\in{\bf B}_j\cup\widehat{\bf B}_j.
\end{align*}
This proves the lemma.
\end{proof}

We also define $\widehat \psi_j$ to be analytic in $\{z|\arg z \neq \arg a_j+\pi\}$ such that
\begin{equation}\label{psi hat}
    {\widehat \psi}_j(z)=\widetilde\psi_j(z)  \quad \text{when~}\arg a_j<\arg z<\arg a_{j+1}, 
\end{equation}
where we use $a_{\nu+1}= a_1$ when $j=\nu$. 

Using the above definitions, let us define $\nu\times\nu$ matrix:
\begin{equation}
    {\widehat \Psi}(z)=\begin{bmatrix}
 \widehat\psi_1(z)\\
\vdots\\
 \widehat\psi_\nu(z)
\end{bmatrix}.
\end{equation}

Next we find the linear transform from $\widehat \Psi(z) $ to $\widetilde \Psi(z) $.   To describe the transform, we need notations to handle the situation: given $z\in{\mathbb C}$ we want to refer to $\{a_j\mbox{'s}\}$ by the order of $\{{\arg }(z/a_j)\mbox{'s}\}$, i.e., by the order of the angular distances from $z$ with respect to the origin.

Let 
$$ {\cal I}= \{1,\dots,\nu\}$$
and we define
\begin{equation}\nonumber
 {\mathfrak s}=   {\mathfrak s}(z)= \#\{k\in {\cal I} :  0<\arg(z/a_k)\leq \pi \}.
\end{equation}
We define $\mathfrak r(k)=\mathfrak r(k;z)$ and $\mathfrak l(k)=\mathfrak l(k;z)$ to be the renaming of the indices in ${\cal I}$, $${\cal I}= \{{\mathfrak r}(1;z),\dots,{\mathfrak r}({\mathfrak s};z)\}\cup\{{\mathfrak l}(1;z),\dots,{\mathfrak l}(\nu-{\mathfrak s};z)\}, $$ 
such that to satisfy
$$-\pi< 
\arg \frac{z}{a_{{\mathfrak l}(1)}}< 
\arg \frac{z}{a_{{\mathfrak l}(2)}}<
\dots<
\arg \frac{z}{a_{{\mathfrak l}(\nu-{\mathfrak s})}}<
0 <
\arg \frac{z}{a_{{\mathfrak r}({\mathfrak s})}}<
\dots<
\arg \frac{z}{a_{{\mathfrak r}(2)}}< 
\arg \frac{z}{a_{{\mathfrak r}(1)}}\leq \pi. $$
Above and below we sometimes write, for example, $\mathfrak r(1)$ and $\mathfrak l(2)$ instead of $\mathfrak l(1,z)$ and $\mathfrak l(2,z)$ when the second argument $z$ is clear from the context.  The same will be true for ${\mathfrak s}$, i.e. $\mathfrak s$ instead of $\mathfrak s(z)$. See Figure \ref{indices}.

\begin{figure}
\begin{center}
\includegraphics[width=0.3\textwidth]{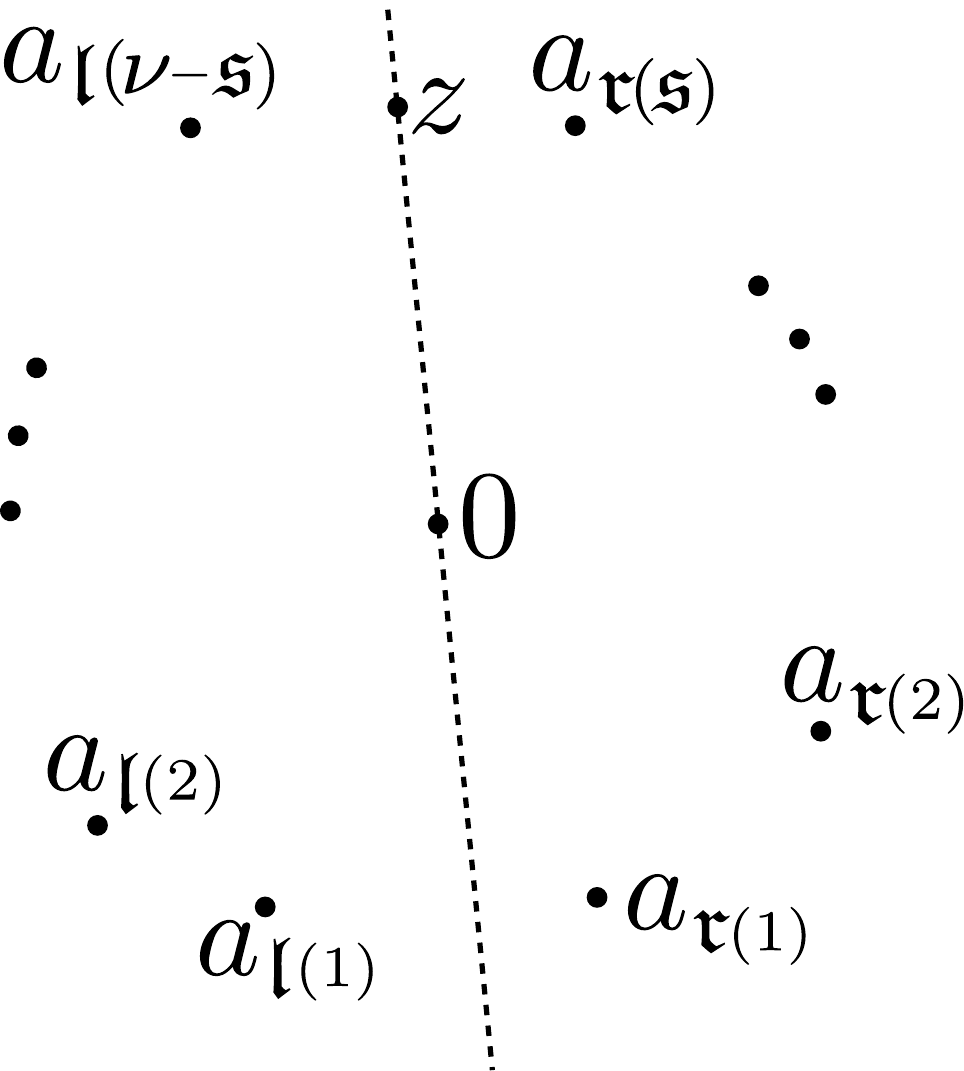}
\end{center}
 \caption{Renaming the indices of $\{a\mbox{'s}\}$ by the relative argument with respect to $z$.  The variables ${\frak l},{\frak r}$ and ${\frak s}$ are determined in terms of $z$, i.e., ${\frak l}(j)={\frak l}(j,z)$, ${\frak r}(j)={\frak r}(j,z)$ and ${\frak s}={\frak s}(z)$.
}\label{indices}
\end{figure}

\begin{lemma}\label{lemma 3.4}
Let ${\mathfrak l}(*)={\mathfrak l}(*;z)$, ${\mathfrak r}(*)={\mathfrak r}(*;z)$ and ${\mathfrak s}={\mathfrak s}(z)$ be defined as above. We have
\begin{align*}
\widetilde\psi_{\mathfrak r(k)}(z) &=\widehat\psi_{\mathfrak r(k)}(z) +\sum_{j=k+1}^{\mathfrak s}(1-\eta_{{\mathfrak r}(j)})\widehat \psi_{\mathfrak r(j)}(z), &&k=1,\dots, \mathfrak{s},
\\
\widetilde\psi_{\mathfrak l(k)}(z) &=\eta_{\mathfrak l(k)} \widehat\psi_{\mathfrak l(k)}(z) +\sum_{j=k+1}^{\nu-\mathfrak s}(\eta_{{\mathfrak l}(j)}-1)\widehat \psi_{\mathfrak l(j)}(z),  &&k=1,\dots, \nu-\mathfrak{s},
\end{align*}
for $z\notin \widehat{\mathbf{B}}\cup{\bf B}$ and $-z\notin \widehat{\mathbf{B}}\cup{\bf B}$.
\end{lemma}
\begin{proof}
These expansions are obtained from the integral representations.
For example, the integration contour for  $\widetilde\psi_{\frak r(k)}(z)$, which is from $a_{\frak r(k)}$ to $z\times\infty$, is the sum of the contours around ${\bf B}_j$'s as shown in Figure \ref{int cont decop} for $\frak r(k)=1$. The red dashed contour can be expressed into the sum of blue contours enclosing $\{{\bf B}_1,{\bf B}_2,\dots,{\bf B}_6\}$.
The lemma follows since the contour enclosing  ${\bf B}_j$ clockwise corresponds to $(1-\eta_j)\widehat\psi_j$.  Since $\widehat\psi_j$ has the branch cut on $\{z|\arg z= \arg a_j +\pi \}$, $\widehat\psi_{\frak r(k)}$ (resp., $\widehat\psi_{\frak l(k)}$) is analytic in the angular sector from the argument of ${\bf B}_{\frak r(k)}$ (resp., ${\bf B}_{\frak l(k)}$) to $\arg z$ in the counterclockwise (resp., clockwise) direction.       \end{proof}
\begin{figure}  
\centering
\begin{tikzpicture}[scale=0.6]
\draw[dotted]  (0,0) -- (5,5);
\draw[dotted]  (0,0) -- (5,7);
\draw[dotted]  (0,0) -- (2,4);
\draw[dotted]  (0,0) -- (1,9);
\draw[dotted]  (0,0) -- (-1,5);
\draw[dotted]  (0,0) -- (-4,8);

\draw[rotate around={-45:(5,5)},blue,dashed,->] (5,5) parabola (4.5,11);
\draw[rotate around={-35:(5,7)},blue,dashed,->] (5,7) parabola (4.5,11.5);
\draw[rotate around={-35:(5,7)},blue,dashed] (5,7) parabola (5.5,11.5);
\draw[rotate around={-27:(2,4)},blue,dashed,->] (2,4) parabola (1.5,12.5);
\draw[rotate around={-27:(2,4)},blue,dashed] (2,4) parabola (2.5,12.5);
\draw[rotate around={-8:(1,9)},blue,dashed,->] (1,9) parabola (0.05,13);
\draw[rotate around={-8:(1,9)},blue,dashed] (1,9) parabola (1.95,13);
\draw[rotate around={11:(-1,5)},blue,dashed,->] (-1,5) parabola (-2,13);
\draw[rotate around={11:(-1,5)},blue,dashed] (-1,5) parabola (-0.0,13);
\draw[rotate around={26:(-4,8)},blue,dashed,->] (-4,8) parabola (-4.8,13.5);
\draw[rotate around={26:(-4,8)},blue,dashed] (-4,8) parabola (-3.2,13.5);





\draw[red, dashed,->] (5,5) .. controls (0,0) and (0,0) .. (-8,13);

\draw[ ->] (5,5) -- (9,9);
\draw[ ->] (5,7) -- (7.5,10.5);
\draw[ ->] (2,4) -- (5.7,11.4);
\draw[ ->] (1,9) -- (1.4,12.6);
\draw[ ->] (-1,5) -- (-2.5,12.5);
\draw[ ->] (-4,8) -- (-6.3,12.6);



\foreach \Point/\PointLabel in {(0,0)/0, (5,5)/a_1, (5,7)/a_2, (2,4)/a_3, (1,9)/a_4, (-1,5)/a_5, (-4,8)/a_6}
\draw[fill=black] \Point circle (0.05) node[below right] {$\PointLabel$};

\foreach \Point/\PointLabel in {(-4.45,8.36)/z}
\draw[fill=red] \Point circle (0.07) node[left] {$\PointLabel$};

\foreach \Point/\PointLabel in {(9,9)/{\bf B}_1, (7.5,10.5)/{\bf B}_2, (5.7,11.4)/{\bf B}_3, (1.4,12.6)/{\bf B}_4, (-2.5,12.5)/{\bf B}_5,(-6.2,12.6)/{\bf B}_6}
\draw[fill=black]  \Point
 node[below right] {$\PointLabel$};
 
 
 
 
  
 \end{tikzpicture}
\caption{The red dashed line is the integration contour of $\widetilde\psi$ \eqref{psi tilde} and the blue dashed lines are the integration contours of $\widehat\psi$'s \eqref{psi hat}. The red contour is equivalent to the sum of the blue contours. }\label{int cont decop} 
\end{figure}
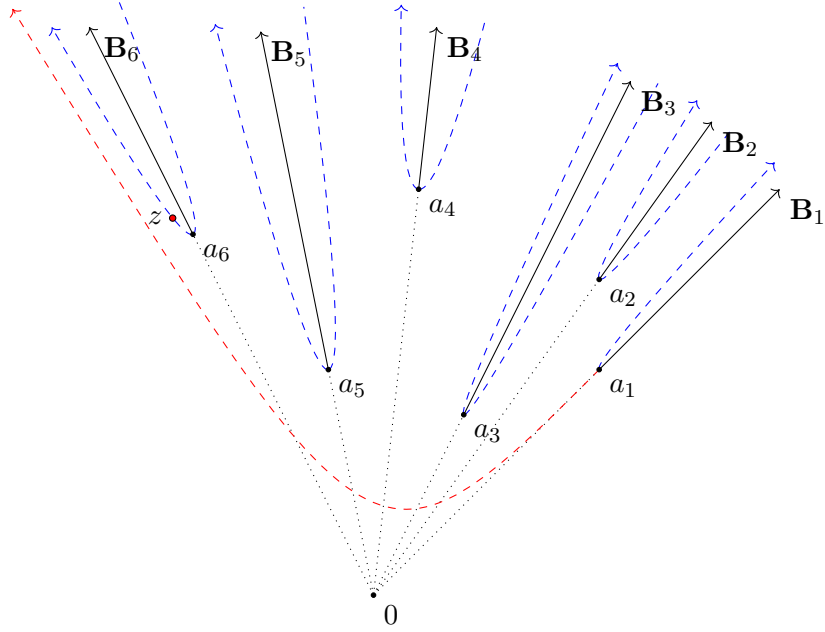

Alternatively one can also understand this relation between $\widetilde\psi$ and $\widehat\psi$ through the integral representations as explained in Figure \ref{int cont decop}.

Let us define the piecewise constant $\nu\times \nu$ matrix $\widetilde{\bf V}(z)$
\begin{align}
   \widetilde{\bf V}(z) &=\displaystyle\sum_{k=1}^\frak{s}\bigg({\bf e}_{\frak r(k),\frak r(k)}+\sum_{j=k+1}^{\frak{s}}(1-\eta_{\frak r(j)}){\bf e}_{\frak r(k),\frak r(j)}\bigg) + \sum_{k=1}^{\nu-\frak{s}} \bigg(\eta_{\frak l(k)}{\bf e}_{\frak l(k),\frak l(k)}+\sum_{j=k+1}^{\nu-\frak{s}} (\eta_{\frak l(j)}-1){\bf e}_{\frak l(k),\frak l(j)}\bigg)\\ \label{eq:C tilde}
 &=\displaystyle\sum_{j=1}^\frak{s}\bigg({\bf e}_{\frak r(j),\frak r(j)}+(1-\eta_{\frak r(j)})\sum_{k=1}^{j-1}{\bf e}_{\frak r(k),\frak r(j)}\bigg) + \sum_{j=1}^{\nu-\frak{s}} \bigg(\eta_{\frak l(j)}{\bf e}_{\frak l(j),\frak l(j)}+(\eta_{\frak l(j)}-1)\sum_{k=1}^{j-1} {\bf e}_{\frak l(k),\frak l(j)}\bigg),  
\end{align}
where ${\bf e}_{ij}$ stands for the basis of $\nu\times\nu$ matrices whose only nonzero entry is 1 at $(i,j)$th entry.

Using the above matrix $\widetilde{\bf V}(z)$ we have, by Lemma \ref{lemma 3.4},
\begin{equation}\label{eq:psi tilde}
\widetilde \Psi(z) = \widetilde{\bf V}(z) \widehat \Psi(z).
\end{equation} 

\subsection{Construction of \texorpdfstring{$\Psi$}{psi}}\label{sec 3.3}

We will define $\Psi(z)$ by 
\begin{equation}\label{eq:psi}
  \Psi(z) = {\bf V}(z) \widetilde \Psi(z),
\end{equation} 
in terms of the piecewise constant $\nu\times\nu$ matrix ${\bf V}(z)$ that we define below.

For $z\notin\widehat{\mathbf{B}}\cup{\bf B}$ and for $i\in{\cal I}$ consider the line segment $\overrightarrow{a_iz}=\{ a_i + t (z-a_i) ~|~ 0\leq t\leq 1\}$.  We define 
\begin{equation}\nonumber
    q=q(i)=q(i;z)=\#\{\text{intersections of $\overrightarrow{a_i z}$ with ${\bf B}\setminus {\bf B}_i$}\}.
\end{equation}
If $0<\arg(z/a_i)\leq \pi$ then by simple geometry $\overrightarrow{a_i z}$ can intersect only ${\bf B}_{\frak r(*)}$ but not ${\bf B}_{\frak l(*)}$, and vice versa for $-\pi<\arg(z/a_i)\leq 0$. For example, $q(1;z)=3$ in Figure \ref{geo inter}. And by versus for $-\pi<\arg(z/a_i)\leq 0$.

Let us consider the case $i=\frak r(k)$ for some $k$, then all the $q$ number of intersections occur with ${\bf B}_{\frak r(p)}$ for $k<p\leq \frak s$. Hence we can define $p_1(i;z),p_2(i;z),\dots,p_q(i;z)$ such that ${\bf B}_{\frak r(p_*)}$ intersects $\overrightarrow{a_i z}$ and satisfies
$$ k<p_1<p_2<\dots<p_q\leq \frak s\qquad \text{where $i=\frak r(k)$.}$$
Similarly, for $-\pi<\arg(z/a_i)\leq 0$, we define $p_1(i;z),p_2(i;z),\dots,p_q(i;z)$ such that ${\bf B}_{\frak l(p_*)}$ intersects $\overrightarrow{a_i z}$ and satisfies
$$ k<p_1<p_2<\dots<p_q\leq \nu-\frak s\qquad \text{where $i=\frak l(k)$.}$$
We note that $p_*$'s are determined by the two arguments, $i$ and $z$. Therefore we will write $p_*=p_*(i;z)$ or, if the second argument is identified from the context, simply $p_*(i)$.  Similarly we write $q(i;z)$ or $q(i)$.

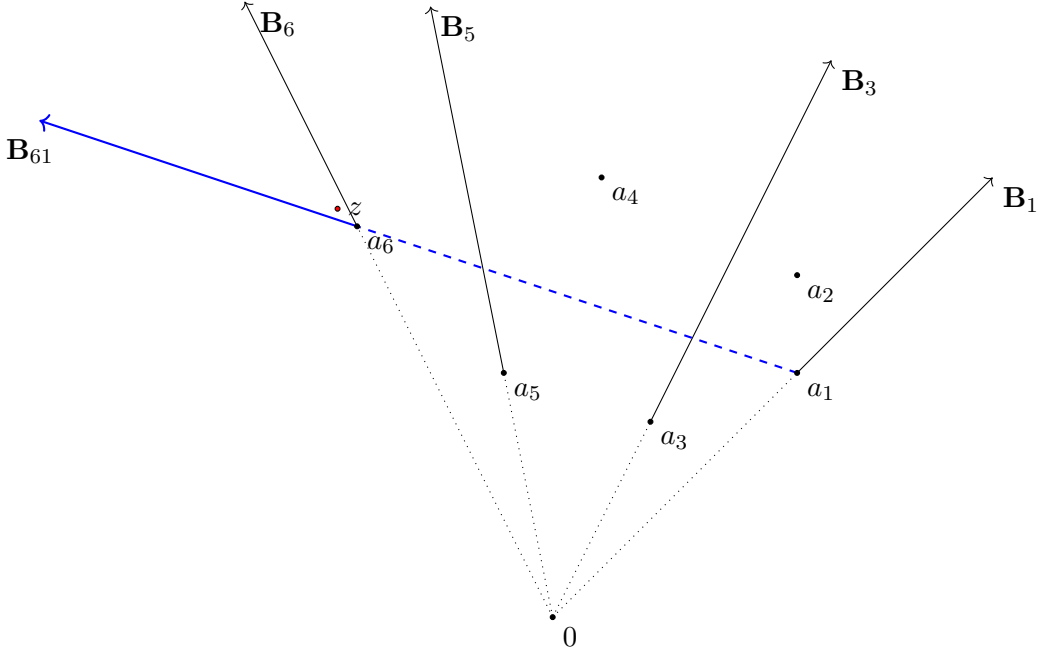
\begin{figure}   
\centering
\begin{tikzpicture}[scale=0.65]
\draw[dotted]  (0,0) -- (5,5);
\draw[dotted]  (0,0) -- (2,4);
\draw[dotted]  (0,0) -- (-1,5);
\draw[dotted]  (0,0) -- (-4,8);

\draw[->] (5,5) -- (9,9);
\draw[ ->] (2,4) -- (5.7,11.4);
\draw[ ->] (-1,5) -- (-2.5,12.5);
\draw[ ->] (-4,8) -- (-6.3,12.6);

\draw[blue,thick, dashed]  (5,5) -- (-4,8);

\draw[blue,  thick,->]  (-4,8) -- (-10.5,30.5/3);

\foreach \Point/\PointLabel in {(0,0)/0, (5,5)/a_1, (5,7)/a_2, (2,4)/a_3, (1,9)/a_4, (-1,5)/a_5, (-4,8)/a_6}
\draw[fill=black] \Point circle (0.05) node[below right] {$\PointLabel$};

\foreach \Point/\PointLabel in {(-4.4,8.36)/z}
\draw[fill=red] \Point circle (0.05) node[right] {$\PointLabel$};

\foreach \Point/\PointLabel in {(9,9)/{\bf B}_1, (5.7,11.4)/{\bf B}_3, (-2.5,12.5)/{\bf B}_5,(-6.2,12.6)/{\bf B}_6}
\draw[fill=black]  \Point
 node[below right] {$\PointLabel$};
 
\foreach \Point/\PointLabel in {(-10,10)/{\bf B}_{61}}
\draw[fill=black]  \Point
 node[below left] {$\PointLabel$};
 
 
 
  
 \end{tikzpicture}
  \caption{The line segment $\protect\overrightarrow{a_1 z}$ intersects with ${\bf B}_3$, ${\bf B}_5$ and ${\bf B}_6$.  Therefore $q(1;z)=3$ and $p_1(1;z)=3$, $p_2(1;z)=5$, $p_3(1;z)=6$. }\label{geo inter} \end{figure}

Now we can define the $\nu\times\nu$ matrix ${\bf V}(z)$ for $z\notin\widehat{\mathbf{B}}\cup{\bf B}$.  
\begin{equation}
\begin{split}\label{eq:C}
    {\bf V}(z)=I_\nu+ &\sum_{k=1}^{\frak s} \sum_{i=1}^{q({\frak r}(k))}
\Big(\prod_{j=1}^{i-1} \eta_{{\frak r}(p_j)}\Big)
(\eta_{{\frak r}(p_i)}-1){\bf e}_{\frak r(k),{{\frak r}(p_i)}}\,\,\qquad \text{(where $p_*=p_*({\frak r}(k);z)$)}
\\
+&\sum_{k=1}^{\nu-\frak s} \sum_{i=1}^{q({\frak l}(k))}
\Big(\prod_{j=1}^{i-1} \eta_{{\frak l}(p_j)}^{-1}\Big)
\left(\eta_{{\frak l}(p_i)}^{-1}-1\right){\bf e}_{\frak l(k),{{\frak l}(p_i)}} \qquad \text{(where $p_*=p_*({\frak l}(k);z)$)}.
\end{split}
\end{equation}
One can see that ${\bf V}(z)$ is a block matrix with $\frak s\times\frak s$ block and $(\nu-\frak s)\times(\nu-\frak s)$ block.  And each block has triangular structure due to $p_*>k$, which leads to $\det {\bf V}(z)=1$.

We also note that when $z$ is near ${\bf B}_j\cup\widehat{\bf B}_j$ the line segment $\overrightarrow{a_j z}$ does not intersect ${\bf B}\setminus\{a_j\}$ and therefore $q(j;z)=0$.  This implies that the the $j$th row of ${\bf V}(z)$ vanishes except $[{\bf V}(z)]_{jj}=1$, where $[M]_{jk}$ stands for the $(j,k)$th entry of the matrix $M$. Then it follows that the $j$th row of ${\bf V}(z)^{-1}$ also vanishes except $[{\bf V}(z)^{-1}]_{jj}=1$.

Let us recall \eqref{def bk}
$$ {\bf B}[j] =\bigcup_{i\neq j} {\bf B}_{ij}\cup{\bf B}_j  \mbox{ where } {\bf B}_{ij}=\{ a_i + (a_i-a_j)t~ |~ 0\leq t<\infty  \},  \quad 1\leq i,j\leq \nu,$$
with the orientation given by the direction of increasing $t$.
We recall \eqref{def wk} that $W_j(z)$ is the analytic continuation of $W(z)$ \eqref{def w} such that $W_j(z)=W(z)$ in a neighborhood of ${\bf B}_j$ and $W_j(z)$ is analytic away from ${\bf B}[j]$.

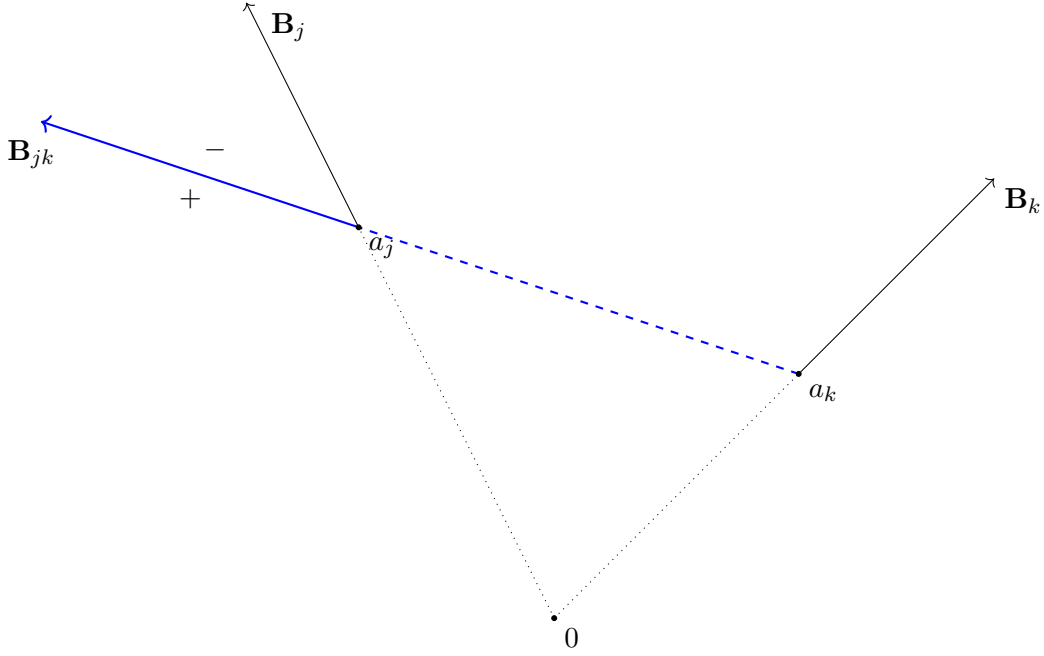
\begin{figure}   
\centering
\begin{tikzpicture}[scale=0.65]
\draw[dotted]  (0,0) -- (5,5);
\draw[dotted]  (0,0) -- (-4,8);

\draw[->] (5,5) -- (9,9);
\draw[ ->] (-4,8) -- (-6.3,12.6);

\draw[blue,thick, dashed]  (5,5) -- (-4,8);

\draw[blue,  thick,->]  (-4,8) -- (-10.5,30.5/3);

\foreach \Point/\PointLabel in {(0,0)/0, (5,5)/a_k, (-4,8)/a_j}
\draw[fill=black] \Point circle (0.05) node[below right] {$\PointLabel$};

\foreach \Point/\PointLabel in {(9,9)/{\bf B}_k, (-6,12.6)/{\bf B}_j}
\draw[fill=black]  \Point
 node[below right] {$\PointLabel$};
 
\foreach \Point/\PointLabel in {(-10,10)/{\bf B}_{jk}, (-7,9)/+, (-6.5,10)/-}
\draw[fill=black]  \Point
 node[below left] {$\PointLabel$};
 \end{tikzpicture}
 \caption{The branch cuts for $0<\arg (a_j/a_k)\leq \pi$. }\label{sub br} \end{figure}

\begin{lemma}\label{za ratio}
When $z\in{\bf B}_{jk}$, we have 
\begin{equation}\label{lem 350}
   \bigg[\frac{(z-a_j)^{c_j}}{[(z-a_j)^{c_j}]_{{\bf B}[k]}}\bigg]_+ = \begin{cases}
1, &\mbox{$0<\arg (a_j/a_k)\leq \pi$},  \\
1/\eta_j, &\mbox{$-\pi<\arg (a_j/a_k)\leq 0$},  
\end{cases}
\end{equation}
where $+$ denotes the boundary value evaluated from the $+$ side of ${\bf B}_{jk}$.
\end{lemma}
\begin{proof} When $0<\arg (a_j/a_k)\leq \pi$, by the definition of $[(z-a_j)^{c_j}]_{{\bf B}[k]}$ in \eqref{def subbk} we have
$$[(z-a_j)^{c_j}]_{{\bf B}[k]}=(z-a_j)^{c_j},\quad\mbox{ when }z\in {\bf B}_k\cup\widehat{\bf B}_k.$$
As $\big[(z-a_j)^{c_j}\big]_{{\bf B}[k]}$ is analytic away from ${\bf B}_{jk}$, it follows from Figure \ref{sub br} that
\begin{equation}\label{lem 351}
[(z-a_j)^{c_j}]_{{\bf B}[k]}=(z-a_j)^{c_j},\quad\mbox{when $z$ is in the $(+)$ side of ${\bf B}_{jk}$}.    \end{equation}
For $-\pi<\arg (a_j/a_k)\leq 0$ we have, by the similar argument,
\begin{equation}\label{lem 352}
[(z-a_j)^{c_j}]_{{\bf B}[k]}=\eta_j(z-a_j)^{c_j},\quad\mbox{when $z$ is in the $(+)$ side of ${\bf B}_{jk}$}.    \end{equation}
\end{proof}
 
\begin{lemma}\label{W ratio}
Let $z\in {\bf B}_{jk}$ be sufficiently close to $a_j$ such that the line segment $\overrightarrow{a_j z}$ does not intersect ${\bf B}\setminus\{a_j\}$.
Let 
$${\frak r}(*)={\frak r}(*;z) \quad\text{and}\quad {\frak l}(*)={\frak l}(*;z) ,$$and
\begin{equation}\nonumber
    q=q(k;z),\qquad p_i = p_i(k;z)\quad \text{ for } i=1,\dots,q.
\end{equation}
For $0<\arg(z/a_k)\leq\pi$
we have
\begin{equation}\label{62}
    \bigg[\frac{W_k(z)}{W_j(z)}\bigg]_- = \prod_{i=1}^q\eta^{-1}_{\frak r(p_i)}
    \text{~~and~~}
    \bigg[\frac{W_k(z)}{W_j(z)}\bigg]_+= \eta_j\prod_{i=1}^q\eta^{-1}_{\frak r(p_i)},
    \qquad z\in{\bf B}_{jk},
\end{equation}
and for $-\pi<\arg(z/a_k)\leq 0$ we have
\begin{equation}\nonumber
    \bigg[\frac{W_k(z)}{W_j(z)}\bigg]_+ =  \prod_{i=1}^{q}\eta_{\frak l(p_i)}
    \text{~~and ~~}
    \bigg[\frac{W_k(z)}{W_j(z)}\bigg]_- = \eta_j^{-1} \prod_{i=1}^{q}\eta_{\frak l(p_i)},  \qquad z\in{\bf B}_{jk}.
\end{equation}
\end{lemma}
\begin{proof} Let $z\in {\bf B}_{jk}$ be sufficiently close to $a_j$ such that the line segment $\overrightarrow{a_j z}$ does not intersect ${\bf B}\setminus\{a_j\}$. For such $z$ we have $$\frac{W(z)}{W_j(z)}=1$$ because $W_j(z)=W(z)$ when $z$ is near ${\bf B}_j$.  

For $z\notin {\bf B}[k]$ the line segment $\overrightarrow{a_kz}$ does not intersect ${\bf B}[k]$.  If $0<\arg(z/a_k)\leq\pi$ the same line segment crosses ${\bf B}_{\frak r(p_i;z)}$ for $i=1,\dots,q$ where $p_i=p_i(k;z)$ and $q=q(k;z)$.  We get 
\begin{equation}\nonumber
\frac{W_k(z)}{W(z)} =  \prod_{i=1}^q\eta^{-1}_{\frak r(p_i)}.
\end{equation}
If $-\pi<\arg(z/a_k)\leq 0$ the similar consideration gives
\begin{equation}\nonumber
\frac{W_k(z)}{W(z)} =  \prod_{i=1}^q\eta_{\frak l(p_i)}.
\end{equation}
Now let $z$ be very close to ${\bf B}_{jk}$.
For $z$ in the $+$ side of  ${\bf B}_{jk}$, we get ${\frak r(p_q)}=j$.  And we get, for $0<\arg(z/a_k)\leq\pi$,
\begin{equation}\nonumber
    \frac{W_{k,-}(z)}{W(z)} =  \prod_{i=1}^q\eta^{-1}_{\frak r(p_i)}
    \text{~~~ and ~~~}
    \frac{W_{k,+}(z)}{W(z)} = \eta_j \prod_{i=1}^{q}\eta^{-1}_{\frak r(p_i)}.
\end{equation}
For $-\pi<\arg(z/a_k)\leq 0$, the same consideration gives
\begin{equation}\nonumber
\frac{W_{k,+}(z)}{W(z)} =  \prod_{i=1}^q\eta_{\frak l(p_i)}
    \text{~~~ and ~~~}
\frac{W_{k,-}(z)}{W(z)} =  \eta^{-1}_j\prod_{i=1}^{q}\eta_{\frak l(p_i)}.
\end{equation}
This proves the lemma when $z\in{\bf B}_{jk}$ is near $a_j$.  
Since ${\bf B}_{jk}$ is a branch cut of $W_k$ and since it does not intersect the branch cuts of $W_j$, $[W_k(z)/W_j(z)]_\pm$ is a constant function over the whole ${\bf B}_{jk}$.  
\end{proof}

Let us define the constant $\eta_{kj}$ by
\begin{align}\label{def etakj}
\eta_{kj}:&=\frac{W_j(z)}{W_{k,+}(z)},\quad z\in{\bf B}_{jk},
\end{align}
where $+$ means the boundary value evaluated from the $+$ side of ${\bf B}_{jk}$.

When $z\in D_{a_j}$, by Lemma \ref{za ratio} and Lemma \ref{W ratio} with the definitions of $\widetilde{\eta}_{kj}$ \eqref{def etakjtilde} and $\eta_{kj}$ \eqref{def etakj}, one can see that
\begin{equation}\label{eq: wj wk}
 \widetilde{\eta}_{kj}=\begin{cases}\displaystyle
\eta_{kj}, &\mbox{$0<\arg (a_j/a_k)\leq \pi$},\\\displaystyle
\eta_j\eta_{kj}, &\mbox{$-\pi<\arg (a_j/a_k)\leq 0$}. 
\end{cases} \end{equation}

The jump condition of $\Psi$ is given by the following proposition.

\begin{prop}\label{jump of psi}
Let $\Psi(z)$ be defined in \eqref{eq:psi} and ${\bf W}(z)$ be defined by 
\begin{equation}\label{def bf w}
  {\bf W}(z)={\rm diag}(W_1(z),\dots,W_\nu(z)),  \end{equation}
 we have the following jump conditions of $\Psi(z)$
\begin{equation}\label{eq: jump of Bbar}
\begin{cases}
\Psi_+(z)=\left(I_\nu -(\eta_j-1)\displaystyle \eta_{kj}\,   {\bf e}_{kj}\right) \Psi_-(z),\quad &z\in{\bf B}_{jk}, 
\\
[{\bf W}(z)\Psi(z)]_+=[{\bf W}(z)\Psi(z)]_-,\quad &z\in{\bf B}_j,
\\
[{\bf W}(z)\Psi(z)]_+=\widetilde{J}_{j}\,
[{\bf W}(z)\Psi(z)]_-,\quad &z\in \widehat{\mathbf{B}}_j.
\end{cases}
\end{equation}
\end{prop}
\begin{proof}
Let us find the jump at $z\in{\bf B}_{jk}$. 

We first assume $0<\arg(z/a_k)\leq\pi$.
The line segment $\overrightarrow{a_k z}$ intersects the branch cuts $\{{\bf B}_{\frak r(p_i)}\}_{i=1}^q$, where $p_i=p_i(k;z)$ and $q=q(k;z)$.  Among them will be ${\bf B}_j$ and we will define $1\leq s\leq q$ such that $\frak r(p_s;z)=j$.  From \eqref{62} it follows that  
$$\eta_{kj}=\frac{W_j(z)}{W_{k,+}(z)}=\prod_{i=1}^{s-1}\eta_{{\frak r}(p_i)} .$$

Let $z_+$ be approaching $z$ from the $+$ the side of  ${\bf B}_{jk}$, and $z_-$ from the $-$ side of  ${\bf B}_{jk}$.  
Then $\overrightarrow{a_k z_-}$ intersects all of $\{{\bf B}_{\frak r(p_i)}\}_{i=1}^q$   whereas $\overrightarrow{a_k z_+}$ intersects all but ${\bf B}_{\frak r(p_s)}={\bf B}_j$.   
One can also see that the line segment $\overrightarrow{a_jz}$, which is a subset of $\overrightarrow{a_k z}$,  intersects exactly $\{{\bf B}_{p_i}\}_{i=s+1}^q$ away from $a_j$.

From this observation we obtain the $j$th row and the $k$th row of the matrix ${\bf V}(z)$ \eqref{eq:C}.
\begin{align*}
{\bf V}_-(z) = I_\nu+\sum_{i=1}^q  \Big(\prod_{\xi=1}^{i-1}\eta_{\frak r(p_\xi)} \Big)
(\eta_{\frak r(p_i)}-1){\bf e}_{k \frak r(p_i)}  +\sum_{i=s+1}^q  \Big(\prod_{\xi=s+1}^{i-1}\eta_{\frak r(p_\xi)} \Big)
(\eta_{\frak r(p_i)}-1){\bf e}_{j \frak r(p_i)} + \dots,
\\
{\bf V}_+(z) = I_\nu+\sum_{\substack{i=1\\i\neq s}}^q  \Big(\prod_{\substack{\xi=1 \\ \xi\neq s}}^{i-1}\eta_{\frak r(p_\xi)} \Big)
(\eta_{\frak r(p_i)}-1){\bf e}_{k\frak r(p_i)}  +\sum_{i=s+1}^q  \Big(\prod_{\xi=s+1}^{i-1}\eta_{\frak r(p_\xi)} \Big)
(\eta_{\frak r(p_i)}-1){\bf e}_{j\frak r(p_i)} + \dots.
\end{align*}
The other rows does not change across ${\bf B}_{jk}$.  
We claim
\begin{equation}\label{C+C-}
{\bf V}_+(z)= \bigg[I_\nu-
\Big(\prod_{i=1}^{s-1} \eta_{\frak r(p_i)}\Big)
(\eta_j-1){\bf e}_{kj}\bigg] {\bf V}_-(z).
\end{equation}
It is enough to check the jump on $k$th row.  
First we check the $(k,{\frak r}(p_i))$th entry of \eqref{C+C-} for $i>s$.
\begin{align*}
\big[{\bf V}_-(z)\big]_{k,{\frak r(p_i)}}&=
\Big(\prod_{\xi=1}^{i-1}\eta_{\frak r(p_\xi)} \Big)
(\eta_{\frak r(p_i)}-1)  -
\Big(\prod_{i'=1}^{s-1} \eta_{\frak r(p_{i'})}\Big)
(\eta_j-1) 
\Big(\prod_{\xi=s+1}^{i-1}\eta_{\frak r(p_\xi)} \Big)
(\eta_{\frak r(p_i)}-1)
\\ &=
 \Big(\prod_{\substack{\xi=1 \\ \xi\neq s} }^{i-1}\eta_{\frak r(p_\xi)} \Big)
(\eta_{\frak r(p_i)}-1)=\big[{\bf V}_+(z)\big]_{k,{\frak r(p_i)}},
\end{align*}
where we have used $\eta_j=\eta_{\frak r(p_s)}$.  For $i=s$, the $(k,{\frak r(p_s)})=(k,j)$th entry of \eqref{C+C-} becomes
\begin{align*}
0=\big[{\bf V}_+(z)\big]_{kj}=\big[{\bf V}_-(z)\big]_{kj}-
\Big(\prod_{\xi=1}^{i-1}\eta_{\frak r(p_\xi)} \Big)
(\eta_{\frak r(p_i)}-1).
\end{align*}
A similar and simpler calculation gives the identity for $i<s$.

From Lemma \ref{W ratio} the jump relation \eqref{C+C-} gives the jump relation in this lemma.

Let us repeat the similar proof for $-\pi<\arg(z/a_k)\leq 0$. 
If $z_+$ approaches $z$ from the + side of ${\bf B}_{jk}$, $\overrightarrow{a_k z_+}$ intersects the branch cuts $\{{\bf B}_{\frak l(p_i)}\}_{i=1}^q$.  Whereas approaching from the $-$ side, $\overrightarrow{a_k z_-}$ intersects $\{{\bf B}_{\frak l(p_i)}\}_{i=1}^q$  except ${\bf B}_{\frak l(p_s)}={\bf B}_j$. 
And we obtain the $j$th row and the $k$th row of the matrix ${\bf V}(z)$ \eqref{eq:C}.
\begin{align*}
{\bf V}_+(z) = I_\nu+\sum_{i=1}^q  \Big(\prod_{\xi=1}^{i-1}\eta^{-1}_{\frak l (p_\xi)} \Big)
(\eta^{-1}_{\frak l (p_i)}-1){\bf e}_{k, \frak l (p_i)}  +\sum_{i=s+1}^q  \Big(\prod_{\xi=s+1}^{i-1}\eta^{-1}_{\frak l (p_\xi)} \Big)
(\eta^{-1}_{\frak l (p_i)}-1){\bf e}_{j,\frak l (p_i)} + \dots,
\\
{\bf V}_-(z) = I_\nu+\sum_{\substack{i=1\\i\neq s}}^q  \Big(\prod_{\substack{\xi=1 \\ \xi\neq s}}^{i-1}\eta^{-1}_{\frak l (p_\xi)} \Big)
(\eta^{-1}_{\frak l (p_i)}-1){\bf e}_{k,\frak l  (p_i)}  +\sum_{i=s+1}^q  \Big(\prod_{\xi=s+1}^{i-1}\eta^{-1}_{\frak l (p_\xi)} \Big)
(\eta^{-1}_{\frak l (p_i)}-1){\bf e}_{j,\frak l (p_i)} + \dots.
\end{align*}
This gives the jump relation by
\begin{align}\nonumber
{\bf V}_+(z)
&= \bigg[I_\nu+
\Big(\prod_{i=1}^{s-1} \eta^{-1}_{\frak l (p_i)}\Big)
(\eta^{-1}_j-1){\bf e}_{kj}\bigg] {\bf V}_-(z)
= \bigg[I_\nu+
\eta_{kj}(1-\eta_j){\bf e}_{kj}\bigg] {\bf V}_-(z),
\end{align}
where the last equality is obtained by Lemma \ref{W ratio}.  This ends the proof of the jump on ${\bf B}_{jk}$.

Now let us prove the jump on $z\in{\bf B}_j$.

We claim that the jump relation on ${\bf B}_j$ in Proposition \ref{jump of psi} is equivalent to the following identity using Lemma \ref{W ratio}.
\begin{equation}\label{eq: CJpsi}
{\bf V}_+(z)\widetilde{J}_{j}
 =  \begin{bmatrix}
    I_{j-1}&&\vspace{0.1cm}\\
    &\eta_j^{-1}&\vspace{0.1cm}\\
    &&I_{\nu-j}
    \end{bmatrix} {\bf V}_-(z).    
\end{equation}
It is clear that only the $j$th column of ${\bf V}(z)$ can have discontinuity on ${\bf B}_j$.  So it is enough to look at the $(*,j)$th entries of the above claim.      Since the $j$th row of ${\bf V}(z)$ consists of zeros except 1 at $(j,j)$th entry, the $(j,j)$th entry holds the identity.
    
For $k\neq j$ we consider $0<\arg(z/a_k)\leq\pi$ and $-\pi<\arg(z/a_k)\leq 0$ separately.  For $0<\arg(z/a_k)\leq\pi$, $\overrightarrow{a_k z_+}$ intersects ${\bf B}_j$ while $\overrightarrow{a_k z_-}$ does not intersect ${\bf B}_j$.  Let $q=q(k;z_+)$ and $p_i=p_i(k;z_+)$ such that $\frak r(p_q)=j$.  
The $k$th row of ${\bf V}_\pm(z)$ are given by
\begin{align*}
{\bf V}_+(z) = {\bf e}_{kk} + \sum_{i=1}^{q}
\Big(\prod_{\xi=1}^{i-1} \eta_{\frak r(p_\xi)}\Big)
(\eta_{\frak r(p_i)}-1){\bf e}_{k,\frak r(p_i)},
\\
{\bf V}_-(z) = {\bf e}_{kk} + \sum_{i=1}^{q-1}
\Big(\prod_{\xi=1}^{i-1} \eta_{\frak r(p_\xi)}\Big)
(\eta_{\frak r(p_i)}-1){\bf e}_{k,\frak r(p_i)}.
\end{align*}
The $(k,j)$th entry of the left hand side of \eqref{eq: CJpsi} becomes
\begin{align*}
& (\eta_j^{-1}-1)\bigg[1 + \sum_{i=1}^{q}
\Big(\prod_{\xi=1}^{i-1} \eta_{\frak r(p_\xi)}\Big)
(\eta_{\frak r(p_i)}-1)\bigg] + \Big(\prod_{\xi=1}^{q-1} \eta_{\frak r(p_\xi)}\Big)
(\eta_{\frak r(p_q)}-1)
\\= ~&
(\eta_j^{-1}-1)\bigg[1 + \sum_{i=1}^{q}
\Big(\prod_{\xi=1}^{i-1} \eta_{\frak r(p_\xi)}\Big)
(\eta_{\frak r(p_i)}-1)- \Big(\prod_{\xi=1}^{q} \eta_{\frak r(p_\xi)}\Big)
\bigg]=0,
\end{align*}
where the last equality is obtained by telescoping cancellation.  This is exactly the $(k,j)$th entry of ${\bf V}_-(z)$ because $\overrightarrow{a_k z_-}$ does not intersect ${\bf B}_j$.    

We repeat the same argument for $-\pi<\arg(z/a_k)\leq 0$.  The line segment $\overrightarrow{a_k z_-}$ intersects ${\bf B}_j$ while $\overrightarrow{a_k z_+}$ does not intersect ${\bf B}_j$. Let $q=q(k;z_-)$ and $p_i=p_i(k;z_-)$ such that $\frak l(p_q)=j$. We can write
\begin{align*}
{\bf V}_-(z) = {\bf e}_{kk} + \sum_{i=1}^{q}
\Big(\prod_{\xi=1}^{i-1} \eta^{-1}_{\frak l(p_\xi)}\Big)
(\eta^{-1}_{\frak l(p_i)}-1){\bf e}_{k,\frak l(p_i)},
\\
{\bf V}_+(z) = {\bf e}_{kk} + \sum_{i=1}^{q-1}
\Big(\prod_{\xi=1}^{i-1} \eta^{-1}_{\frak l(p_\xi)}\Big)
(\eta^{-1}_{\frak l(p_i)}-1){\bf e}_{k,\frak l(p_i)}.
\end{align*}
The $(k,j)$th entry of the left hand side of \eqref{eq: CJpsi} becomes
\begin{align}\nonumber
(\eta_j^{-1}-1)\bigg[1 + \sum_{i=1}^{q-1}
\Big(\prod_{\xi=1}^{i-1} \eta^{-1}_{\frak l(p_\xi)}\Big)
(\eta^{-1}_{\frak l(p_i)}-1)\bigg] 
=
(\eta_j^{-1}-1) \Big(\prod_{\xi=1}^{q-1} \eta^{-1}_{\frak l(p_\xi)}\Big),
\end{align}
which is exactly the $(k,j)$th entry of ${\bf V}_-(z)$.    

Since 
$${\bf W}_+(z)^{-1}{\bf W}_-(z)=\sum_{i=1}^\nu \frac{W_{j,-}(z)}{W_{j,+}(z)}{\bf e}_{ii} = \begin{bmatrix}
    I_{j-1}&&\vspace{0.1cm}\\
    &\eta_j^{-1}&\vspace{0.1cm}\\
    &&I_{\nu-j}
    \end{bmatrix},
$$
the identity \eqref{eq: CJpsi} proves the jump relation on ${\bf B}_j$.

Lastly, we claim that the jump relation on $\widehat{\bf B}_j$ in Proposition \ref{jump of psi} is equivalent to the following identity using Lemma \ref{W ratio}.
\begin{equation}\label{eq: CJpsiC}
{\bf V}_+(z)\widetilde{J}_{j}{\bf V}_-(z)^{-1}
 ={\bf W}(z)^{-1}\widetilde{J}_{j}{\bf W}(z).    
\end{equation}

We first evaluate the product of ${\bf V}(z)$ with a matrix with only $j$th column being nonzero as follows:
\begin{align*}
    {\bf V}(z)\bigg( \sum_{i=1}^\nu {\bf e}_{ij}\bigg)&=
    \sum_{k=1}^{\frak s} {\bf e}_{\frak r(k),j}
    \bigg(1+
    \sum_{i=1}^{q({\frak r}(k);\zeta)}
\Big(\prod_{\xi=1}^{i-1} \eta_{{\frak r}(p_\xi)}\Big)
(\eta_{{\frak r}(p_i)}-1)\bigg)
&&\text{where $p_*=p_*({\frak r}(k);z)$}
\\
&\quad + \sum_{k=1}^{\nu-\frak s} {\bf e}_{\frak l(k),j}\bigg(1+
    \sum_{i=1}^{q({\frak l}(k);\zeta)}
\Big(\prod_{\xi=1}^{i-1} \eta^{-1}_{{\frak l}(p_\xi)}\Big)
(\eta^{-1}_{{\frak r}(p_i)}-1)\bigg)
&&\text{where $p_*=p_*({\frak l}(k);z)$}
\\
&=\sum_{k=1}^{\frak s} {\bf e}_{\frak r(k),j}
\Big(\prod_{\xi=1}^{q(\frak r(k))} \eta_{{\frak r}(p_\xi)}\Big)+\sum_{k=1}^{\nu-\frak s} {\bf e}_{\frak l(k),j}
\Big(\prod_{\xi=1}^{q(\frak l(k))} \eta^{-1}_{{\frak l}(p_\xi)}\Big)&&\text{with $p_*$ as above}
\\
&=\sum_{i=1}^\nu \frac{W(z)}{W_i(z)} {\bf e}_{ij}.
    \end{align*}

From above we get
\begin{equation}\nonumber
    {\bf V}(z) \bigg(\sum_{i=1}^\nu {\bf e}_{ij}\bigg){\bf V}(z)^{-1} = \sum_{i=1}^\nu \bigg[\frac{W(z)}{W_j(z)}\bigg] {\bf e}_{ij}= \sum_{i=1}^\nu \frac{W_j(z)}{W_i(z)} {\bf e}_{ij},\qquad \text{$z\in\widehat{\bf B}_j$}.
\end{equation}
The identity \eqref{eq: CJpsiC} is proved because $\widetilde{J}_{j}=I_\nu+(\eta_j^{-1}-1)\sum_{i=1}^\nu{\bf e}_{ij}.$
\end{proof}

\subsection{Large \texorpdfstring{$z$}{z} behavior of \texorpdfstring{$\Psi(z)$}{psiz}}

We will identify the asymptotic behavior of $\Psi(z)$ as $z$ goes to $\infty$.  

\begin{lemma}
Let $\Psi(z)$ and ${\widetilde \Psi}(z)$ be defined in \eqref{eq:psi} and \eqref{eq:psi tilde} respectively, we have 
\begin{equation}\nonumber
{\Psi}(z) = {\bf V}(z){\widetilde {\bf V}}(z){\widehat \Psi}(z)  
\end{equation}
and
\begin{equation}\label{eq:cctilde}
\begin{split}
{\bf V}(z) {\widetilde {\bf V}}(z)=\displaystyle I_\nu+&\sum_{k=1}^{\frak s}\sum_{\substack{j=k+1\\j\notin {\cal P}({\frak r(k)};z)}}^{{\frak s}}
\Big(\prod_{\substack{i=1\\p_{i}<j}}^{q({\frak r(k)};z)} \eta_{\frak r(p_{i})}\Big)
(1-\eta_{\frak r(j)}){\bf e}_{{\frak r(k)},{\frak r(j)}}\quad \text{where $p_*=p_*({\frak r}(k);z)$}\\
+&\displaystyle \sum_{k=1}^{\nu-{\frak s}}\sum_{\substack{j=k+1\\j\notin {\cal P}({\frak l(k)};z)}}^{\nu-{\frak s}}
\Big(\prod_{\substack{i=1\\p_{i}<j}}^{q({\frak l(k)};z)} \eta_{\frak l(p_{i})}^{-1}\Big)
(\eta_{\frak l(j)}-1){\bf e}_{\frak l(k),\frak l(j)}\,\,\,\quad \text{where $p_*=p_*({\frak l}(k);z)$}, \end{split}
\end{equation} where
$ {\cal P}(k;z)=\{p_1(k;z),\dots,p_q(k;z)\}$ and $q=q(k;z).$
\end{lemma}

\begin{proof}
By the definitions of ${\bf V}(z)$ and ${\widetilde {\bf V}}(z)$ in \eqref{eq:C} and \eqref{eq:C tilde}, we get the $({\frak r(k)},{\frak r(j)})$th entry of ${\bf V}(z){\widetilde {\bf V}}(z)$ by the following calculation
\begin{align*}
&\Big[{\bf V}(z) {\widetilde {\bf V}}(z)\Big]_{({\frak r(k)},{\frak r(j)})}=\Big[{\bf V}(z)\Big]_{{\frak r(k)}\text{th row}}\Big[ {\widetilde {\bf V}}(z)\Big]_{{\frak r(j)}\text{th column}}\vspace{0.3cm}\\
&=\displaystyle\bigg[\bigg({\bf e}_{\frak r(k),\frak r(k)}+ \sum_{i=1}^{q({\frak r}(k))}
\Big(\prod_{j=1}^{i-1} \eta_{{\frak r}(p_j)}\Big)
(\eta_{{\frak r}(p_i)}-1){\bf e}_{\frak r(k),{{\frak r}(p_i)}}\bigg)\bigg({\bf e}_{\frak r(j),\frak r(j)}+(1-\eta_{\frak r(j)})\sum_{k=1}^{j-1}{\bf e}_{\frak r(k),\frak r(j)}\bigg)\bigg]_{({\frak r(k)},{\frak r(j)})}\\
&= \begin{cases}
    \displaystyle\bigg(1+\sum_{\substack{i=1\\p_i<j}}^{q({\frak r}(k))}
\Big(\prod_{s=1}^{i-1} \eta_{\frak r(p_{s})}\Big)
(\eta_{\frak r(p_i)}-1)\bigg)
\left(
1-\eta_{\frak r(j)}\right)+
\Big(\prod_{\substack{i=1\\p_{i}<j}}^{q({\frak r}(k))} \eta_{\frak r(p_{i})}\Big)
\left(
\eta_{\frak r(j)}-1\right) &\text{if $j\in{\cal P}({\frak r(k)};z)$}
  \\
   \displaystyle \bigg(1+\sum_{\substack{i=1\\p_i<j}}^{q({\frak r}(k))}
\Big(\prod_{s=1}^{i-1} \eta_{\frak r(p_{s})}\Big)
(\eta_{\frak r(p_i)}-1)\bigg)
\left(
1-\eta_{\frak r(j)}\right)&\text{if $j\notin{\cal P}({\frak r(k)};z)$}
    \end{cases}\\
&= \begin{cases}
    0 &\text{if $j\in{\cal P}({\frak r(k)};z)$}
  \\
 \displaystyle \Big(\prod_{\substack{i=1\\p_{i}<j}}^{q({\frak r}(k))} \eta_{\frak r(p_{i})}\Big)
\left(
1-\eta_{\frak r(j)}\right) &\text{if $j\notin{\cal P}({\frak r(k)};z)$}
    \end{cases}.     
\end{align*}
Similarly, we obtain the  $({\frak l(k)},{\frak l(j)})$th entry of ${\bf V}(z){\widetilde {\bf V}}(z)$
\begin{align*}
&\Big[{\bf V}(z) {\widetilde {\bf V}}(z)\Big]_{({\frak l(k)},{\frak l(j)})}=\Big[{\bf V}(z)\Big]_{{\frak l(k)}\text{th row}}\Big[ {\widetilde {\bf V}}(z)\Big]_{{\frak l(j)}\text{th column}}\\
&=\displaystyle\bigg[\bigg({\bf e}_{\frak l(k),\frak l(k)}+ \sum_{i=1}^{q({\frak l}(k))}
\Big(\prod_{j=1}^{i-1} \eta^{-1}_{{\frak l}(p_j)}\Big)
\left(\eta^{-1}_{{\frak l}(p_i)}-1\right){\bf e}_{\frak l(k),{{\frak l}(p_i)}}\bigg)\bigg(\eta_{\frak l(j)}{\bf e}_{\frak l(j),\frak l(j)}+(\eta_{\frak l(j)}-1)\sum_{k=1}^{j-1} {\bf e}_{\frak l(k),\frak l(j)}\bigg)\bigg]_{({\frak l(k)},{\frak l(j)})}\\
&= \begin{cases}
    \displaystyle\bigg(1+\sum_{\substack{i=1\\p_i<j}}^{q({\frak l}(k))}
\Big(\prod_{s=1}^{i-1} \eta_{\frak l(p_{s})}^{-1}\Big)
(\eta_{\frak l(p_i)}^{-1}-1)\bigg)
\left(
\eta_{\frak l(j)}-1\right)+
\Big(\prod_{\substack{i=1\\p_{i}<j}}^{q({\frak l}(k))} \eta_{\frak l(p_{i})}^{-1}\Big)
\left(
\eta_{\frak l(j)}^{-1}-1\right) &\text{if $j\in{\cal P}({\frak l(k)};z)$}
   \\
   \displaystyle \bigg(1+\sum_{\substack{i=1\\p_i<j}}^{q({\frak l}(k))}
\Big(\prod_{s=1}^{i-1} \eta_{\frak l(p_{s})}^{-1}\Big)
(\eta_{\frak l(p_i)}^{-1}-1)\bigg)
\left(
\eta_{\frak l(j)}-1\right)&\text{if $j\notin{\cal P}({\frak l(k)};z)$}
    \end{cases}\\
&= \begin{cases}
    0 &\text{if $j\in{\cal P}({\frak l(k)};z)$}
   \\
 \displaystyle
\Big(\prod_{\substack{i=1\\p_{i}<j}}^{q({\frak l}(k))} \eta_{\frak l(p_{i})}^{-1}\Big)
\left(
\eta_{\frak l(j)}-1\right) &\text{if $j\notin{\cal P}({\frak l(k)};z)$}
    \end{cases}.    
\end{align*}
\end{proof}

\begin{prop}\label{eq: psi definition}
Let $K>0$ be sufficiently large such that ${\bf B}\cup {\bf B}[j]$ do not intersect in $\{z:|z|>K\}$. For $|z|>K$, we define $\psi_j(z)$ to be the row vector whose $k$th entry is given by
\begin{equation}\label{eq def psi}
[\psi_j(z)]_k =\bigg(\int_{\gamma_j} W_j( s)\prod_{i=1}^\nu{(s-a_i)^{n_i-\delta_{ik}}}\, \ee^{-N \overline z s } \dd s\bigg)^*,    
\end{equation}
where the integration contour $\gamma_j=\{a_j+(z-a_j)t,~t\geq 0\}$ is oriented in the direction of increasing $t$. Then we have
$$[\Psi(z)]_{jk}=[\psi_j(z)]_k,$$
where $\Psi(z)$ is defined at \eqref{eq:psi}.
\end{prop}
\begin{proof}
For $z$ as given in the proposition and $1\leq j <{\frak s(z)}$, we will show that
\begin{equation}\label{eq: psi relation}
   [\Psi(z)]_{{\frak r(j)}\text{th row}}=\psi_{\frak r(j)}(z).
\end{equation}

There exist ${\cal P}({\frak r(j)};z)=\{p_1({\frak r(j)};z),\,p_2({\frak r(j)};z),\dots,p_q({\frak r(j)};z)\}$ and $q=q({\frak r(j)};z)$ such that exactly $\{{\bf B}_{\frak r(p_1)},\,{\bf B}_{\frak r(p_2)},\dots,{\bf B}_{\frak r(p_q)}\}$ among ${\bf B}\setminus\{{\bf B}_{\frak r(j)}\}$ intersects the line segment $\overrightarrow{a_{\frak r(j)}z}$. Then we have
\begin{equation}\nonumber
[\psi_{\mathfrak r(j)}(z)]_k =[\widehat\psi_{\mathfrak r(j)}(z)]_k +\sum_{\substack{i>j\\i\notin{\cal P}({\frak r(j)};z) }}^{\mathfrak s(z)}\bigg(\oint_{{\bf B}_{\frak r(i),\frak r(j)}} W_{\frak r(j)}( s)\prod_{i=1}^\nu{(s-a_i)^{n_i-\delta_{ik}}}\, \ee^{-N \overline z s } \dd s\bigg)^*,
\end{equation} 
where the integration contour is enclosing ${{\bf B}_{\frak r(i),\frak r(j)}}$  in the clockwise orientation. See Figure \ref{fig:sub}.

\begin{figure}   
\centering
\begin{tikzpicture}[scale=0.65]

\draw[->] (5,5) -- (9,9);

\draw[blue, dotted]  (5,5) -- (-4,8);
\draw[blue, dotted]  (5,5) -- (5,7);
\draw[blue, dotted]  (5,5) -- (1,9);

\draw[blue,->]  (-4,8) -- (-10.5,30.5/3);
\draw[blue,->]  (5,7) -- (5,12);
\draw[blue,->]  (1,9) -- (-3,13);

\draw[rotate around={-45:(5,5)},blue,dashed,->] (5,5) parabola (4,11);
\draw[rotate around={0:(5,7)},blue,dashed,->] (5,7) parabola (3.5,12.2);
\draw[rotate around={0:(5,7)},blue,dashed] (5,7) parabola (6.5,12.2);
\draw[rotate around={45:(1,9)},blue,dashed,->] (1,9) parabola (-0.5,14.5);
\draw[rotate around={45:(1,9)},blue,dashed] (1,9) parabola (2.5,14.5);

\foreach \Point/\PointLabel in { (5,5)/a_1, (5,7)/a_2, (1,9)/a_4, (-4,8)/a_6}
\draw[fill=black] \Point circle (0.05) node[below right] {$\PointLabel$};

\foreach \Point/\PointLabel in {(-4.4,8.36)/z}
\draw[fill=red] \Point circle (0.07) node[right] {$\PointLabel$};
\draw[red, dashed,->]  (5,5) -- (-10,10.3617);

\foreach \Point/\PointLabel in {(9,9)/{\bf B}_1, (-10,11)/{z\times\infty}}
\draw[fill=black]  \Point
 node[below right] {$\PointLabel$};
 
\foreach \Point/\PointLabel in {(5,12)/{\bf B}_{21}, (-3,13)/{\bf B}_{41}, (-10,10)/{\bf B}_{61}}
\draw[fill=black]  \Point
 node[below left] {$\PointLabel$};
 
 
 
  
 \end{tikzpicture}
 \caption{The red dashed line is the integration contour of $\psi$ \eqref{eq def psi} and the blue dashed lines are the integration contours of $\widehat\psi$'s \eqref{psi hat}. }\label{fig:sub}
 \end{figure}
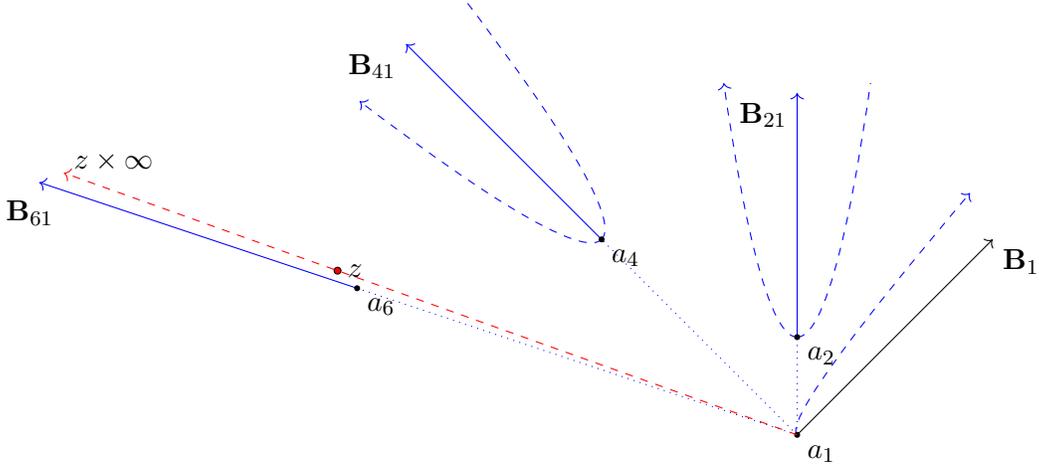

Let us assume that  $\{{\bf B}_{\frak r(i),\frak r(j)}\}_{{j<i\leq {\mathfrak s}(z)},i\notin{\cal P}({\frak r(j)};z) }$ and $\{{\bf B}_{\frak r(i)}\}_{{j<i\leq {\mathfrak s}(z)},i\notin{\cal P}({\frak r(j)};z) }$ are connected to $\infty$ by the same order in the following sense: when $i<i'$
\begin{equation}\label{eq: arg a}
    0<\arg a_{\frak r(i')}-\arg a_{\frak r(i)}<\pi,\quad  0<\arg (a_{\frak r(i')}-a_{\frak r(j)})-\arg (a_{\frak r(i)}-a_{\frak r(j)})<\pi.
\end{equation}
In such case, $\{{\bf B}_{\frak r(i),\frak r(j)}\}_{{j<i\leq {\mathfrak s}(z)},i\notin{\cal P}({\frak r(j)};z) }$ can be smoothly deformed into $\{{\bf B}_{\frak r(i)}\}_{{j<i\leq {\mathfrak s}(z)},i\notin{\cal P}({\frak r(j)};z) }$ without changing the homotopy relation in the branch cuts of $W_{\frak r(j)}(z)$. Let $W_{\frak r(j)}(z)$ change into ${\widetilde W}_{\frak r(j)}(z)$ by this deformation, then we have
\begin{equation}\nonumber
[\psi_{\mathfrak r(j)}(z)]_k =[\widehat\psi_{\mathfrak r(j)}(z)]_k +\sum_{\substack{i>j\\i\notin{\cal P}({\frak r(j)};z) }}^{\mathfrak s(z)}\bigg(\oint_{{\bf B}_{\frak r(i)}} {\widetilde W}_{\frak r(j)}(s)\prod_{i=1}^\nu{(s-a_i)^{n_i-\delta_{ik}}}\, \ee^{-N \overline z s } \dd s\bigg)^*.
\end{equation} 
For a given ${\frak r(i)}$ with $i\notin {\cal P}({\frak r(j)};z)$, we claim
\begin{equation}\label{eq: claim sub}
\bigg(\oint_{{\bf B}_{\frak r(i)}} {\widetilde W}_{\frak r(j)}(s)\prod_{i=1}^\nu{(s-a_i)^{n_i-\delta_{ik}}}\, \ee^{-N \overline z s } \dd s\bigg)^*=\Big(\prod_{\substack{\xi=1\\p_{\xi}<i}}^{q({\frak r(j)})} \eta_{\frak r(p_{\xi})}\Big)\bigg(\oint_{{\bf B}_{\frak r(i)}} {W}(s)\prod_{i=1}^\nu{(s-a_i)^{n_i-\delta_{ik}}}\, \ee^{-N \overline z s } \dd s\bigg)^*.
\end{equation} It follows from the fact that
$$\Big({\widetilde W}_{\frak r(j)}(w)\Big)^*=\Big(\prod_{\substack{\xi=1\\p_{\xi}<i}}^{q({\frak r(j)};z)} \eta_{\frak r(p_{\xi})}\Big)\Big(W(w)\Big)^*,\quad  \text{when $w$ is near $a_{\frak r(i)}$}.$$
This is proven because the line segment $\overrightarrow{a_{\frak r(j)}a_{\frak r(i)}}$ intersects exactly $\{{\bf B}_{\frak r(p_\xi)}\}_{\xi=1,p_\xi<i}^{q({\frak r(j)})}$ among the branch cuts of $W(z)$ and the same line segment does not intersect branch cuts of ${\widetilde W}_{\frak r(j)}(z)$.   If there is $\frak r(i')$ with $j<i'<i$ and $i'\notin {\cal P}({\frak r(j)};z)$ such that ${\bf B}_{\frak r(i')}$ intersects the line segment $\overrightarrow{a_{\frak r(j)}a_{\frak r(i)}}$, then we have 
\begin{equation}\nonumber
    0<\arg a_{\frak r(i)}-\arg a_{\frak r(i')}<\pi,\quad  0<\arg (a_{\frak r(i')}-a_{\frak r(j)})-\arg (a_{\frak r(i)}-a_{\frak r(j)})<\pi,
\end{equation}
which contradicts the assumption \eqref{eq: arg a}.

We also note that ${\widetilde W}_{\frak r(j)}(z)=W(z)$ when $z$ is near ${\bf B}_{\frak r(j)}$. As a consequence of the claim in \eqref{eq: claim sub}, we have
\begin{equation}\label{eq: sub result}
\begin{split}
[\psi_{\mathfrak r(j)}(z)]_k&=\displaystyle[\widehat\psi_{\mathfrak r(j)}(z)]_k +\sum_{\substack{i>j\\i\notin{\cal P}({\frak r(j)};z) }}^{\mathfrak s(z)}\bigg(\oint_{{\bf B}_{\frak r(i),\frak r(j)}} W_{\frak r(j)}( s)\prod_{i=1}^\nu{(s-a_i)^{n_i-\delta_{ik}}}\, \ee^{-N \overline z s } \dd s\bigg)^*\\ &=\displaystyle[\widehat\psi_{\mathfrak r(j)}(z)]_k +\sum_{\substack{i>j\\i\notin{\cal P}({\frak r(j)};z) }}^{\mathfrak s(z)}\bigg(\oint_{{\bf B}_{\frak r(i)}} {\widetilde W}_{\frak r(j)}(s)\prod_{i=1}^\nu{(s-a_i)^{n_i-\delta_{ik}}}\, \ee^{-N \overline z s } \dd s\bigg)^*\\ &=\displaystyle[\widehat\psi_{\mathfrak r(j)}(z)]_k +\sum_{\substack{i>j\\i\notin{\cal P}({\frak r(j)};z) }}^{\mathfrak s(z)}\Big(\prod_{\substack{\xi=1\\p_{\xi}<i}}^{q({\frak r(j)})} \eta_{\frak r(p_{\xi})}\Big)\bigg(\oint_{{\bf B}_{\frak r(i)}} {W}(s)\prod_{i=1}^\nu{(s-a_i)^{n_i-\delta_{ik}}}\, \ee^{-N \overline z s } \dd s\bigg)^*\\ &=\displaystyle[\widehat\psi_{\mathfrak r(j)}(z)]_k +\sum_{\substack{i>j\\i\notin{\cal P}({\frak r(j)};z) }}^{\mathfrak s(z)}\Big(\prod_{\substack{\xi=1\\p_{\xi}<i}}^{q({\frak r(j)})} \eta_{\frak r(p_{\xi})}\Big)(1-\eta_{\frak r(i)})[\widehat\psi_{\mathfrak r(i)}(z)]_k .
\end{split}
\end{equation} 
A similar calculation shows that
\begin{equation}\nonumber
   [\Psi(z)]_{{\frak l(j)}\text{th row}}=\psi_{\frak l(j)}(z).
\end{equation}

We observe that both ${\frak r(j)}$th row of ${\bf V}(z)$ \eqref{eq:C} and the matrix ${\widetilde {\bf V}}(z)$ \eqref{eq:C tilde} do not change under changing the locations of $ a_{\frak r(i)}$'s as long as $\arg a_{\frak r(i)}$'s are all preserved for the corresponding $i$'s and $\{{\bf B}_{\frak r(i)}\}$'s do not intersect $\overrightarrow{a_{\frak r(j)}z}$. 
One can see that any given $a_{\frak r(i)}$'s can be deformed in this way into the distribution described in \eqref{eq: arg a}. Under this deformation, both 
$\widehat\psi$'s \eqref{psi hat} and $\psi$'s \eqref{eq def psi} can be analytically continued in the space of parameters $\{\overline{a}_{\frak r(i)} \}_{i>j,i\notin {\cal P}}$.   Similar argument for $a_{\frak l(*)}$ proves \eqref{eq: sub result} for arbitrary  $a$'s. 
\end{proof}

\begin{prop}\label{poles}
For $j$ and $k$ are from $\{1,\dots,\nu\}$, $[\psi_j(z)]_k$ is analytic away from ${\bf B}\cup{\widehat{\bf B}}\cup {\bf B}[j]$ and the strong asymptotic behavior of $[\psi_j(z)]_k$ is given by
\begin{equation}\label{eq:poles}
[\psi_j(z)]_k=
\displaystyle\frac{C_{jk}\, \ee^{-N z \overline{a}_j}}{(Nz)^{c_j+n_j+1-\delta_{kj}}}\left(1+{\cal O}\left(\frac{1}{z}\right)\right),\quad z\to\infty,
\end{equation}
where 
\begin{equation}\nonumber
    C_{jk}=\lim_{\omega\to a_j}\frac{W_j(\omega)}{(\omega-a_j)^{c_j}} \prod_{i\neq j}^\nu{\left(a_j-a_i\right)^{n_i-\delta_{ik}}}\Gamma\left(c_j+n_j+1-\delta_{kj}\right).
\end{equation} Also $z^{n+\sum c}\psi_j(z)$ is bounded as $z$ goes to the origin. Furthermore, $W_j(z)\psi_j(z)-W_k(z)\psi_k(z)$ is bounded near the origin for $j,k\in\{1,\dots,\nu\}$.
\end{prop}
\begin{proof}
By Proposition \ref{jump of psi}, we have $[\psi_j(z)]_k$ is analytic away from ${\bf B}\cup{\widehat{\bf B}}\cup {\bf B}[j]$. By Proposition \ref{eq: psi definition}, we get
$$[\psi_j(z)]_k =\bigg(\int_{\gamma_j} W_j( s)\prod_{i=1}^\nu{(s-a_i)^{n_i-\delta_{ik}}}\, \ee^{-N \overline z s } \dd s\bigg)^*, $$
where the integration contour $\gamma_j$ is described in \eqref{eq def psi}. Changing the integration variable such that $X=N(s-a_j)(\overline{z}-\overline{a}_j),$ or, equivalently,  $s={a}_j+\frac{X}{N(\overline{z}-\overline{a}_j)},$ we obtain
\begin{equation}\nonumber
\begin{split}
\displaystyle[\psi_j(z)]_k
&=\displaystyle\Bigg(\int_{0}^{\infty} W_j\left(\frac{X}{N(\overline{z}-\overline{a}_j)}+a_j\right)\prod_{i=1}^\nu{\left(\frac{X}{N(\overline{z}-\overline{a}_j)}+a_j-a_i\right)^{n_i-\delta_{ik}}}\frac{\ee^{-N\overline{z}\left(a_j+\frac{X}{N(\overline{z}-\overline{a}_j)}\right)  } \dd X}{N(\overline{z}-\overline{a}_j)}\Bigg)^*\\
&=\displaystyle\bigg(\int_{0}^{\infty}\lim_{\omega\to a_j}\frac{W_j(\omega)}{(\omega-a_j)^{c_j}} \prod_{i\neq j}^\nu{\left(a_j-a_i\right)^{n_i-\delta_{ik}}}\left(\frac{X}{N\overline{z}}\right)^{c_j+n_j-\delta_{kj}}\frac{\ee^{-N \overline{z} a_j-X } \dd X}{N\overline{z}}\bigg)^*\left(1+{\cal O}\left(\frac{1}{z}\right)\right).
\end{split}
\end{equation}
We have \eqref{eq:poles} by $\left(a_j-a_i+\frac{X}{N(\overline{z}-\overline{a}_j)}\right)\sim\left(a_j-a_i\right)$ and $(z-a_j)= z\left(1+{\cal O}\left(1/{z}\right)\right)$ as $z\to\infty.$ 

When $z\to0,$ by the definition of $[\widetilde \psi_j(z)]_k$ in \eqref{chi}, we have
\begin{align*}
[\widetilde \psi_j(z)]_k &=\displaystyle\bigg(\int_{a_j}^{z\times\infty} W( s)\prod_{i=1}^\nu{(s-a_i)^{n_i-\delta_{ik}}}\, \ee^{-N \overline z s } \dd s\bigg)^* \\
&=\displaystyle\bigg(\int_{a_j}^{0} W( s)\prod_{i=1}^\nu{(s-a_i)^{n_i-\delta_{ik}}}\, \ee^{-N \overline z s } \dd s\bigg)^*+\bigg(\int_{0}^{z\times\infty} W( s)\prod_{i=1}^\nu{(s-a_i)^{n_i-\delta_{ik}}}\, \ee^{-N \overline z s } \dd s\bigg)^*,
    \end{align*}
where the 1st term in the 2nd equality is finite. For the 2nd term, we will change the integration variable such that $X=Ns\overline{z},$ we obtain
\begin{equation}\nonumber
\begin{split}
\displaystyle\bigg(\int_{0}^{z\times\infty} W( s)\prod_{i=1}^\nu{(s-a_i)^{n_i-\delta_{ik}}}\, \ee^{-N \overline z s } \dd s\bigg)^*
&=\displaystyle\bigg(\int_{0}^{\infty} W\left( \frac{X}{N\overline{z}}\right)\prod_{i=1}^\nu{\left(\frac{X}{N\overline{z}}-a_i\right)^{n_i-\delta_{ik}}}\frac{\ee^{-X } \dd X}{N\bar{z}}\bigg)^*\\
&=\displaystyle\frac{\Gamma\left(n+\sum c\right)}{(Nz)^{n+\sum c}}\left(1+{\cal O}\left(z\right)\right),
\end{split}
\end{equation}
where we apply $\left(\frac{X}{Nz}-a_i\right)=\frac{X}{Nz}\left(1+{\cal O}\left(z\right)\right),\,\, z\to0$ to the last equality. Therefore, $z^{n+\sum c}[\widetilde \psi_j(z)]_k$ is bounded as $z$ goes to the origin. It follows by $\Psi(z) = {\bf V}(z) \widetilde \Psi(z)$ in \eqref{eq:psi} and ${\bf V}(z)$ \eqref{eq:C} is piecewise constant function, $z^{n+\sum c}\psi_j(z)$ is also bounded as $z$ goes to the origin.

To prove $W_j(z)\psi_j(z)-W_k(z)\psi_k(z)$ is bounded near the origin for $j,k\in\{1,\dots,\nu\}$, it is enough to show that $\widetilde\psi_j(z)-\widetilde\psi_k(z)$ is bounded as $z$ goes to the origin for $j,k\in\{1,\dots,\nu\}$.
By the definition of $[\widetilde \psi_j(z)]_{i'}$ in \eqref{chi}, we have
\begin{equation}\nonumber
\begin{array}{lll}
[\widetilde \psi_j(z)-\widetilde \psi_k(z)]_{i'} &=&\displaystyle\bigg(\int_{a_j}^{a_k} W( s)\prod_{i=1}^\nu{(s-a_i)^{n_i-\delta_{ii'}}}\, \ee^{-N \overline z s } \dd s\bigg)^* ,
\end{array}
    \end{equation}
which is the integral of an entire function over a compact set, this shows that $\widetilde\psi_j(z)-\widetilde\psi_k(z)$ is bounded near the origin for $j,k\in\{1,\dots,\nu\}$.  Expanding $W_{\frak r(j)}(z)\psi_{\frak r(j)}(z)$ in terms of $W(z)\widetilde\psi$'s by \eqref{eq:C} and \eqref{eq:psi}, we observe that
the sum of the linear coefficient of the expansion is independent of ${\frak r(j)}$ and given by
$$\frac{W_{\frak r(j)}(z)}{W(z)}\left(\sum_{i=1}^{q({\frak r}(j))}
\Big(\prod_{j=1}^{i-1} \eta_{{\frak r}(p_j)}\Big)
(\eta_{{\frak r}(p_i)}-1)\right)=\frac{W_{\frak r(j)}(z)}{W(z)}\prod_{i=1}^{q({\frak r}(j))}\eta_{\frak r(p_i)}=1,$$
where we apply Lemma \ref{W ratio} to the last identity. Similarly, we expand $W_{\frak l(j)}(z)[\psi_{\frak l(j)}(z)]_{i'}$ in terms of $W(z)\widetilde\psi$'s, we have
$$\frac{W_{\frak l(j)}(z)}{W(z)}\left(\sum_{i=1}^{q({\frak l}(k))}
\Big(\prod_{j=1}^{i-1} \eta^{-1}_{{\frak l}(p_j)}\Big)
\left(\eta^{-1}_{{\frak l}(p_i)}-1\right)\right)=\frac{W_{\frak l(j)}(z)}{W(z)}\prod_{i=1}^{q({\frak l}(j))}\eta^{-1}_{\frak r(p_i)}=1.$$
Consequently, if one expands $\left[W_j(z)\psi_j(z)-W_k(z)\psi_k(z)\right]_{i'}$ in terms of $W(z)\widetilde\psi$'s the sum of the linear coefficient is zero and, therefore, $\left[W_j(z)\psi_j(z)-W_k(z)\psi_k(z)\right]_{i'}$ can be expressed as the sum of pairwise difference of $W(z)\widetilde\psi$'s. Since the difference of a pair, $[\widetilde\psi_j(z)-\widetilde\psi_k(z)]$, is entire for $j,k\in\{1,\dots,\nu\}$, we have that $W_j(z)\psi_j(z)-W_k(z)\psi_k(z)$ is bounded near the origin for $j,k\in\{1,\dots,\nu\}$.
\end{proof}

Let us define the $\nu\times \nu$ matrix functions $\Psi_0(z)$ and $\Psi_j(z)$ by,
\begin{equation}\label{Psi0Psij} \begin{cases}
\Psi_0(z):={\bf C}\,{\bf E}(z){\bf N}_0(z)\Psi(z),\\
\Psi_j(z):={\bf W}(z)^{-1}\,{\bf C}\,{\bf E}(z){\bf N}_j(z){\bf W}(z)\,\Psi(z),\quad j=1,\dots,\nu,
\end{cases}
\end{equation}
where
\begin{align}\label{const 1}
   &{\bf C}={\rm diag}\bigg(\frac{N^{c_1+n_1}}{C_{11}\ee^{N \ell_1}},\dots,\frac{N^{c_\nu+n_\nu}}{C_{\nu\nu}\ee^{N \ell_\nu}}\bigg),
   \\ 
   &C_{jj}=\lim_{\omega\to a_j}\frac{W_j(\omega)}{(\omega-a_j)^{c_j}} \prod_{i\neq j}^\nu{\left(a_j-a_i\right)^{n_i}}\Gamma\left(c_j+n_j\right), \\
   &{\bf E}(z)={\rm diag}(E_1(z),\dots,E_\nu(z))\quad \mbox{where }
   E_j(z)=\exp\big[N(\overline a_jz+ \ell_j)\big], \\
   \label{def N0}
&{\bf N}_0(z)={\rm diag}(z^{c_1+n_1},\dots,z^{c_\nu+n_\nu}),\\ 
&{\bf N}_j(z)=\left[{\begin{array}{c:c:c}
\begin{matrix}
\vspace{-0.4cm}\\ I_{j-1}
\end{matrix}
&\begin{matrix}\vspace{-0.4cm}\\ -1\\\vdots\\-1\end{matrix} &\begin{matrix}\vspace{-0.4cm}\\ {\bf
0}\end{matrix}\vspace{0.1cm}\\\hdashline
\begin{matrix}\vspace{-0.4cm}\\{\bf 0}\end{matrix}
&\begin{matrix}\vspace{-0.4cm}\\ [z^{n+\sum c}]_{{\bf B}[j]}\end{matrix}
&\begin{matrix}\vspace{-0.4cm}\\ {\bf 0}\end{matrix}\vspace{0.1cm}\\\hdashline
\begin{matrix}\vspace{-0.4cm}\\ {\bf 0}\end{matrix}
&\begin{matrix}\vspace{-0.4cm}\\ -1\\\vdots\\-1\end{matrix} &\begin{matrix}
\vspace{-0.4cm}\\ I_{\nu-j}
\end{matrix}\end{array}}\right].    \end{align}
We remind that $[z^{n+\sum c}]_{{\bf B}[j]}$ has the branch cuts on ${\bf B}[j]\cup \widehat {\bf B}$, see the definition in \eqref{eq24}.  

We also note that
\begin{equation} \label{zWzW}
    \frac{z^{\sum c}}{W(z)}=\frac{[z^{\sum c}]_{{\bf B}[j]}}{W_j(z)}.
\end{equation}

\begin{lemma} For $z\notin {\bf B}\cup{\bf B}[j]\cup {\widehat{\bf B}}$, $\{\Psi_j(z)\},j\in\{0,1,\dots,\nu\}$ satisfy the following properties,
\begin{itemize}
\setlength\itemsep{-0.2em}
    \item[(i)] $\Psi_j(z)$ for $j\neq 0$ is bounded near the origin;
    \item[(ii)] $\Psi_0(z)=I_\nu+{\cal O}\left(1/z\right),\quad z\to\infty;$
    \item[(iii)] $\det \Psi_0(z) = 1,\quad |\det \Psi_j(z)|=1, \quad j=1,\dots,\nu;$
    \item[(iv)] $\Psi_j(z)$ is analytic away from ${\bf B}\cup{\bf B}[j]\cup {\widehat{\bf B}}$.
\end{itemize}
\end{lemma}
\begin{proof}
Analyticity of $\Psi_j(z)$ follows from Proposition \ref{jump of psi}. $\Psi_j(z), j\neq 0$ is bounded near the origin is because of the statements, $z^{n+\sum c}\psi_j(z)$ is bounded as $z$ goes to the origin and $W_j(z)\psi_j(z)-W_k(z)\psi_k(z)$ is bounded near the origin for $j,k\in\{1,\dots,\nu\}$,  in Proposition \ref{poles}. The strong asymptotics of $\Psi_0(z)$ is due to \eqref{eq:poles} in Proposition \ref{poles}.

Finally we prove $(iii)$. Since
 \begin{equation}\nonumber
 \begin{split}
 \partial_z [\widetilde \psi_j(z)]_k&=\displaystyle\bigg(\int_{{a}_j}^{z\times\infty}(-Ns)\prod_{i=1}^\nu(s-a_i)^{n_i+c_i-\delta_{ik}}e^{-N\bar{z}s}\dd s\bigg)^*\vspace{0.3cm}\\\displaystyle
 &=\displaystyle-\bigg(\int_{a_j}^{z\times\infty}\prod_{i=1}^\nu(s-a_i)^{n_i+c_i}e^{-N\bar{z}s}ds\bigg)^*-N\bar{a}_k[\widetilde\psi_j(z)]_k\vspace{0.3cm}\\
&=\displaystyle-\frac{\sum_{i=1}^\nu (n_i+c_i)[\widetilde\psi_j(z)]_k}{Nz}-N\overline{a}_k[\widetilde\psi_j(z)]_k,
 \end{split}
 \end{equation}
we have
\begin{equation}\nonumber
\begin{array}{lll}
\partial_z\widetilde\Psi(z)=-\widetilde\Psi(z)\left(\frac{1}{Nz}\begin{bmatrix}
n_1+c_1&\dots&n_1+c_1\vspace{0.2cm}\\
\vdots&\ddots&\vdots\vspace{0.2cm}\\
n_\nu+c_\nu&\dots&n_\nu+c_\nu
\end{bmatrix}+\begin{bmatrix}
N{\overline a}_1&&\vspace{0.2cm}\\
&\ddots&\vspace{0.2cm}\\
&&N{\overline a}_\nu
\end{bmatrix}\right).
\end{array}
\end{equation}
Therefore,
$$\partial_z\det\widetilde\Psi(z)=\det\widetilde\Psi(z)\left(-N\sum_{j=1}^\nu{\overline a}_j-\frac{n+\sum c}{Nz}\right).$$
Solving the differential equation, we get
$$\det\widetilde\Psi(z)=\displaystyle \frac{{\text Const.}\ee^{-N\left(\sum_{j=1}^\nu{\overline a}_j\right)z}}{z^{n+\sum c}},$$ where the constant term is a constant function in a connected region of $\CC\setminus\{ {\bf B}\cup {\widehat{\bf B}}\}.$
By the definition of $\Psi_0(z)$ with $\det {\bf V}(z)=1$ and the asymptotics of $\Psi_0$ in (ii), we have $\det \Psi_0 (z) = 1$. 
One compares $\Psi_0(z)$ and $\Psi_j(z)$, they must have the same determinant upto a phase factor.  The phase factor is due to the branch cuts chosen for $z^{n+\sum c}$.
\end{proof}

\section{Transformations of Riemann-Hilbert problem}\label{steepest}

In the previous section we have constructed a $\nu\times\nu$ matrix $\Psi$ \eqref{eq:psi}.  Now we deal with the full $(\nu+1)\times(\nu+1)$ matrices and we adopt the following notations.  We will use the index from $\{0,1,\dots,\nu\}$ to count the entries of the matrices such that, for $(\nu+1)\times(\nu+1)$ matrix $M$, $[M]_{jk}$ refers to the entry in the $(j+1)$th row and the $(k+1)$th column.  We prefer such numbering because our matrices are structured such that the 1st row and the 1st column play a distinct role than the other rows and columns.  By this new convention, from now on, the $j$th row in any $(\nu+1)\times(\nu+1)$ matrix will refer to the $(j+1)$th row in the old convention.

We apply the method of nonlinear steepest descent analysis \cite{DKMVZ 1999} and define successive transformations of $Y$ into $\widetilde Y$ \eqref{tildeY}, $T$ \eqref{T} and $S$ \eqref{def S}.  We will finish the section by defining the global parametrix $\Phi$ \eqref{def global phi}.   

\subsection{\texorpdfstring{$\widetilde Y$}{widetildeY} transform }

Let us redefine the Riemann-Hilbert problem for $Y$ in Theorem \ref{thm 3.2}, with the deformation of jump contours as described below. 

\noindent{\bf Riemann-Hilbert problem for $Y$:}
\begin{equation}
\begin{cases}
Y_+(z)=Y_-(z)\left[{\begin{array}{c:c}
\begin{matrix}
1
\end{matrix}
&\begin{matrix}W(z)\widetilde\psi_1(z)\end{matrix} \\\hdashline
\begin{matrix}
{\bf 0}
\end{matrix}
&\begin{matrix}I_{\nu}\end{matrix}
\end{array}}\right],& z\in\bigcup_{j} \Gamma_{j0},
\\
Y_+(z)=Y_-(z)\left[{\begin{array}{c:c}
\begin{matrix}
1
\end{matrix}
&\begin{matrix}W(z)\widetilde\psi_1(z)\end{matrix} \\\hdashline
\begin{matrix}
{\bf 0}
\end{matrix}
&\begin{matrix}I_{\nu}\end{matrix}
\end{array}}\right]_-^{-1}
\left[{\begin{array}{c:c}
\begin{matrix}
1
\end{matrix}
&\begin{matrix}W(z)\widetilde\psi_1(z)\end{matrix} \\\hdashline
\begin{matrix}
{\bf 0}
\end{matrix}
&\begin{matrix}I_{\nu}\end{matrix}
\end{array}}\right]_+
,& z\in  {\bf B}\cap (\Omega_0)^c,
\\
Y(z)=\displaystyle\left(I_{\nu+1}+{\cal O}\left(\frac{1}{z}\right)\right)\begin{bmatrix}
z^{n}&&&\\
&z^{-n_1}&&\\
&&\ddots&\\
&&&z^{-n_\nu}
\end{bmatrix},& z\to\infty,
\\
Y(z)= {\cal O}(1), & z\to a_j.
\end{cases}
\end{equation}
Here $\widetilde\psi_1$ is defined at \eqref{psi tilde}.

The above Riemann-Hilbert problem is by deforming the jump contour $\gamma$ in Theorem \ref{thm 3.2}, such that the resulting contour is along the boundary of the domain $({\rm clos\,}\Omega_0)^c \setminus {\bf B}$. See Figure \ref{figure contour}.  When the contour goes around the branch cut ${\bf B}$ \eqref{def B} the jump matrix can be expressed as in the second equation by the product of the jumps that come from either sides of the branch cut.   
The subscripts $\pm$ of the jump matrices on ${\bf B}\cap (\Omega_0)^c$ stand for the boundary values evaluated from the $\pm$ sides of ${\bf B}$ respectively.   
\bigskip

\begin{figure}
\begin{center}
\includegraphics[width=0.3\textwidth]{MSzego0.pdf}
\includegraphics[width=0.3\textwidth]{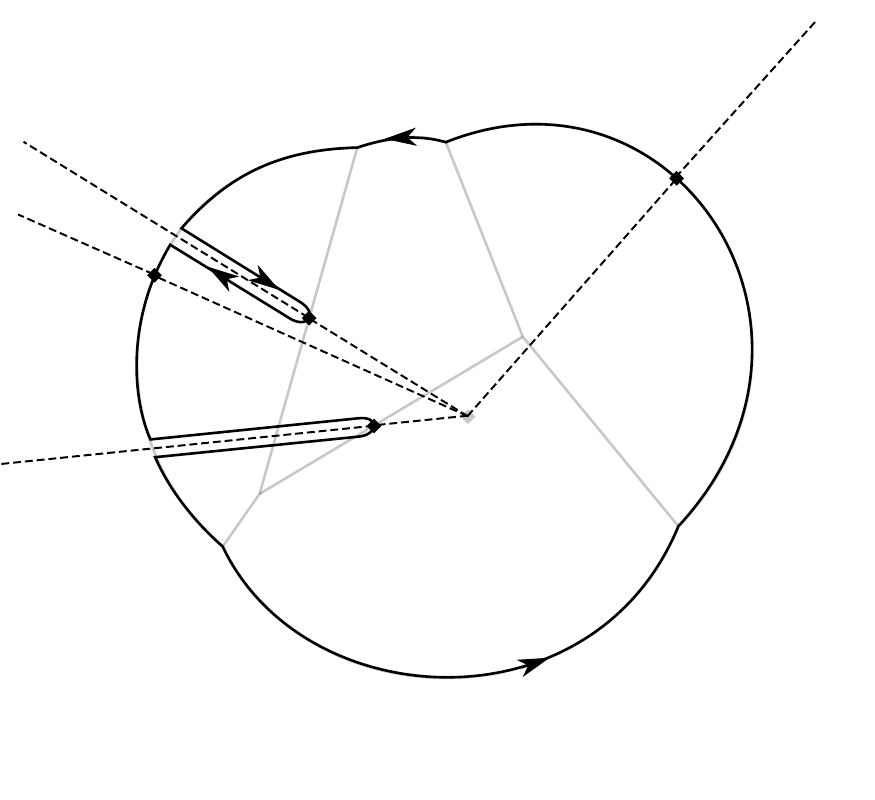}
\includegraphics[width=0.3\textwidth]{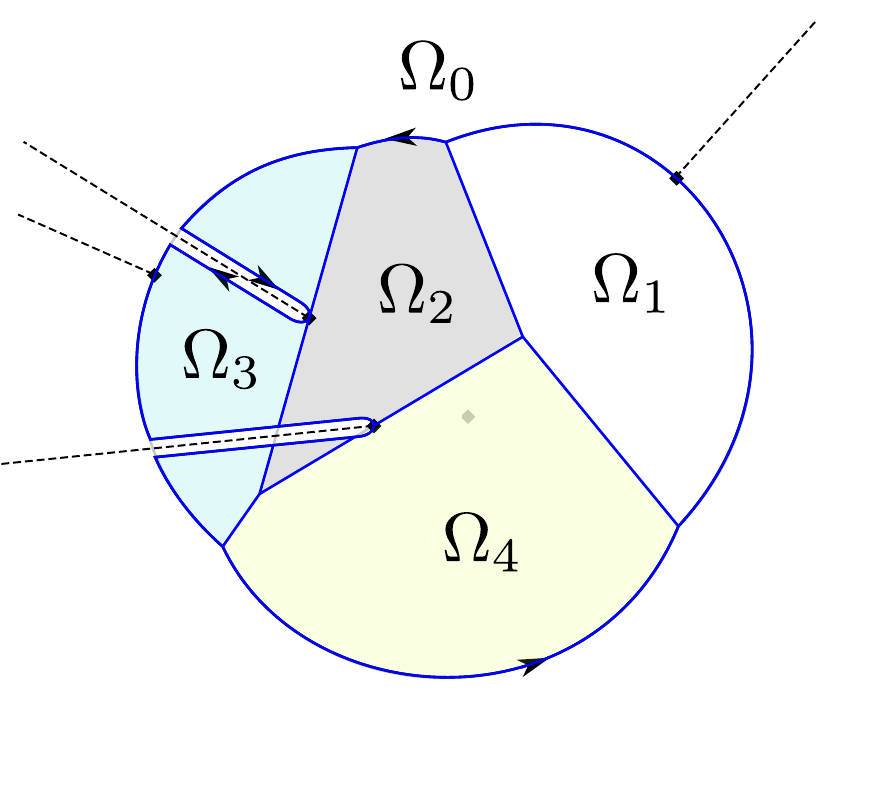}
\end{center}
 \caption{The jump contour of $Y$ for $\nu=4$ (thick line in the middle figure). The dotted lines in the middle figure are branch cuts ${\bf B}$ \eqref{def B} and $\widehat{\bf B}$ \eqref{def bhat}. The contour goes around the branch cuts ${\bf B}_2$ and ${\bf B}_4$.
} \label{figure contour}
\end{figure}

\begin{lemma}
${\bf B}_{kj}$ does not intersect $\Omega_j$.
\end{lemma}
\begin{proof}
Since $a_k\in{\rm clos\,} \Omega_k$ we have ${\rm Re}(\overline{a}_kz)+l_k\geq {\rm Re}(\overline{a}_jz)+l_j$ at $z=a_k$.  Since ${\rm Re}(\overline{a}_kz)-{\rm Re}(\overline{a}_jz)$ increases as one moves from $a_k$ to $\infty$ along ${\bf B}_{kj}$ \eqref{def bk}
we have that ${\rm Re}(\overline{a}_kz)+l_k\geq {\rm Re}(\overline{a}_jz)+l_j$ for $z\in{\bf B}_{kj}$.
\end{proof} 

By Proposition \ref{jump of psi} the $j$th row of $\Psi$ \eqref{eq:psi} has nontrivial jump only on ${\bf B}_{*j}$, where $*$ stands for all possible numbers from $\{1,\dots,\nu\}$. Therefore
$W_j\psi_j$ is analytic away from ${\bf B}_{*j}$ and $\widehat{\bf B}$ \eqref{def bhat}.  Since ${\bf B}_{*j}$ does not intersect $\Omega_j$ by the previous lemma $W_j\psi_j$ is analytic on $\Omega_j\setminus\widehat{\bf B}$.

Then $W_j(z)\psi_j(z)-W(z)\widetilde\psi_1(z)$ is analytic in $\Omega_j\setminus({\bf B}\cup\widehat{\bf B})$.

Let us define $\widetilde Y(z)$ by
\begin{equation}\label{tildeY}
\begin{cases}
\widetilde Y(z) = Y(z), & z\in\Omega_0, \\
    \widetilde Y(z) = Y(z)  \left[{\begin{array}{c:c}
\begin{matrix}
1
\end{matrix}
&\begin{matrix}W_j(z)\psi_j(z)-W(z)\widetilde\psi_1(z)\end{matrix} \\\hdashline
\begin{matrix}
{\bf 0}
\end{matrix}
&\begin{matrix}I_{\nu}\end{matrix}
\end{array}}\right], &z\in \Omega_j\setminus({\bf B}\cup \widehat{\bf B}),
\end{cases}
\end{equation}
where $j=1,\dots,\nu.$

\noindent{\bf Riemann-Hilbert problem for $\widetilde Y$:}
\begin{equation}\label{jump of T tilde}
\begin{cases}
\widetilde Y_+(z)=\widetilde Y_-(z)\left[{\begin{array}{c:c}
\begin{matrix}
1
\end{matrix}
&\begin{matrix}W_k(z)\psi_k(z)\end{matrix} \\\hdashline
\begin{matrix}
{\bf 0}
\end{matrix}
&\begin{matrix}I_{\nu}\end{matrix}
\end{array}}\right],& z\in\Gamma_{k0},
\\
\widetilde Y_+(z)=\widetilde Y_-(z)\left[{\begin{array}{c:c}
\begin{matrix}
1
\end{matrix}
&\begin{matrix}W_k(z)\psi_k(z)-W_j(z)\psi_j(z)\end{matrix} \\\hdashline
\begin{matrix}
{\bf 0}
\end{matrix}
&\begin{matrix}I_{\nu}\end{matrix}
\end{array}}\right],&z\in\Gamma_{kj},
\\
\widetilde Y_+(z)= \widetilde Y_-(z)
\left[{\begin{array}{c:c}
\begin{matrix}
1
\end{matrix}
&\begin{matrix}\left(W_j(z)\psi_j(z)-W(z)\widetilde \psi_1(z)\right)\Big|_-^+\end{matrix} \\\hdashline
\begin{matrix}
{\bf 0}
\end{matrix}
&\begin{matrix}I_{\nu}\end{matrix}
\end{array}}\right]
,& z\in \widehat{\bf B}\cap \Omega_j,
\\
\widetilde Y(z)=\left(I_{\nu+1}+{\cal O}\left(\displaystyle\frac{1}{z}\right)\right)\begin{bmatrix}
z^{n}&&&\\
&z^{-n_1}&&\\
&&\ddots&\\
&&&z^{-n_\nu}
\end{bmatrix},& z\to\infty,
\\
\widetilde Y(z)={\cal O}(1), & z\to a_j,
\end{cases}
\end{equation}
where $j\neq k$ and $1\leq j,k\leq\nu.$
One can check that the jump on ${\bf B}$ is absent because $W_j\psi_j$ is analytic on ${\bf B}$ by Proposition \ref{jump of psi}.

Using the fact that $\widetilde \psi_j(z)-\widetilde \psi_1(z)$ is analytic everywhere for any $j$, the jump on $\widehat{\bf B}\cap\Omega_j$ can be written by
\begin{equation}\label{YYBhat}
    \widetilde Y_{+}(z)= \widetilde Y_{-}(z)
\left[{\begin{array}{c:c}
\begin{matrix}
1
\end{matrix}
&\begin{matrix} W_j(z)\big(\psi_{j,+}(z)-\psi_{j,-}(z)\big)-W(z)\big(\widetilde \psi_{j,+}(z)-\widetilde \psi_{j,-}(z)\big) \end{matrix} \\\hdashline
\begin{matrix}
{\bf 0}
\end{matrix}
&\begin{matrix}I_{\nu}\end{matrix}
\end{array}}\right]
,\quad z\in \widehat{\bf B}\cap\Omega_j.
\end{equation}

Let us define
\begin{equation}
   \widetilde Y_0(z)= \begin{cases}
 \widetilde Y(z) ,&z\in \Omega_0,\\
\widetilde Y(z) \left[{\begin{array}{c:c}
\begin{matrix}
1
\end{matrix}
&\begin{matrix}-W_k(z)\psi_k(z)\end{matrix} \\\hdashline
\begin{matrix}
{\bf 0}
\end{matrix}
&\begin{matrix}I_{\nu}\end{matrix}
\end{array}}\right],&z\in \Omega_k \mbox { for $1\leq k\leq\nu$.}
\end{cases}
\end{equation}

By the jump condition of $\widetilde Y$  \eqref{jump of T tilde} the only jump of $\widetilde Y_0$ is at $\widehat{\bf B}$.

\subsection{T transform}
Let us define
\begin{equation}
\begin{cases}
G_0(z)={\rm diag}\left(z^{-n},z^{n_1},\dots,z^{n_\nu}\right),\\
G_j(z)=\begin{bmatrix}
E_j(z)^{-1}&&&\\
&I_{j-1}&&\\
&&E_j(z)&\\
&&&I_{\nu-j}
\end{bmatrix}\quad \mbox { for $j\neq 0$.}
\end{cases}
\end{equation}
Define $T(z)$ by
\begin{equation}\label{T}
T(z)=\begin{bmatrix}
1&{\bf 0}\\
{\bf 0}&{\bf C}^{-1}
\end{bmatrix}\widetilde Y(z)\begin{bmatrix}
1&{\bf 0}\\
{\bf 0}&\Psi_j(z)^{-1}
\end{bmatrix}\begin{bmatrix}
1&{\bf 0}\\
{\bf 0}&{\bf C}
\end{bmatrix}G_j(z),\qquad  z\in \Omega_j,
\end{equation}
for $j=0,1,\dots,\nu$, where we use $\Psi_0$ and $\Psi_j$ in \eqref{Psi0Psij}. 

We note that $G_0$ and $G_j$ are made of the exponents of those functions that appeared in the definitions of the multiple Szeg\"o curve, which corresponds to the support of the limiting roots of $p_n$.  The transform is to separate the leading exponential behavior of $Y$ as $n\to\infty$, and it corresponds to the so called $g$-function transform.

The jump of $T$ on $\Gamma_{j0}$ is given by
\begin{align}\label{Jump of T on Gamma}
T_{-}(z)^{-1}T_{+}(z)=G_0(z)^{-1}
\begin{bmatrix}
1&{\bf 0}\\
{\bf 0}&{\bf C}^{-1}
\end{bmatrix}\begin{bmatrix}
1&{\bf 0}\\
{\bf 0}&\Psi_0(z)
\end{bmatrix}\left[{\begin{array}{c:c}
\begin{matrix}
1
\end{matrix}
&\begin{matrix}W_j(z)\psi_j(z)\end{matrix} \\\hdashline
\begin{matrix}
{\bf 0}
\end{matrix}
&\begin{matrix}I_{\nu}\end{matrix}
\end{array}}\right]\begin{bmatrix}
1&{\bf 0}\\
{\bf 0}&\Psi_j(z)^{-1}
\end{bmatrix}\begin{bmatrix}
1&{\bf 0}\\
{\bf 0}&{\bf C}
\end{bmatrix}
G_j(z).
\end{align}
Using the fact that $\psi_j(z) \Psi(z)^{-1}$ is the $\nu$ dimensional row vector ${\bf e}_j:=(\dots,0,1,0,\dots)$ with only nonvanishing entry being at the $j$th entry, the above jump of $T$ on $\Gamma_{j0}$ is given by
\begin{equation}\label{jum of t gammaj0}
M_{j0}(z)=\left[\setstretch{1.55}{\begin{array}{c:c:c:c}
\begin{matrix}
\frac{z^n}{E_j(z)}
\end{matrix}
&\begin{matrix}{\bf 0}\end{matrix} &\begin{matrix}\frac{W_j(z)}{z^{\sum c}}\end{matrix}&\begin{matrix}{\bf
0}\end{matrix}\vspace{0.1cm}\\\hdashline
\begin{matrix}
{\bf 0}
\end{matrix}
&\begin{matrix}z^{c_1}&&\\
&\ddots&\\
&&z^{c_{j-1}}\end{matrix}
&\begin{matrix}\frac{E_1(z)z^{c_1}W_j(z)}{z^{n+\sum c}W_1(z)}\\\vdots\\\frac{E_{j-1}(z)z^{c_{j-1}}W_j(z)}{z^{n+\sum c}W_{j-1}(z)}\end{matrix}&\begin{matrix}{\bf
0}\end{matrix}\vspace{0.1cm}\\\hdashline
\begin{matrix}0\end{matrix}
&\begin{matrix}{\bf 0}\end{matrix} &\begin{matrix}
\frac{E_j(z)z^{c_j}}{z^{n+\sum c}}
\end{matrix}&\begin{matrix}{\bf
0}\end{matrix}\vspace{0.1cm}\\\hdashline
\begin{matrix}{\bf 0}\end{matrix}
&\begin{matrix}{\bf 0}\end{matrix} &\begin{matrix}
\frac{E_{j+1}(z)z^{c_{j+1}}W_j(z)}{z^{n+\sum c}W_{j+1}(z)}\\\vdots\\\frac{E_{\nu}(z)z^{c_{\nu}}W_j(z)}{z^{n+\sum c}W_\nu(z)}
\end{matrix}&\begin{matrix}z^{c_{j+1}}&&\\
&\ddots&\\
&&z^{c_{\nu}}\end{matrix}\end{array}}\right].
\end{equation}

The following decomposition will be useful.
\begin{equation}
M_{j0}(z)= M_0(z)^{-1}J_j(z)M_j(z), \quad j=1,\dots,\nu,
\end{equation}
where
$$M_0(z)=\left[\setstretch{1.55}{\begin{array}{c:c}
\begin{matrix}
1
\end{matrix}
&\begin{matrix}{\bf 0}\end{matrix} \\\hdashline
\begin{matrix}-\frac{z^{c_1}}{W_1(z)}\frac{E_1(z)}{z^n}\\
\vdots\\-\frac{z^{c_\nu}}{W_\nu(z)}\frac{E_\nu(z)}{z^n}\end{matrix}
&\begin{matrix}I_\nu\end{matrix} \end{array}}\right],\quad M_j(z)= \left[\setstretch{1.55}{\begin{array}{c:c:c:c}
\begin{matrix}
1
\end{matrix}
&\begin{matrix}{\bf 0}\end{matrix} &\begin{matrix}0\end{matrix}&\begin{matrix}{\bf
0}\end{matrix}\\\hdashline
\begin{matrix}
\frac{-E_1(z)}{E_j(z)W_1(z)}\\
\vdots\\
\frac{-E_{j-1}(z)}{E_j(z)W_{j-1}(z)}
\end{matrix}
&\begin{matrix}I_{j-1}\end{matrix}
&\begin{matrix}{\bf 0}\end{matrix}&\begin{matrix}{\bf
0}\end{matrix}\vspace{0.1cm}\\\hdashline
\begin{matrix}\frac{z^{n+\sum c}}{E_j(z)W_j(z)}\end{matrix}
&\begin{matrix}{\bf 0}\end{matrix} &\begin{matrix}
1
\end{matrix}&\begin{matrix}{\bf
0}\end{matrix}\vspace{0.1cm}\\\hdashline
\begin{matrix}\frac{-E_{j+1}(z)}{E_j(z)W_{j+1}(z)}\\
\vdots\\
\frac{-E_{\nu}(z)}{E_j(z)W_{\nu}(z)}\end{matrix}
&\begin{matrix}{\bf 0}\end{matrix} &\begin{matrix}
{\bf 0}
\end{matrix}&\begin{matrix}I_{\nu-j}\end{matrix}\end{array}}\right]$$
and
$$J_j(z)=\left[\setstretch{1.55}{\begin{array}{c:c:c:c}
\begin{matrix}
0
\end{matrix}
&\begin{matrix}{\bf 0}\end{matrix} &\begin{matrix}\frac{W_j(z)}{z^{\sum c}}\end{matrix}&\begin{matrix}{\bf
0}\end{matrix}\vspace{0.1cm}\\\hdashline
\begin{matrix}
{\bf 0}
\end{matrix}
&\begin{matrix}z^{c_1}&&\\
&\ddots&\\
&&z^{c_{j-1}}\end{matrix}
&\begin{matrix}{\bf 0}\end{matrix}&\begin{matrix}{\bf
0}\end{matrix}\\\hdashline
\begin{matrix}\frac{-z^{c_j}}{W_j(z)}\end{matrix}
&\begin{matrix}{\bf 0}\end{matrix} &\begin{matrix}
0
\end{matrix}&\begin{matrix}{\bf
0}\end{matrix}\vspace{0.1cm}\\\hdashline
\begin{matrix}{\bf 0}\end{matrix}
&\begin{matrix}{\bf 0}\end{matrix} &\begin{matrix}
{\bf 0}
\end{matrix}&\begin{matrix}z^{c_{j+1}}&&\\
&\ddots&\\
&&z^{c_{\nu}}\end{matrix}\end{array}}\right].$$
We set that $z^{c_j}$ has the branch cut on ${\bf B}_j$ and the branch cut of $z^{\sum c}$ comes from the factorization $\prod_{j=1}^\nu z^{c_j}$.  One can check that $M_0$ has no branch cut on ${\bf B}$.

Let us define $T_0(z)$ by
\begin{equation}\label{T0}
T_0(z)=
\begin{bmatrix}
1&{\bf 0}\\
{\bf 0}&{\bf C}^{-1}
\end{bmatrix}\widetilde Y_0(z)\begin{bmatrix}
1&{\bf 0}\\
{\bf 0}&\Psi_0(z)^{-1}
\end{bmatrix}\begin{bmatrix}
1&{\bf 0}\\
{\bf 0}&{\bf C}
\end{bmatrix}G_0(z).
\end{equation}
We have $$ T(z)=T_0(z),\quad z\in \Omega_0 $$ and for $j\neq 0$ we have
\begin{equation}\label{eq87}
   T(z)=T_0(z)M_{j0}(z),\quad z\in \Omega_j
\end{equation}
by the similar calculation as in the jump of $T$ on $\Gamma_{j0}$ \eqref{Jump of T on Gamma}.

Then the jump of $T$ on $\Gamma_{jk}$ is given in terms of the above definitions by
\begin{equation}\label{jum of t gammajk}
    T_{-}(z)^{-1}T_{+}(z)= M_{k0}(z)^{-1}T_0(z)^{-1} T_0(z)M_{j0}(z)=M_k(z)^{-1} J_k(z)^{-1}J_j(z)M_j(z).
\end{equation}

The jump of $T$ on ${\mathbf{B}}_{jk}\cap \Omega_0$ is given by
\begin{equation}\label{eq:TDD}
\begin{split}
T_{-}(z)^{-1}T_{+}(z)&=G_0(z)^{-1}
\begin{bmatrix}
1&{\bf 0}\\
{\bf 0}&{\bf C}^{-1}
\end{bmatrix}\begin{bmatrix}
1&{\bf 0}\\
{\bf 0}&\Psi_0(z)
\end{bmatrix}_-\begin{bmatrix}
1&{\bf 0}\\
{\bf 0}&\Psi_0(z)^{-1}
\end{bmatrix}_+\begin{bmatrix}
1&{\bf 0}\\
{\bf 0}&{\bf C}
\end{bmatrix}
G_0(z)\\
&={\bf D}(z)\left(I_{\nu+1} +\eta_{kj}(\eta_j-1)  {\bf e}_{kj}\right){\bf D}(z)^{-1},
    \end{split}
\end{equation}
where $${\bf D}(z)=G_0(z)^{-1}\begin{bmatrix}
1&{\bf 0}\\
{\bf 0}&{\bf E}(z)
\end{bmatrix}\begin{bmatrix}
1&{\bf 0}\\
{\bf 0}&{\bf N}_0(z)
\end{bmatrix}= {\rm diag}\big(z^{n},E_1(z) z^{c_1},\dots,E_\nu(z) z^{c_\nu}\big).$$

\begin{lemma}\label{jump invar}
We have
\begin{equation}
    M_{0,-}(z)T_{-}(z)^{-1}T_{+}(z)M_{0,+}(z)^{-1}=T_{-}(z)^{-1}T_{+}(z),\quad z\in {\mathbf{B}}_{jk}\cap \Omega_0.
\end{equation}
\end{lemma}
\begin{proof}
On ${\bf B}_{jk}$ $W_k(z)$ \eqref{def wk} has nontrivial jump and we get
\begin{equation}
    \begin{split}
    &{\bf D}(z)^{-1}M_{0,-}(z){\bf D}(z)\left(I_{\nu+1} +\eta_{kj}(\eta_j-1) {\bf e}_{kj}\right){\bf D}(z)^{-1}M_{0,+}(z)^{-1}{\bf D}(z)\\
    &=\left[{\begin{array}{c:c}
\begin{matrix}
1
\end{matrix}
&\begin{matrix}{\bf 0}\end{matrix} \\\hdashline
\begin{matrix}\vspace{-0.4cm}\\\frac{1}{W_{1}(z)}\\
\vdots\\\frac{1}{W_{\nu}(z)}\end{matrix}
&\begin{matrix}\vspace{-0.4cm}\\I_\nu\end{matrix} \end{array}}\right]_-^{-1}\left(I_{\nu+1} +\eta_{kj}(\eta_j-1)  {\bf e}_{kj}\right)\left[{\begin{array}{c:c}
\begin{matrix}
1
\end{matrix}
&\begin{matrix}{\bf 0}\end{matrix} \\\hdashline
\begin{matrix}\vspace{-0.4cm}\\\frac{1}{W_{1}(z)}\\
\vdots\\\frac{1}{W_{\nu}(z)}\end{matrix}
&\begin{matrix}\vspace{-0.4cm}\\ I_\nu\end{matrix} \end{array}}\right]_+\\
&=I_{\nu+1} +\eta_{kj}(\eta_j-1)  {\bf e}_{kj}+\left(\frac{\eta_j-1}{W_{k,+}(z)}+\frac{1}{W_{k,+}(z)}-\frac{1}{W_{k,-}(z)}\right){\bf e}_{k0}\\
&=I_{\nu+1} +\eta_{kj}(\eta_j-1){\bf e}_{kj}.
    \end{split}
\end{equation}
The last equality is obtained by $W_{k,+}(z)=W_{k,-}(z)\eta_j$ for $z\in{\mathbf{B}}_{jk}.$  Conjugating by ${\bf D}(z)$ and using \eqref{eq:TDD} the proof is complete.
\end{proof}

Since $T(z)=T_0(z)M_{i0}(z)$ for $z\in\Omega_i$ when $i\neq 0$ \eqref{eq87}  the jump of $T$ on ${\mathbf{B}}_{jk}\cap \Omega_i$ is given by
\begin{equation}\label{jump of T in i}
\begin{split}
T_{-}(z)^{-1}T_{+}(z)&=\left(T_0(z)M_{i0}(z)\right)_-^{-1}\left(T_0(z)M_{i0}(z)\right)_+\\
&=\left[M_{0}(z)^{-1}J_{i}(z)M_{i}(z)\right]_-^{-1}T_{0,-}(z)^{-1}T_{0,+}(z)\left[M_{0}(z)^{-1}J_{i}(z)M_{i}(z)\right]_+
\\
&=M_{i,-}(z)^{-1}J_{i,-}(z)^{-1}T_{0,-}(z)^{-1}T_{0,+}(z)J_{i,+}(z)M_{i,+}(z)\\
&=\left({\bf D}(z)^{-1}J_{i}(z)M_{i}(z)\right)_-^{-1}\left(I_{\nu+1} +\eta_{kj}(\eta_j-1) {\bf e}_{kj}\right)\left({\bf D}(z)^{-1}J_{i}(z)M_{i}(z)\right)_+.
\end{split}
\end{equation}
The 3rd equality is obtained by Lemma \ref{jump invar}.

The jump of $T$ on ${\mathbf{B}}_{j}\cap \Omega_0$ is given by
\begin{equation}\label{jump of T in 0}
\begin{split}
T_{-}(z)^{-1}T_{+}(z)
&=G_0(z)^{-1}
\begin{bmatrix}
1&{\bf 0}\\
{\bf 0}&{\bf C}^{-1}
\end{bmatrix}\begin{bmatrix}
1&{\bf 0}\\
{\bf 0}&\Psi_0(z)
\end{bmatrix}_-\begin{bmatrix}
1&{\bf 0}\\
{\bf 0}&\Psi_0(z)^{-1}
\end{bmatrix}_+\begin{bmatrix}
1&{\bf 0}\\
{\bf 0}&{\bf C}
\end{bmatrix}
G_0(z)=I_{\nu+1}.
    \end{split}
\end{equation}
Here we have used that $\Psi_0(z)$ \eqref{Psi0Psij} does not jump on ${\bf B}$.  It can be seen by $\Psi_-(z)\Psi_+(z)^{-1}={\bf W}(z)_-^{-1}{\bf W}(z)_+$ from Proposition \ref{jump of psi} and that ${\bf N}_0(z)$ \eqref{def N0} satisfies the same jump as ${\bf W}(z)$ \eqref{def bf w} on ${\bf B}$.

Using again that $\Psi_0(z)$ does not jump on ${\bf B}$ one can also see that $T_0$ \eqref{T0} does not jump on ${\bf B}$ since $\widetilde Y_0$ does not jump on ${\bf B}$.  The jump of $T$ on ${\mathbf{B}}_{j}\cap \Omega_k$ for $k\neq 0$ is given by
\begin{equation}\label{jump of T in k}
\begin{split}
T_{-}(z)^{-1}T_{+}(z)&=\left(T_0(z)M_{k0}(z)\right)_-^{-1}\left(T_0(z)M_{k0}(z)\right)_+
\\&=M_{k0,-}(z)^{-1}M_{k0,+}(z)\\
&=M_{k,-}(z)^{-1}J_{k,-}(z)^{-1}M_{0,-}(z)M_{0,+}(z)^{-1}J_{k,+}(z)M_{k,+}(z)\\
&=M_{k,-}(z)^{-1}J_{k,-}(z)^{-1}J_{k,+}(z)M_{k,+}(z)\\
&=M_{k,-}(z)^{-1}\left(I_{\nu+1}+(\eta_j-1){\bf e}_{jj}\right)M_{k,+}(z).
\end{split}
\end{equation}
The 4th equality is obtained by the fact that $M_0$ has no jump on ${\bf B}$.
\bigskip

\begin{lemma}\label{jump NJN}
When $z\in \widehat{\mathbf{B}}_j$, we have
\begin{equation}\nonumber
{\bf N}_{k,-}(z)\,\widetilde{J}_{j}^{-1}\,{\bf N}_{k,+}(z)^{-1}=\left\{\begin{aligned}
&I_\nu,& k=j,\\
&I_\nu+(\eta_j-1)[z^{n+\sum c}]_-{\bf e}_{kj},&k\neq j.
\end{aligned}\right.
\end{equation}
Here $[z^{n+\sum c}]_-$ stands for the boundary value at the $-$ side of $\widehat{\bf B}_j$.

\end{lemma}

By Lemma \ref{jump NJN}, the jump of  $T$ on $\widehat{\mathbf{B}}_j\cap \Omega_k$ for $k\neq0$ is given by
\begin{align}\nonumber
&T_{-}(z)^{-1}T_{+}(z)\\\nonumber
=& G_k(z)^{-1}
\begin{bmatrix}
1&{\bf 0}\\
{\bf 0}&{\bf C}^{-1}
\end{bmatrix}\begin{bmatrix}
1&{\bf 0}\\
{\bf 0}&\Psi_{k,-}(z)
\end{bmatrix}\begin{bmatrix}
1&{\bf 0}\\
{\bf 0}&\Psi_{k,+}^{-1}(z)
\end{bmatrix}\begin{bmatrix}
1&{\bf 0}\\
{\bf 0}&{\bf C}
\end{bmatrix}
G_k(z)\\\nonumber
&+ G_k(z)^{-1}
\begin{bmatrix}
1&{\bf 0}\\
{\bf 0}&{\bf W}(z)^{-1}{\bf E}(z){\bf N}_k(z)
\end{bmatrix}_-\begin{bmatrix}
1&{\bf 0}\\
{\bf 0}&{\bf W}(z)\Psi(z)
\end{bmatrix}_-\\\nonumber
&\times
\begin{bmatrix}
1&{\bf 0}\\
{\bf 0}&{\bf W}(z)\Psi(z)
\end{bmatrix}_{+}^{-1}\begin{bmatrix}
1&{\bf 0}\\
{\bf 0}&{\bf N}_k(z)^{-1}{\bf E}(z)^{-1}{\bf W}(z)
\end{bmatrix}_+
G_k(z)\\\nonumber
=&G_k(z)^{-1}
\begin{bmatrix}
1&{\bf 0}\\
{\bf 0}&{\bf W}(z)^{-1}{\bf E}(z){\bf N}_k(z)
\end{bmatrix}_-\begin{bmatrix}
1&{\bf 0}\\
{\bf 0}&\widetilde{J}_{j}
\end{bmatrix}^{-1}\begin{bmatrix}
1&{\bf 0}\\
{\bf 0}&{\bf N}_k(z)^{-1}{\bf E}(z)^{-1}{\bf W}(z)
\end{bmatrix}_+
G_k(z)\\\nonumber
=&\begin{cases}
G_k(z)^{-1}
\begin{bmatrix}
1&{\bf 0}\\
{\bf 0}&{\bf W}(z)^{-1}{\bf E}(z)
\end{bmatrix}_-\begin{bmatrix}
1&{\bf 0}\\
{\bf 0}&{\bf E}(z)^{-1}{\bf W}(z)
\end{bmatrix}_+
G_k(z),&k=j,\\
G_k(z)^{-1}
\begin{bmatrix}
1&{\bf 0}\\
{\bf 0}&{\bf W}(z)^{-1}{\bf E}(z)
\end{bmatrix}_-\left(I_{\nu+1}+(\eta_j-1)[z^{n+\sum c}]_-{\bf e}_{kj}\right)\begin{bmatrix}
1&{\bf 0}\\
{\bf 0}&{\bf E}(z)^{-1}{\bf W}(z)
\end{bmatrix}_+
G_k(z),&k\neq j,
\end{cases}
\\
=&
\begin{cases}
I_{\nu+1},\qquad &k=j,
\\ \displaystyle
I_{\nu+1}+(\eta_j-1)[z^{n+\sum c}]_-\frac{W_j(z)}{E_j(z) W_k(z)} {\bf e}_{kj},\qquad &k\neq j.
\end{cases}\label{jump of T in kk}
\end{align}
In the 1st equality, we have used the fact that
$\widetilde Y_-^{-1}\widetilde Y_+$ \eqref{YYBhat} does not contribute to the jump because of the following computation.  When $z\in \widehat{\bf B}_j$ we have 
\begin{equation}\nonumber
    \begin{split}
   & \big[W_j(z)\big(\psi_{j,+}(z)-\psi_{j,-}(z)\big)-W(z)\big(\widetilde \psi_{j,+}(z)-\widetilde \psi_{j,-}(z)\big)\big] \Psi_+(z)^{-1}  
    \\
  = & W_j(z) \big[\Psi_-(z)\Psi_+(z)^{-1}-I_\nu\big]_\text{$j$th row} - W(z) \big[{\bf V}(z)^{-1} \Psi_-(z)\Psi_+(z)^{-1}-{\bf V}(z)^{-1}\big]_\text{$j$th row} 
    \\
  =&W_j(z) \big[{\bf W}(z)^{-1} \widetilde{J}_{j}^{-1}\,{\bf W}(z)-I_\nu\big]_\text{$j$th row} - W(z) \big[{\bf V}(z)^{-1} {\bf W}(z)^{-1} \widetilde{J}_{j}^{-1}\,{\bf W}(z) -{\bf V}(z)^{-1}\big]_\text{$j$th row}
        \\
   =&\big[ \widetilde{J}_{j}^{-1}\,{\bf W}(z)-I_\nu\big]_\text{$j$th row} - W(z) \big[{\bf V}(z)^{-1} {\bf V}(z) \widetilde{J}_{j}^{-1}\,{\bf V}(z)^{-1} -{\bf V}(z)^{-1}\big]_\text{$j$th row}
        \\
   =&\big[ \widetilde{J}_{j}^{-1}\,{\bf W}(z)-{\bf W}(z)\big]_\text{$j$th row} - W(z) \big[(\widetilde{J}_{j}^{-1}-I_\nu){\bf V}(z)^{-1}\big]_\text{$j$th row} 
    , \quad z\in \widehat{\bf B}_j.
\end{split}
\end{equation}
Since the $j$th row of $\widetilde{J}_{j}^{-1} = \eta_j {\bf e}_j$ where ${\bf e}_j$ is the row basis vector whose only nonzero entry being $1$ at the $j$th entry, the above becomes
\begin{align}
 W_j(z)    (\eta_j-1){\bf e}_j - W(z) (\eta_j-1) [{\bf V}(z)^{-1}]_\text{$j$th row} 
\end{align}
which vanishes because $[{\bf V}(z)^{-1}]_\text{$j$th row}={\bf e}_j$. 

Collecting all the jumps that we obtained in \eqref{Jump of T on Gamma},\eqref{jum of t gammajk},\eqref{eq:TDD},\eqref{jump of T in i},\eqref{jump of T in 0},\eqref{jump of T in k} and \eqref{jump of T in kk}, we have the following Riemann-Hilbert problem.

\bigskip
\noindent{\bf Riemann-Hilbert problem for $T$:}
\begin{equation}
\begin{cases}
T_+(z)=T_-(z)M_0(z)^{-1}J_j(z)M_j(z), &z\in{\Gamma}_{j0},\\
T_+(z)=T_-(z) M_k(z)^{-1}J_k(z)^{-1}J_j(z)M_j(z), &z\in{\Gamma}_{jk},\\
T_+(z)=T_-(z),  &z\in {\mathbf{B}}_j\cap \Omega_0,\quad z\in \widehat{\mathbf{B}}_j\cap \Omega_j,\\
T_+(z)=T_-(z)M_{k,-}(z)^{-1}\left(I_{\nu+1}+(\eta_j-1){\bf e}_{jj}\right)M_{k,+}(z),  &z\in {\mathbf{B}}_j\cap \Omega_i,\quad i\neq 0,\\
\displaystyle T_+(z)=T_-(z)\left(I_{\nu+1}+(\eta_j-1)[z^{n+\sum c}]_-\frac{W_j(z)}{E_j(z) W_k(z)} {\bf e}_{kj}\right),  &z\in \widehat{\mathbf{B}}_j\cap \Omega_k,\quad k\neq 0,\vspace{0.1cm}\\
\displaystyle T_+(z)=T_-(z){\bf D}(z)\left(I_{\nu+1} +\eta_{kj}(\eta_j-1)  {\bf e}_{kj}\right){\bf D}(z)^{-1}, &z\in {\mathbf{B}}_{jk}\cap \Omega_0,
\vspace{0.1cm}\\
\displaystyle T_+(z)=T_-(z)\left({\bf D}(z)^{-1}J_{i}(z)M_{i}(z)\right)_-^{-1}\\
\qquad\qquad\times\left(I_{\nu+1} +\eta_{kj}(\eta_j-1) {\bf e}_{kj}\right)\left({\bf D}(z)^{-1}J_{i}(z)M_{i}(z)\right)_+,  &z\in {\mathbf{B}}_{jk}\cap \Omega_i,\quad i\neq 0,
\vspace{0.1cm}\\
T(z)=I_{\nu+1}+{\cal O}\left(1/z\right), &z\to\infty,
\vspace{0.1cm}\\
T(z)={\cal O}(1), & z\to a_j,
\end{cases}
\end{equation}
where $j\neq k$ and $1\leq j,k\leq\nu.$

\subsection{S transform: lens opening}
Let us define $S(z)$ by 
\begin{equation}\label{def S}
    S(z)=\begin{cases}
       T(z),\quad &\text{ when } z\in \Omega_0 \setminus U,\\
     T(z)M_0(z)^{-1},\quad & \text{ when } z\in \Omega_0 \cap U, \\
    T(z)M_j(z)^{-1},\quad & \text{ when } z\in\Omega_j \text{ for } j=1,\dots,\nu,
    \end{cases}
\end{equation}
where $U\subset {\mathbb D}$ is a fixed neighborhood of $\bigcup_{j=1}^\nu\Gamma_{j0}$.

By the definition of $S$ in \eqref{def S} the jump of $S$ on $\partial U\cap\Omega_0$ is given by
\begin{equation}
S_{-}(z)^{-1}S_{+}(z)=M_0(z)^{-1}.
\end{equation}

By the jump of $T$ on $\Gamma_{j0}$ in \eqref{jum of t gammaj0} we have the jump of $S$ on $\Gamma_{j0}$ is given by
\begin{equation}
S_{-}(z)^{-1}S_{+}(z)=M_0(z)M_{j0}(z)M_j(z)^{-1}=J_j(z).
\end{equation}

By the jump of $T$ on $\Gamma_{jk}$ in \eqref{jum of t gammajk} we have the jump of $S$ on $\Gamma_{jk}$ is given by
\begin{equation}
S_{-}(z)^{-1}S_{+}(z)=M_k(z)M_{k0}(z)^{-1}M_{j0}(z)M_j(z)^{-1}=J_k(z)^{-1}J_j(z).
\end{equation}

By Lemma \ref{jump invar} the jump of $S$ on ${\mathbf{B}}_{jk}\cap \Omega_0$ is given by
\begin{equation}
\begin{split}
S_{-}(z)^{-1}S_{+}(z)&=M_0(z)T_{-}(z)^{-1}T_{+}(z)M_0(z)^{-1}=T_{-}(z)^{-1}T_{+}(z)\\
&={\bf D}(z)\left(I_{\nu+1} +\eta_{kj}(\eta_j-1)  {\bf e}_{kj}\right){\bf D}(z)^{-1}\\
&=\left(I_{\nu+1}+
\eta_{kj}(\eta_j-1)\frac{E_k(z)z^{c_k}}{E_j(z)z^{c_j}}{\bf e}_{kj}\right).
\end{split}
\end{equation}

By the jump of $T$ on ${\mathbf{B}}_{jk}\cap \Omega_i$ in \eqref{jump of T in i} we have the jump of $S$ on ${\mathbf{B}}_{jk}\cap \Omega_i$ is given by
\begin{equation}
\begin{split}
S_{-}(z)^{-1}S_{+}(z)&=M_{i,-}(z)T_{-}(z)^{-1}T_{+}(z)M_{i,+}(z)^{-1}\\
&=\left({\bf D}(z)^{-1}J_{i}(z)\right)_-^{-1}\left(I_{\nu+1} +\eta_{kj}(\eta_j-1) {\bf e}_{kj}\right)\left({\bf D}(z)^{-1}J_{i}(z)\right)_+\\
&=\begin{cases}
\displaystyle I_{\nu+1}+
\eta_{kj}(\eta_j-1)\frac{E_k(z)}{E_j(z)W_j(z)}{\bf e}_{k0},&i=j, \\
\displaystyle I_{\nu+1}+
\eta_{kj}(\eta_j-1)\frac{E_k(z)}{E_j(z)}{\bf e}_{kj},&i\neq j.
\end{cases}
\end{split}
\end{equation}

By the jump of $T$ in ${\mathbf{B}}_{j}\cap \Omega_0$ in \eqref{jump of T in 0} the jump of $S$ on ${\mathbf{B}}_{j}\cap \Omega_0$ is given by
\begin{equation}
S_{-}(z)^{-1}S_{+}(z)=M_{0,-}(z)T_{-}(z)^{-1}T_{+}(z)M_{0,+}(z)^{-1}=M_{0,-}(z)M_{0,+}(z)^{-1}=I_{\nu+1}.
\end{equation}

By the jump of $T$ in ${\mathbf{B}}_{j}\cap \Omega_k$ in \eqref{jump of T in k} the jump of $S$ on ${\mathbf{B}}_{j}\cap \Omega_k$ is given by
\begin{equation}
S_{-}(z)^{-1}S_{+}(z)=M_{k,-}(z)T_{-}(z)^{-1}T_{+}(z)M_{k,+}(z)^{-1}= I_{\nu+1}+(\eta_j-1){\bf e}_{jj}.
\end{equation}

By the jump of $T$ in $\widehat{\mathbf{B}}_j\cap \Omega_k$ in \eqref{jump of T in kk} the jump of $S$ on $\widehat{\mathbf{B}}_j\cap \Omega_k$ is given by
\begin{equation}
\begin{split}
S_{-}(z)^{-1}S_{+}(z)
&=M_{k,-}(z)T_{-}(z)^{-1}T_{+}(z)M_{k,+}(z)^{-1}\\
&=\begin{cases}
M_{k,-}(z)M_{k,+}(z)^{-1},&k=j, \\
\displaystyle M_{k,-}(z)\bigg(I_{\nu+1}+(\eta_j-1)[z^{n+\sum c}]_-\frac{W_j(z)}{E_j(z) W_k(z)} {\bf e}_{kj}\bigg)M_{k,+}(z)^{-1},&k\neq j,
\end{cases}\\
&=\begin{cases}
\displaystyle I_{\nu+1}-   \frac{\left(\eta_j-1\right)[z^{n+\sum c}]_-}{E_j(z)W_j(z)}{\bf e}_{j0},&k=j, \\
\displaystyle I_{\nu+1}+\left(\eta_j-1\right)\frac{W_j(z)}{W_k(z)} \frac{[z^{n+\sum c}]_-}{E_j(z)}{\bf e}_{kj},&k\neq j.
\end{cases}
\end{split}
\end{equation}
Collecting all the jumps of $S$ we have the following Riemann-Hilbert problem.

\bigskip
\noindent{\bf Riemann-Hilbert problem for $S$:}
\begin{equation}\label{srhp}
\begin{cases}
S_+(z)=S_-(z)J_{j}(z), &z\in{\Gamma}_{j0},\\
S_+(z)=S_-(z) J_{k}(z)^{-1}J_{j}(z), &z\in{\Gamma}_{jk},\\
S_{+}(z)=S_{-}(z)M_0(z)^{-1},&z\in\partial U\cap\Omega_0,\\
S_+(z)=S_-(z),  &z\in {\mathbf{B}}_j\cap \Omega_0,\quad z\in {\mathbf{B}}_j\cap \Omega_j,\\
S_+(z)=S_-(z)\left(I_{\nu+1}+(\eta_j-1){\bf e}_{jj}\right),  &z\in {\mathbf{B}}_j\cap \Omega_k,\\
\displaystyle S_+(z)=S_-(z)\left(I_{\nu+1}-\frac{\left(\eta_j-1\right)[z^{n+\sum c}]_-}{E_j(z)W_j(z)}{\bf e}_{j0}\right),  &z\in \widehat{\mathbf{B}}_j\cap \Omega_j,\vspace{0.1cm}\\
\displaystyle S_+(z)=S_-(z)\left(I_{\nu+1}+\left(\eta_j-1\right)\frac{W_j(z)}{W_k(z)}\frac{[z^{n+\sum c}]_-}{E_j(z)}{\bf e}_{kj}\right),  &z\in \widehat{\mathbf{B}}_j\cap \Omega_k,\vspace{0.1cm}\\
\displaystyle S_+(z)=S_-(z)\left(I_{\nu+1}+
\eta_{kj}(\eta_j-1)\frac{E_k(z)z^{c_k}}{E_j(z)z^{c_j}}{\bf e}_{kj}\right), &z\in {\mathbf{B}}_{jk}\cap \Omega_0,
\vspace{0.1cm}\\
\displaystyle S_+(z)=S_-(z)\left(I_{\nu+1}-
\eta_{kj}(\eta_j-1)\frac{E_k(z)}{E_j(z)W_j(z)}{\bf e}_{k0}\right),  &z\in {\mathbf{B}}_{jk}\cap \Omega_j,
\vspace{0.1cm}\\
\displaystyle S_+(z)=S_-(z)\left(I_{\nu+1}+
\eta_{kj}(\eta_j-1)\frac{E_k(z)}{E_j(z)}{\bf e}_{kj}\right), &z\in \mathbf{B}_{jk}\cap \Omega_i,\quad i\neq j,k,
\vspace{0.1cm}\\
S(z)=I_{\nu+1}+{\cal O}\left(1/z\right), &z\to\infty,
\vspace{0.1cm}\\
S(z)M_0(z)={\cal O}(1),\,\, &\text{as $z\to a_j$ in $\Omega_0$},
\vspace{0.1cm}\\
S(z)M_j(z)={\cal O}(1),\,\, &\text{as $z\to a_k$ or $z\to a_j$ in $\Omega_j$},
\end{cases}
\end{equation}
where $j\neq k$ and $1\leq j,k\leq\nu.$

We note that the jump conditions on $\widehat {\bf B}_j$ and on ${\bf B}_{jk}$ are all exponentially small as $N$ grows away from the points $\{a_1,\dots,a_\nu\}$ because $z^n/E_j(z)=\exp(N(\log z-\overline{a}_jz-\ell_j))$ is exponentially small on $\widehat {\bf B}_j$ and $E_k(z)/E_j(z) = \exp(N((\overline{a}_j-\overline{a}_j)z+\ell_j-\ell_k))$ is exponentially small on ${\bf B}_{jk}$.

\subsection{Global Parametrix}\label{global}

We set up the model Riemann-Hilbert problem of $\Phi(z)$ from that of $S$ by ignoring the jump matrices that are exponentially small as $N\to\infty$.

\begin{equation}\label{mrhp}
\begin{cases}
\Phi_+(z)=\Phi_-(z)J_{j}(z),& z\in{\Gamma}_{j0},\\
\Phi_+(z)=\Phi_-(z)J_{k}(z)^{-1}J_{j}(z),& z\in{\Gamma}_{jk},\\
\Phi_+(z)=\Phi_-(z)\left(I_{\nu+1}+(\eta_j-1){\bf e}_{jj}\right),& z\in \mathbf{B}_j\cap \Omega_k,~~ k\neq j,\\
\Phi(z)=I_{\nu+1}+{\cal O}(1/z),&z\to\infty,
\end{cases}
\end{equation}
where $1\leq j,k\leq\nu.$

A solution of the model Riemann-Hilbert problem is given by
\begin{equation}\label{def global phi}
\Phi(z)=
\begin{cases}
\displaystyle{\rm diag}\left(\frac{z^{\sum c}}{W(z)}, \left(\frac{z-a_1}{z}\right)^{c_1},\dots,\left(\frac{z-a_\nu}{z}\right)^{c_\nu}\right)
, & z\in\Omega_0,\\
\displaystyle\left[\setstretch{1.55}{\begin{array}{c:c:c:c}
\begin{matrix}
0
\end{matrix}
&\begin{matrix}{\bf 0}\end{matrix} & 1 &\begin{matrix}{\bf
0}\end{matrix}\\\hdashline
\begin{matrix}
{\bf 0}
\end{matrix}
&\begin{matrix}(z-a_1)^{c_1}&&\\
&\ddots&\\
&&(z-a_{j-1})^{c_{j-1}}\end{matrix}
&\begin{matrix}{\bf 0}\end{matrix}&\begin{matrix}{\bf
0}\end{matrix}\\\hdashline
\begin{matrix}\frac{-(z-a_j)^{c_j}}{W_j(z)}\end{matrix}
&\begin{matrix}{\bf 0}\end{matrix} &\begin{matrix}
0
\end{matrix}&\begin{matrix}{\bf
0}\end{matrix}\\\hdashline
\begin{matrix}{\bf 0}\end{matrix}
&\begin{matrix}{\bf 0}\end{matrix} &\begin{matrix}
{\bf 0}
\end{matrix}&\begin{matrix}(z-a_{j+1})^{c_{j+1}}&&\\
&\ddots&\\
&&(z-a_\nu)^{c_{\nu}}\end{matrix}\end{array}}\right],& z\in\Omega_j,
\end{cases}
\end{equation}
where $j=1,\dots,\nu$. We assign the branch cut of $(z-a_j)^{c_j}$ on ${\bf B}_j$ for each $j=1,\dots,\nu$.
One can check the jump on $\Gamma_{j0}$ by \eqref{zWzW}.
We note that $z^{\sum c}/W(z)$ is analytic away from $\widehat{\bf B}$.

We obtain the following jump relations.
\begin{align}\label{jump of sphi1}
[S(z)\Phi(z)^{-1}]_+&=[S(z)\Phi(z)^{-1}]_-\left(I_{\nu+1}+\frac{\left(\eta_j-1\right)[z^{n+\sum c}]_-}{E_j(z)(z-a_j)^{c_j}}{\bf e}_{0j}\right),  &&z\in \widehat{\mathbf{B}}_j\cap \Omega_j,\\\label{jump of sphi2}
[S(z)\Phi(z)^{-1}]_+&=[S(z)\Phi(z)^{-1}]_-\left(I_{\nu+1}+\frac{W_j(z)}{W_k(z)}\frac{\left(\eta_j-1\right)[z^{n+\sum c}]_-}{E_j(z)(z-a_j)^{c_j}}{\bf e}_{0j}\right), && z\in \widehat{\mathbf{B}}_j\cap \Omega_k,\\\label{jump of sphi3}
[S(z)\Phi(z)^{-1}]_+&=[S(z)\Phi(z)^{-1}]_-\left(I_{\nu+1}+
\eta_{kj}(\eta_j-1)\frac{E_k(z)(z-a_k)^{c_k}}{E_j(z)(z-a_j)^{c_j}}{\bf e}_{kj}\right), && z\in {\mathbf{B}}_{jk},\\
[S(z)\Phi(z)^{-1}]_+&=[S(z)\Phi(z)^{-1}]_-\left(I_{\nu+1}+
\sum_{i=1}^\nu\frac{E_i(z)(z-a_i)^{c_i}}{z^{n+\sum c}}{\bf e}_{i0}\right), && z\in \partial U\cap\Omega_0,\label{jump of sphi4}
\end{align} where $j\neq k$.

\section{Local Parametrices}\label{section local}

Near $a_j$'s the jump matrices of $\Phi$ \eqref{mrhp} do not converge to the jump matrices of $S$ \eqref{srhp}.   We therefore need the local parametrix around $a_j$ that satisfies the exact jump condition of $S$.  
In Section 5.1 and 5.2 we construct the local parametrices separately when $a_j\in \Gamma_{j0}$ and when $a_j\in \Gamma_{jk}$ for $k\neq 0$. 
In Section 5.3 we construct a rational matrix function ${\cal R}$ such that the improved global parametrix, ${\cal R}\Psi$, matches the local parametrix better. 

\subsection{\texorpdfstring{$a_j\in \Gamma_{j0}$}{ajgammaj0}}\label{subsection j0}

Let $D_{a_j}$ be a disk neighborhood of $a_j$ with a fixed radius $r$ such that the map
$\zeta:D_{a_j}\to \CC$ given below is univalent.
\begin{equation}\label{def zeta1}
\zeta(z)=-N(\overline{a}_jz-\log z+ \ell_j).
    \end{equation}
This is linearly approximated by 
\begin{equation}\label{eq linear 1}
 \zeta(z)=\frac{1-|a_j|^2}{a_j}N(z-a_j)(1+{\cal O}(z-a_j)),\quad\text{when}\,\,z\to a_j.   
\end{equation}
Note that $\zeta$ maps ${\bf B}_j$ \eqref{def B} into the positive real axis and $\widehat{\bf B}_j$ \eqref{def bhat} into the negative real axis.

Let us define the diagonal matrix function whose diagonal entries are nonvanishing and analytic at $a_j$ by
\begin{equation}\label{def qj1}
Q_j(z)=I_{\nu+1}+\left(\frac{\zeta(z)^{c_j/2}z^{\sum c/2}}{(z-a_j)^{c_j/2}}-1\right){\bf e}_{00}+\left(\frac{(z-a_j)^{c_j/2}}{\zeta(z)^{c_j/2}z^{\sum c/2}}-1\right){\bf e}_{jj}.
    \end{equation}
We choose the branch cut of $\zeta^{c_j}$ at the negative real axis such that $Q_j$ is analytic in $D_{a_j}$.

We define the matrix functions ${\cal U}(a_j,z)$ and ${\cal F}_j(\zeta(z))$ by 
\begin{align}\label{def uaj1}
{\cal U}(a_j,z):&=
\displaystyle I_{\nu+1}-
\Bigg(\sum_{i\neq 0,j}\frac{\widetilde{\eta}_{kj}E_i(z)(z-a_i)^{c_i}}{E_j(z)[(z-a_j)^{c_j}]_{{\bf B}[i]}}{\bf e}_{ij}\Bigg),\quad  a_j\in \Gamma_{j0},\\
\label{fbdry}
{\cal F}_j(\zeta(z)):&=I_{\nu+1}+f_{c_j}(\zeta(z)){\bf e}_{0j},\end{align}
where $f_c(\zeta)$ is defined at \eqref{def fc}. See Appendix \ref{appendix} for more detail about $f_c(\zeta)$.

Inside $D_{a_j}$, from the definitions of ${\cal U}(a_j,z)$ and ${\cal F}_j(\zeta(z))$, one can see that
\begin{align}\label{def uaj1}
 {\cal U}_+(a_j,z)&={\cal U}_-(a_j,z)\left(I_{\nu+1}+
\left(\frac{{\widetilde{\eta}_{kj}}E_k(z)(z-a_k)^{c_k}}{E_j(z)[(z-a_j)^{c_j}]_{{\bf B}[k]}}{\bf e}_{kj}\right)\bigg|_+^{-}\right),\quad z\in{\bf B}_{jk},\\\label{def faj1}
 {\cal F}_{j,+}(\zeta(z))&={\cal F}_{j,-}(\zeta(z))\left(I_{\nu+1}+\left(\eta_j-1\right)\frac{\ee^{\zeta(z)}}{[\zeta(z)^{c_j}]_-}{\bf e}_{0j}\right),\quad z\in \widehat{\mathbf{B}}_j,
\end{align}
where, in \eqref{def uaj1}, $\pm$ means the boundary value evaluated from the $+$ or $-$ side of ${\bf B}_{jk}$.

\begin{lemma}\label{lemma local1}
${\cal F}_j(\zeta)Q_j(z)^{-1}{\cal U}(a_j,z)\Phi(z)$ satisfies the exact Riemann-Hilbert problem of $S(z)$ in $D_{a_j}$.\end{lemma}
\begin{proof}
When $z\in\widehat{\mathbf{B}}_j\cap D_{a_j}$,
using the identity $z^n/E_j(z)=e^{\zeta}$ from the definition of $\zeta$ in \eqref{def zeta1}, the jump of ${\cal F}_j(\zeta)Q_j(z)^{-1}{\cal U}(a_j,z)$ is given by
\begin{equation}\label{eq:w}
\begin{split}
&\quad\left({\cal F}_j(\zeta)Q_j(z)^{-1}{\cal U}(a_j,z)\right)_-^{-1}\left({\cal F}_j(\zeta)Q_j(z)^{-1}{\cal U}(a_j,z)\right)_+\\
&={\cal U}(a_j,z)^{-1}Q_{j}(z){\cal F}_{j,-}(\zeta)^{-1}{\cal F}_{j,+}(\zeta)Q_{j}(z)^{-1}{\cal U}(a_j,z)\\
&={\cal U}(a_j,z)^{-1}Q_{j}(z)\left(I_{\nu+1}+\left(\eta_j-1\right)\frac{\ee^{\zeta(z)}}{[\zeta(z)^{c_j}]_-}{\bf e}_{0j}\right)Q_{j}(z)^{-1}{\cal U}(a_j,z)\\
&={\cal U}(a_j,z)^{-1}\left(I_{\nu+1}+\frac{\left(\eta_j-1\right)[z^{n+\sum c}]_-}{E_j(z)(z-a_j)^{c_j}}{\bf e}_{0j}\right){\cal U}(a_j,z)\\
&=I_{\nu+1}+\frac{\left(\eta_j-1\right)[z^{n+\sum c}]_-}{E_j(z)(z-a_j)^{c_j}}{\bf e}_{0j},
  \end{split}
\end{equation}
which matches the jump condition of $S\Phi^{-1}$ on $\widehat{\mathbf{B}}_j\cap D_{a_j}$ in \eqref{jump of sphi1}. In the 3rd equality we have used that $[\zeta^{c_j} z^{c_j}]_-=[\zeta^{c_j} z^{c_j}]_+$. We remark that $Q_j$ was chosen in hindsight such that \eqref{eq:w} holds while the jump of ${\cal F}_j$ at \eqref{def faj1} is written only in terms of the local coordinate $\zeta$.  

When $z\in{\bf B}_{jk}\cap D_{a_j}$, we have
\begin{equation}\label{eq cal p1}
\begin{split}
\left({\cal F}_j(\zeta)Q_j(z)^{-1}{\cal U}(a_j,z)\right)_-^{-1}\left({\cal F}_j(\zeta)Q_j(z)^{-1}{\cal U}(a_j,z)\right)_+&={\cal U}_-(a_j,z)^{-1}{\cal U}_+(a_j,z)\\
&=I_{\nu+1}+
\left(\frac{{\widetilde{\eta}_{kj}}E_k(z)(z-a_k)^{c_k}}{E_j(z)[(z-a_j)^{c_j}]_{{\bf B}[k]}}{\bf e}_{kj}\right)\bigg|_+^{-}
\\&=\displaystyle I_{\nu+1}+
(\eta_{j}-1)\frac{\eta_{kj}E_k(z)(z-a_k)^{c_k}}{E_j(z)(z-a_j)^{c_j}}{\bf e}_{kj},
  \end{split}
\end{equation}
which agrees with the jump of $S\Phi^{-1}$ on ${\mathbf{B}}_{jk}\cap D_{a_j}$ in \eqref{jump of sphi3}. In the 2nd equality $\pm$ means the boundary value evaluated from the $+$ or $-$ side of ${\bf B}_{jk}$. The last equality is obtained by Lemma \ref{za ratio} and the relation \eqref{eq: wj wk}.

Lastly, we need to show that ${\cal F}_jQ_j^{-1}{\cal U}\Phi$ satisfies the boundedness of $S$ in \eqref{srhp} as $z\to a_j$.  

When $z\in D_{a_j}\cap\Omega_j$ we have
\begin{align}\nonumber
Q_j(z){\cal F}_j(\zeta(z))Q_j(z)^{-1}{\cal U}(a_j,z)\Phi(z)M_j(z)
 &=\left(I_{\nu+1}+\frac{\zeta(z)^{c_j}z^{\sum c}}{(z-a_j)^{c_j}}f_{c_j}(\zeta(z)){\bf e}_{0j}\right){\cal U}(a_j,z)\Phi(z)M_j(z)
 \\\nonumber
 &=\Phi(z)+\frac{\zeta(z)^{c_j}z^{\sum c}}{W_j(z)}\left(\frac{z^{n}}{E_j(z)\zeta(z)^{c_j}}-f_{c_j}(\zeta(z))\right){\bf e}_{00}\\\nonumber
 &=\Phi(z)+\frac{\zeta(z)^{c_j}z^{\sum c}}{W_j(z)}\left(\frac{\ee^{\zeta(z)}}{\zeta(z)^{c_j}}-f_{c_j}(\zeta(z))\right){\bf e}_{00}.
\end{align}
Since $\ee^{\zeta}/\zeta^{c_j}-f_{c_j}(\zeta)$ is an entire function in $\zeta$, the $(0,0)$th entry is bounded. The boundedness of the other entries follow from the boundedness of the corresponding entries in $\Phi(z)$. By a similar argument, $Q_j(z){\cal F}_j(\zeta(z))Q_j(z)^{-1}{\cal U}(a_j,z)\Phi(z) M_0(z)={\cal O}(1)$ as $z\to a_j$ and $z\in\Omega_0$. This ends the proof of Lemma \ref{lemma local1}.
\end{proof}

\subsection{\texorpdfstring{$a_j\in \Gamma_{jk}$}{ajgammajk}}\label{subsection jk}
Similar to the above subsection we define
\begin{equation}\label{def zeta2}
\zeta(z)=-N((\overline{a}_j-\overline{a}_k)z+ \ell_j- \ell_k).    
\end{equation}
This is linearly approximated by 
$$
\zeta(z)=N(\overline{a}_k-\overline{a}_j)(z-a_j)(1+{\cal O}(z-a_j)),\quad \text{when}\,\, z\to a_j.
$$
Note that $\zeta$ maps $\Gamma_{jk}$ into the imaginary axis and ${\bf B}_{jk}$ into the negative real axis.

Let us define the diagonal matrix function whose diagonal entries are nonvanishing and analytic at $a_j$ by
\begin{equation}\label{def qj2}
   Q_j(z)=I_{\nu+1}+\left(\frac{[(z-a_j)^{c_j/2}]_{{\bf B}[k]}}{\zeta^{c_j/2}(z-a_k)^{c_k/2}}-1\right){\bf e}_{jj}+\left(\frac{\zeta^{c_j/2}(z-a_k)^{c_k/2}}{[(z-a_j)^{c_j/2}]_{{\bf B}[k]}}-1\right){\bf e}_{kk}. 
\end{equation}
We require $Q_j$ being analytic in $D_{a_j}$ therefore we set that $\zeta^{c_j/2}$ has the branch cut on the negative real axis which results in $(z-a_j)^{c_j/2}$ having the branch cut on ${\bf B}_{jk}$.  The subscript at $[(z-a_j)^{c_j/2}]_{{\bf B}[k]}$ refers to this fact. 

We define the matrix functions ${\cal U}(a_j,z)$ and ${\cal F}_j(\zeta)$ by 
\begin{align}\label{def uaj2}
{\cal U}(a_j,z):&=
\displaystyle I_{\nu+1}+\Bigg(\frac{ \eta' z^{n+\sum c}}{E_j(z)(z-a_j)^{c_j}}{\bf e}_{0j}-
\sum_{i\neq0, j,k}\frac{E_i(z)(z-a_i)^{c_i}}{E_j(z)[(z-a_j)^{c_j}]_{{\bf B}[i]}}{\bf e}_{ij}\Bigg),\quad a_j\in \Gamma_{jk},\\
{\cal F}_j(\zeta(z)):&=I_{\nu+1}-\widetilde{\eta}_{kj}f_{c_j}(\zeta(z)){\bf e}_{kj}.
\end{align}
Above we define the constant $\eta'=1$ if $\widehat{\bf B}_j$ \eqref{def bhat} sits in $\Omega_j$, while $\eta'=W_j(z)/W_k(z)\big|_{z\in\Omega_k\cap U_{a_j}}$ if  $\widehat{\bf B}_j$ sits in $\Omega_k$.

Inside $D_{a_j}$, from the definitions of ${\cal U}(a_j,z)$ and ${\cal F}_j(\zeta)$, one can see that
\begin{align}
 {\cal U}_+(a_j,z)&={\cal U}_-(a_j,z)\left(I_{\nu+1}+\left(\frac{z^{n+\sum c}}{E_j(z)(z-a_j)^{c_j}}{\bf e}_{0j}\right)\bigg|^+_-\right),\quad z\in\widehat{\bf B}_{j}\cap \Omega_j,\\
 {\cal U}_+(a_j,z)&={\cal U}_-(a_j,z)\left(I_{\nu+1}+\left(\frac{W_j(z)}{W_k(z)}\frac{z^{n+\sum c}}{E_j(z)(z-a_j)^{c_j}}{\bf e}_{0j}\right)\bigg|^+_-\right),\quad z\in\widehat{\bf B}_{j}\cap \Omega_k,\\
 \label{eq uaj2}{\cal U}_+(a_j,z)&={\cal U}_-(a_j,z)\left(I_{\nu+1}+
\left(\frac{E_i(z)(z-a_i)^{c_i}}{E_j(z)[(z-a_j)^{c_j}]_{{\bf B}[i]}}{\bf e}_{ij}\right)\bigg|_+^{-}\right),\quad z\in{\bf B}_{ji},\\\label{def faj2}
 {\cal F}_{j,+}(\zeta)&={\cal F}_{j,-}(\zeta)\left(I_{\nu+1}+\left(\eta_j-1\right)\frac{{\widetilde{\eta}_{kj}}\ee^{\zeta(z)}}{[\zeta(z)^{c_j}]_-}{\bf e}_{kj}\right),\quad z\in {\mathbf{B}}_{jk},
\end{align}
where, in the above equations, $\pm$ means the boundary value evaluated from the $+$ or $-$ side of the corresponding contour.

\begin{lemma}\label{lemma loc2}
${\cal F}_j(\zeta)Q_j(z)^{-1}{\cal U}(a_j,z)\Phi(z)$ satisfies the exact Riemann-Hilbert problem of $S(z)$ in $D_{a_j}$.\end{lemma}
\begin{proof}

On $\mathbf{B}_{jk}\cap D_{a_j}$, 
using the identity $E_k(z)/E_j(z)=e^{\zeta}$ from the definition of $\zeta$ in \eqref{def zeta2}, the jump of ${\cal F}_j(\zeta)Q_j(z)^{-1}{\cal U}(a_j,z)$ is given by
\begin{equation}\label{eq:fqu}
\begin{split}
&\quad\left({\cal F}_j(\zeta)Q_j(z)^{-1}{\cal U}(a_j,z)\right)_-^{-1}\left({\cal F}_j(\zeta)Q_j(z)^{-1}{\cal U}(a_j,z)\right)_+\\
&={\cal U}(a_j,z)^{-1}Q_{j}(z){\cal F}_{j,-}(\zeta)^{-1}{\cal F}_{j,+}(\zeta)Q_{j}(z)^{-1}{\cal U}(a_j,z)\\
&={\cal U}(a_j,z)^{-1}Q_{j}(z)\left(I_{\nu+1}+(\eta_{j}-1) \frac{{\widetilde{\eta}_{kj}}\ee^{\zeta}}{[\zeta^{c_j}]_-}{\bf e}_{kj}\right)Q_{j}(z)^{-1}{\cal U}(a_j,z)\\
&={\cal U}(a_j,z)^{-1}Q_{j}(z)\left(I_{\nu+1}+\eta_{kj}(\eta_{j}-1)\frac{[(z-a_j)^{c_j}]_+}{(z-a_j)^{c_j}}  \frac{\ee^{\zeta}}{[\zeta^{c_j}]_-}{\bf e}_{kj}\right)Q_{j}(z)^{-1}{\cal U}(a_j,z)\\
&={\cal U}(a_j,z)^{-1}\left(I_{\nu+1}+
\eta_{kj}(\eta_j-1)\frac{E_k(z)(z-a_k)^{c_k}}{E_j(z)(z-a_j)^{c_j}}{\bf e}_{kj}\right){\cal U}(a_j,z)\\
&=I_{\nu+1}+
\eta_{kj}(\eta_j-1)\frac{E_k(z)(z-a_k)^{c_k}}{E_j(z)(z-a_j)^{c_j}}{\bf e}_{kj},
  \end{split}
\end{equation}
which matches the jump condition of $S\Phi^{-1}$ on $\mathbf{B}_{jk}\cap D_{a_j}$ in \eqref{jump of sphi3}. In the 3rd equality $[(z-a_j)^{c_j}]_+$ comes from the $+$ side of $[(z-a_j)^{c_j}]_{{\bf B}[k]}$. The 3rd equality is obtained by Lemma \ref{za ratio} and the relation \eqref{eq: wj wk}.  We also note that the $-$ side of $\zeta^{c_j}$ corresponds to the $+$ side of ${\bf B}_{jk}$ \eqref{def bk}, hence $[\zeta^{c_j}]_-$ refers to the boundary values evaluated from the $+$ side of ${\bf B}_{jk}$. We remark that $Q_j$ was chosen in hindsight such that \eqref{eq:fqu} holds while the jump of ${\cal F}_j$ at \eqref{def faj2} is written only in terms of the local coordinate $\zeta$.  

The other jump conditions of $S\Phi^{-1}$ in \eqref{jump of sphi1}, \eqref{jump of sphi2} and \eqref{jump of sphi3} can be satisfied by the following calculations.

When $z\in\widehat{\bf B}_{j}\cap D_{a_j}\cap \Omega_j$ and $a_j\in\Gamma_{jk}$ we have
\begin{equation}\nonumber
\begin{split}
\left({\cal F}_j(\zeta)Q_j(z)^{-1}{\cal U}(a_j,z)\right)_-^{-1}\left({\cal F}_j(\zeta)Q_j(z)^{-1}{\cal U}(a_j,z)\right)_+&={\cal U}_-(a_j,z)^{-1}{\cal U}_+(a_j,z)\\
&=I_{\nu+1}+
\left(\frac{z^{n+\sum c}}{E_j(z)(z-a_j)^{c_j}}{\bf e}_{0j}\right)\bigg|_+^{-}
\\&=\displaystyle I_{\nu+1}+\frac{\left(\eta_j-1\right)[z^{n+\sum c}]_-}{E_j(z)(z-a_j)^{c_j}}{\bf e}_{0j}.
  \end{split}
\end{equation}

When $z\in\widehat{\bf B}_{j}\cap D_{a_j}\cap \Omega_k$ and $a_j\in\Gamma_{jk}$ we have
\begin{equation}\nonumber
\begin{split}
\left({\cal F}_j(\zeta)Q_j(z)^{-1}{\cal U}(a_j,z)\right)_-^{-1}\left({\cal F}_j(\zeta)Q_j(z)^{-1}{\cal U}(a_j,z)\right)_+&={\cal U}_-(a_j,z)^{-1}{\cal U}_+(a_j,z)\\
&=I_{\nu+1}+
\left(\frac{W_j(z)}{W_k(z)}\frac{z^{n+\sum c}}{E_j(z)(z-a_j)^{c_j}}{\bf e}_{0j}\right)\bigg|_+^{-}
\\&=\displaystyle I_{\nu+1}+\frac{W_j(z)}{W_k(z)}\frac{\left(\eta_j-1\right)[z^{n+\sum c}]_-}{E_j(z)(z-a_j)^{c_j}}{\bf e}_{0j}.
  \end{split}
\end{equation}

When $z\in{\bf B}_{ji}\cap D_{a_j}$ and $a_j\in\Gamma_{jk}$ the calculation is similar to that in \eqref{eq cal p1}.

Lastly, we need to show that ${\cal F}_jQ_j^{-1}{\cal U}\Phi$ satisfies the boundedness of $S$ in \eqref{srhp} as $z\to a_j$.

When $z\in D_{a_j}\cap\Omega_j$ we have
\begin{align*}
 Q_j(z){\cal F}_j(\zeta(z))Q_j(z)^{-1}{\cal U}(a_j,z)\Phi(z)M_j(z)
 &=\left(I_{\nu+1}-\frac{\zeta^{c_j}(z-a_k)^{c_k}{\widetilde{\eta}_{kj}}}{[(z-a_j)^{c_j}]_{{\bf B}[k]}}f_{c_j}(\zeta(z)){\bf e}_{kj}\right){\cal U}(a_j,z)\Phi(z)M_j(z)
 \\\nonumber
 =\Phi(z)&+\left(\frac{(z-a_j)^{c_j}{\widetilde{\eta}_{kj}}f_{c_j}(\zeta(z))\zeta^{c_j}(z-a_k)^{c_k}}{W_j(z)[(z-a_j)^{c_j}]_{{\bf B}[k]}}-\frac{E_k(z)(z-a_k)^{c_k}}{E_j(z)W_k(z)}\right){\bf e}_{j0}
 \\
 &=\Phi(z)+\frac{\zeta(z)^{c_k}(z-a_k)^{c_k}}{W_k(z)}\left(f_{c_j}(\zeta(z))-\frac{E_k(z)}{E_j(z)\zeta(z)^{c_j}}\right){\bf e}_{j0}\\
  &=\Phi(z)+\frac{\zeta(z)^{c_k}(z-a_k)^{c_k}}{W_k(z)}\left(f_{c_j}(\zeta(z))-\frac{\ee^{\zeta(z)}}{\zeta(z)^{c_j}}\right){\bf e}_{j0}.
\end{align*}
The 3rd equality is obtained by \eqref{def etakjtilde}. Since $f_{c_j}(\zeta)-\ee^\zeta/\zeta^{c_j}$ is an entire function in $\zeta$, the $(j,0)$th entry is bounded. The boundedness of the other entries are inherited from the boundedness of the corresponding entries in $\Phi(z)$. By a similar argument, we also have $Q_j(z){\cal F}_j(\zeta)Q_j(z)^{-1}{\cal U}(a_j,z)\Phi(z)M_k(z)={\cal O}(1)$ as $z\to a_j$, $z\in\Omega_k$ and $k\neq j$. This ends the proof of Lemma \ref{lemma loc2}.
\end{proof}

\subsection{Construction of \texorpdfstring{${\cal R}$}{r} and \texorpdfstring{$H$}{h}}

To match the local parametrices obtained in Section \ref{subsection j0} and \ref{subsection jk} with the global parametrix obtained in Section \ref{global} along $\partial D_{a_j}$, we need to modify the global parametrix $\Phi$ into ${\cal R}\Phi$ with a rational function ${\cal R}$ with poles at $\{a_j\text{'s}\}$. This is called the ``partial Schlesinger transform'' \cite{Bertola 2008}, and it was used also for $\nu=1$ \cite{Lee 2017}.

From the  the asymptotic expansion of $f_c(\zeta)$ in Appendix \ref{appendix} the asymptotic expansion of ${\cal F}_j(\zeta)$ as $\zeta\to\infty$ follows. We denote the coefficients in the expansion by $\alpha_i(c_j)$ as below.

\begin{equation}\label{FF22}
{\cal F}_j(\zeta)=\begin{cases}
\displaystyle I_{\nu+1} +\sum_{i=1}^\infty\frac{\alpha_i(c_j)}{\zeta^i} {\bf e}_{0j} ,&\mbox{if } a_j\in \Gamma_{j0},\\
\displaystyle I_{\nu+1} -\widetilde{\eta}_{kj}\left(\sum_{i=1}^\infty\frac{\alpha_i(c_j)}{\zeta^i}\right){\bf e}_{kj},&\mbox{if } a_j\in \Gamma_{jk}.
\end{cases}
\end{equation}
In fact $\alpha_i(c_j)$ can be explicitly written by \eqref{exp fc}, see Appendix \ref{appendix}.

Let ${\cal F}_j^{(m)}(\zeta)$ be  the truncated asymptotic expansion of ${\cal F}_j(\zeta)$ given by
 \begin{equation}\label{hat fm}
      {\cal F}_j^{(m)}(\zeta)=\begin{cases}
\displaystyle I_{\nu+1}+\sum_{i=1}^m\frac{\alpha_i(c_j)}{\zeta^i}{\bf e}_{0j},&
\mbox{if } a_j\in \Gamma_{j0},\\
\displaystyle I_{\nu+1}-\widetilde{\eta}_{kj}\left(\sum_{i=1}^m\frac{\alpha_i(c_j)}{\zeta^i}\right){\bf e}_{kj} ,& \mbox{if } a_j\in \Gamma_{jk}.
\end{cases}
 \end{equation}

It follows that
\begin{equation}\label{fffm}
{\widehat{\cal F}}_j(\zeta):={\cal F}_j(\zeta)\left({\cal F}_j^{(m)}(\zeta)\right)^{-1}=I_{\nu+1}+{\cal O}\left(|\zeta|^{-m-1}\right),\quad{\text as}\quad |\zeta|\to \infty.    
\end{equation}

For a function $f$ with pole singularity at $z=a$ let us define
$$[f(z)]_{\text{sing}(a)}=\displaystyle\frac{1}{2\pi\ii}\oint_{a}\frac{f(s)}{z-s}\dd s,$$ which represents the singular part of the Laurent expansion of $f(z)$ at $z=a$. The integration contour circles around $a$ such that $a$ is the only singularity of the integrand inside the circle; especially $z$ must be outside the circle.

\begin{lemma}\label{handr}
Let ${\cal R}(z)$ be a rational matrix function of size $\nu+1$ by $\nu+1$ whose $(p,q)$th entry is given by $r_{p,q}(z)$ and $0\leq p,q\leq\nu$, where
\begin{equation}\label{r0jrkj}
    \left\{\begin{aligned}
\displaystyle r_{0,j}(z)&=\bigg[\frac{\zeta(z)^{c_j}z^{\sum c}}{(z-a_j)^{c_j}}\sum_{i=1}^{m}\frac{\alpha_i(c_j)}{\zeta(z)^i}\bigg]_{\text{sing}(a_j)},& a_j\in\Gamma_{j0},\vspace{0.1cm}\\
\displaystyle r_{k,j}(z)&=-\widetilde{\eta}_{kj}\bigg[\frac{\zeta(z)^{c_j} (z-a_k)^{c_k} }{[(z-a_j)^{c_j}]_{{\bf B}[k]}} \sum_{i=1}^{m}\frac{\alpha_i(c_j)}{\zeta(z)^i}\bigg]_{\text{sing}(a_j)},& a_j\in\Gamma_{jk}.
\end{aligned}\right.
\end{equation}
For any $p$ that belongs to the chain of $a_j$, i.e. $j\to k\to\dots\to p\to\dots$, we set $r_{p,j}$ by the recursively applying the following relation: 
\begin{equation}\label{eq r0j}
    r_{p,j}(z)=[r_{p,k}(z) r_{k,j}(z)]_{{\rm sing}(a_j)}.
\end{equation}
For all other entries we set $$r_{p,q}(z)=\delta_{pq},$$ where $\delta_{pq}=1$ for $p=q$ and $\delta_{pq}=0$ for $p\neq q$. 
With the above definitions the matrix function $H_j(z)$ given by \begin{equation}\label{def hj}
   H_j(z):=Q_j(z)^{-1}{\cal R}(z)Q_j(z)\left({\cal F}_j^{(m)}(\zeta(z))\right)^{-1}
\end{equation} is holomorphic at $a_j$.  We note that $r_{k,j}(z)$ has pole only at $a_j$ and only for $k$ that belongs to the chain of $a_j$.  

We have the following asymptotic behavior
\begin{equation}\label{r chain}
  r_{0,j}(z)= \frac{{\sf chain}(j)}{z-a_j}\left(1+{\cal O}\left(\frac{1}{N}\right)\right)
\end{equation}
uniformly over a compact subset of ${\mathbb C}\setminus D_{a_j}$, 
where the constant ${\sf chain}(j)$ is defined at \eqref{chain15}. 

For $p\neq 0$ and $z\in{\mathbb C}\setminus D_{a_j}$, we have
\begin{equation}
    r_{p,j}(z)={\cal O}(N^C)
\end{equation}
for a fixed finite $C$.
\end{lemma}
\begin{proof}
When $a_j\in \Gamma_{j0}$, ${\cal F}_j^{(m)}(z)$ is given in \eqref{hat fm}.  We have
\begin{equation}\nonumber
\begin{split}
H_j(z)
&=\displaystyle\sum_{q\neq j}\big[Q_j(z)^{-1}{\cal R}(z)Q_j(z)\big]_{pq}{\bf e}_{pq}+\left(\frac{(z-a_j)^{c_j}r_{0,j}(z)}{\zeta^{c_j}z^{\sum c}}-\left(\sum_{i=1}^{m}\frac{\alpha_i(c_j)}{\zeta^i}\right)\right){\bf e}_{0j}\\
&+\displaystyle\sum_{p\notin \{0, j\} }\left(\frac{(z-a_j)^{c_j/2}r_{p,j}(z)}{\zeta^{c_j/2}z^{\sum c/2}}-\frac{\zeta^{c_j/2}{z^{\sum c/2}}r_{p,0}(z)}{(z-a_j)^{c_j/2}}\left(\sum_{i=1}^{m}\frac{\alpha_i(c_j)}{\zeta^i}\right)\right){\bf e}_{pj}\\
&-\left(\frac{\zeta^{c_j}{z^{\sum c}} r_{j,0}(z)}{(z-a_j)^{c_j}}\sum_{i=1}^{m}\frac{\alpha_i(c_j)}{\zeta^i}\right){\bf e}_{jj}.
\end{split}\end{equation}
Let us discuss each term.  In the 1st summation, since $q\neq j$, $r_{p,q}(z)$ in the summation are all holomorphic at $a_j$ and it follows that the summation is holomorphic at $a_j$.  The 2nd term with ${\bf e}_{0j}$ becomes holomorphic exactly because of \eqref{r0jrkj}.  The 3rd term with the summation vanishes because $r_{p,j}=r_{p,0}=0$ for the corresponding $p$'s.  The last term also vanishes because $r_{j,0}=0$ by definition.

When $a_j\in\Gamma_{jk}$, we have
\begin{equation}\label{2ndHj}
\begin{split}
H_j(z)
&=\displaystyle\sum_{q\neq j}^\nu\big[Q_j(z)^{-1}{\cal R}(z)Q_j(z)\big]_{pq}{\bf e}_{pq}
\\&+ \frac{\zeta^{c_j/2}{(z-a_k)^{c_k/2}}}{[(z-a_j)^{c_j/2}]_{{\bf B}[k]}}\left(\frac{[(z-a_j)^{c_j}]_{{\bf B}[k]} r_{0,j}(z)}{\zeta^{c_j}(z-a_k)^{c_k}}+{r_{0,k}(z)}\widetilde{\eta}_{kj}\sum_{i=1}^{m}\frac{\alpha_i(c_j)}{\zeta^i}\right){\bf e}_{0j}\\
&+\displaystyle\left(\frac{\zeta^{c_j}{(z-a_k)^{c_k}}\,{r_{j,k}(z)}\widetilde{\eta}_{kj}}{[(z-a_j)^{c_j}]_{{\bf B}[k]}} \sum_{i=1}^{m}\frac{\alpha_i(c_j)}{\zeta^i}\right){\bf e}_{jj}+\left(\frac{[(z-a_j)^{c_j}]_{{\bf B}[k]}r_{k,j}(z)}{\zeta^{c_j}(z-a_k)^{c_k}}+\widetilde{\eta}_{kj}\sum_{i=1}^{m}\frac{\alpha_i(c_j)}{\zeta^i}\right){\bf e}_{kj}\\
&+\displaystyle\sum_{p\notin \{0,j,k\}}\left(\frac{[(z-a_j)^{c_j/2}]_{{\bf B}[k]} r_{p,j}(z)}{\zeta^{c_j/2}(z-a_k)^{c_k/2}}+\frac{\zeta^{c_j/2}{(z-a_k)^{c_k/2}}{r_{p,k}(z)}\widetilde{\eta}_{kj}}{[(z-a_j)^{c_j/2}]_{{\bf B}[k]}}\sum_{i=1}^{m}\frac{\alpha_i(c_j)}{\zeta^i}\right){\bf e}_{pj}.
\end{split}\end{equation}
The 1st term with the summation is holomorphic by the similar argument as above.   The term with ${\bf e}_{jj}$ vanishes because $r_{j,k}(z)=0$ for $j$ does not belong to the chain of $a_k$.  The term with ${\bf e}_{kj}$ is holomorphic at $a_j$ exactly by the definition \eqref{r0jrkj}.  For the term with ${\bf e}_{0j}$ to be holomorphic one obtains the following (recursive) relation:
\begin{equation}\label{recursive-proof}
   r_{0,j}(z)= -\left[{r_{0,k}(z)}\frac{\zeta^{c_j}(z-a_k)^{c_k}\widetilde{\eta}_{kj}}{[(z-a_j)^{c_j}]_{{\bf B}[k]}}\sum_{i=1}^{m}\frac{\alpha_i(c_j)}{\zeta^i}\right]_{\text{sing}(a_j)}  \text{ when $a_j\in\Gamma_{jk}$.}
\end{equation}
In the last term with the summation in \eqref{2ndHj} to be holomorphic we obtain  
\begin{equation}\nonumber
   r_{p,j}(z)= -\left[{r_{p,k}(z)}\frac{\zeta^{c_j}(z-a_k)^{c_k}\widetilde{\eta}_{kj}}{[(z-a_j)^{c_j}]_{{\bf B}[k]}}\sum_{i=1}^{m}\frac{\alpha_i(c_j)}{\zeta^i}\right]_{\text{sing}(a_j)}  \text{ when $j\to k\to p\neq 0$.}
\end{equation}
The above two relations combined with the definition of $r_{k,j}(z)$ \eqref{r0jrkj},  gives \eqref{eq r0j}.

Now let us prove the asymptotic behaviors of $r_{0,j}(z)$ when $a_j\in \partial\Omega_0$. Using the linear approximation \eqref{eq linear 1} we may write
\begin{equation}\label{124}
\begin{split}
r_{0,j}(z)&=
\bigg[\frac{\zeta(z)^{c_j}z^{\sum c}}{(z-a_j)^{c_j}}\frac{\alpha_1(c_j)}{\zeta(z)}\bigg]_{\text{ sing}(a_j)}+\bigg[\frac{\zeta(z)^{c_j}z^{\sum c}}{(z-a_j)^{c_j}}\sum_{i=2}^{m}\frac{\alpha_i(c_j)}{\zeta(z)^i}\bigg]_{\text{sing}(a_j)}
\\
&=\frac{N^{c_j-1} a_j^{1+\sum_{i\neq j} c_i}}{\Gamma(c_j) (1-|a_j|^2)^{1-c_j}}\frac{1}{z-a_j} + \frac{1}{2\pi \ii}\oint_{a_j} \frac{\zeta(s)^{c_j}s^{\sum c}}{(s-a_j)^{c_j}}\sum_{i=2}^{m}\frac{\alpha_i(c_j)}{\zeta(s)^i}\frac{\dd s}{z-s}
\\
&=\displaystyle\frac{N^{c_j-1} a_j^{1+\sum_{i\neq j} c_i}}{\Gamma(c_j) (1-|a_j|^2)^{1-c_j}}\frac{1}{z-a_j}\left(1+{\cal O}\left(\frac{1}{N}\right)\right),\quad z\notin D_{a_j}.
 \end{split}   
\end{equation}
For the integration in the 2nd line we may choose the circular integration contour centered at $a_j$ with half the radius of $D_{a_j}$ such that $|\zeta(s)|> C N$ over the integration contour for some positive constant $C$.  Then the integral is bounded by ${\cal O}(N^{c_j-2})$ and we obtain the estimate.  

Note that the coefficient of $1/(z-a_j)$ in the 1st term of the 2nd line in \eqref{124} comes from the evaluation of $\zeta(z)^{c_j} z^{\sum c}/(z-a_j)^{c_j}$ at $z=a_j$.  We use that 
$$\zeta(z)^{c_j}= N^{c_j}(1-|a_j|^2)^{c_j} \frac{(z-a_j)^{c_j}}{z^{c_j}}(1+{\cal O}(z-a_j))$$ to determine the exact branch of the exponents.

By a similar consideration, we obtain
\begin{equation}\label{bd rkj}
    r_{k,j}(z)=\displaystyle\frac{-N^{c_j-1}}{\Gamma(c_j)} \frac{\widetilde{\eta}_{kj}(a_j-a_k)^{c_k}|a_k-a_j|^{2c_j}}{[(a_k-a_j)^{c_j}]_{{\bf B}[k]} (\overline a_k-\overline a_j)}  \frac{1}{z-a_j}\left(1+{\cal O}\left(\frac{1}{N}\right)\right),~~ a_j\in\Gamma_{jk}, ~~k\neq 0,
\end{equation}
where we have used the following identity to evaluate the correct branch of the leading coefficient. 
$$ \zeta^{c_j} = N^{c_j}|a_k-a_j|^{2c_j} \frac{[(z-a_j)^{c_j}]_{{\bf B}[k]}}{[(a_k-a_j)^{c_j}]_{{\bf B}[k]}}(1+{\cal O}(z-a_j)).$$
Here $[(a_k-a_j)^{c_j}]_{{\bf B}[k]}$ means the evaluation of $[(z-a_j)^{c_j}]_{{\bf B}[k]}$ at $z=a_k$.

Finally we estimate the $r_{0,j}$ when $j\to k$ and $k\neq 0$ using the relation \eqref{recursive-proof}.
\begin{align*}
   r_{0,j}(z)&= -\left[{r_{0,k}(z)}\frac{\zeta^{c_j}(z-a_k)^{c_k}\widetilde{\eta}_{kj}}{[(z-a_j)^{c_j}]_{{\bf B}[k]}}\frac{\alpha_1(c_j)}{\zeta}\right]_{{\rm sing}(a_j)}
   -\left[{r_{0,k}(z)}\frac{\zeta^{c_j}(z-a_k)^{c_k}\widetilde{\eta}_{kj}}{[(z-a_j)^{c_j}]_{{\bf B}[k]}}\sum_{i=2}^{m}\frac{\alpha_i(c_j)}{\zeta^i}\right]_{{\rm sing}(a_j)}
   \\ &=
 -r_{0,k}(a_j)  
 \frac{N^{c_j-1}}{\Gamma(c_j)} \frac{\widetilde{\eta}_{kj}(a_j-a_k)^{c_k}|a_k-a_j|^{2c_j}}{[(a_k-a_j)^{c_j}]_{{\bf B}[k]} (\overline a_k-\overline a_j)}
 \frac{1}{z-a_j} + {\cal O}(\|r_{0,k}\|_\infty N^{c_j-2}),
\end{align*}
where $\|r_{0,k}\|_\infty$ is the norm taken over $D_{a_j}$ and the error bound is uniformly over a compact subset in ${\mathbb C}\setminus D_{a_j}$. 
This allows us to define the constant ${\sf chain}(j)$ by 
\begin{equation}\nonumber
N^{\mu_j}\lim_{N\to\infty}\frac{r_{0,j}(z)}{N^{\mu_j}}  = \frac{{\sf chain}(j)}{z-a_j}, \quad \text{ for all $j=1,\dots,\nu,$}
\end{equation}
where $\mu_{j}$ is fixed such that the limit is non-trivial.
From the above two equations we obtain the recurrence relation
$$ \mu_j = \mu_k + c_j-1  $$ 
and
$$ {\sf chain}(j)=-\frac{{\sf chain}(k)}{a_j-a_k} \frac{N^{c_j-1}\widetilde{\eta}_{kj}(a_j-a_k)^{c_k}|a_k-a_j|^{2c_j}}{\Gamma(c_j)[(a_k-a_j)^{c_j}]_{{\bf B}[k]} (\overline a_k-\overline a_j)}.   $$
For a given chain $j=k_s\to k_{s-1}\to\dots\to k_1\to 0$ the above relation provides the recurrence relation that can be solved with the initial condition given by \eqref{124} as below.
\begin{equation}\label{def cha}
\begin{aligned}
 \mu_j &= \sum_{i=1}^s(c_{k_i}-1),\\
 {\sf chain}(j)&= \frac{a_{k_1}^{1+\sum_{i\neq k_1} c_i}N^{\sum_{i=1}^s(c_{k_i}-1)}}{\Gamma(c_{k_1})(1-|a_{k_1}|^2)^{1-c_{k_1}}}\prod_{i=1}^{s-1} \frac{\widetilde{\eta}_{k_i,k_{i+1}}(a_{k_{i+1}}-a_{k_i})^{c_{k_i}}|a_{k_i}-a_{k_{i+1}}|^{2c_{k_{i+1}}}}{\Gamma(c_{k_{i+1}})[(a_{k_i}-a_{k_{i+1}})^{c_{k_{i+1}}}]_{{\bf B}[k_i]} | a_{k_i}- a_{k_{i+1}}|^2}.
 \end{aligned}
\end{equation}
Using the identity \eqref{def subbk} about the branches of multivalued functions the above expression of ${\sf chain}(j)$ becomes the original definition at \eqref{chain15}.
Note that $s$ is the level of $a_j$ as defined in Definition \ref{def chain} and, if $s=1$, the product part is one. 

For other $r_{p,j}(z)$'s with $p\neq 0$ similar estimates as in \eqref{124} and \eqref{bd rkj}  can be made and shown to be bounded by ${\cal O}(N^C)$ with a finite $C$. 
\end{proof}

From the above lemma we realize that all $r_{p,q}$'s grows (or decays) algebraically in $N$ away from $a_q$. Therefore we have \begin{equation}\label{eq rbdd}
    {\cal R}(z)={\cal O}\left(N^C\right),\quad{\mbox{when $z\in \partial D_{a_j}$}}\end{equation} for some fixed finite $C$.

When $z\in \partial D_{a_j}$, by the definition of $H_j(z)$ in \eqref{def hj} and the boundedness of ${\cal R}(z)$, we have 
$$H_j(z)={\cal O}\left(N^{C'}\right),$$ for some finite ${C'}$.

\section{Strong asymptotics}\label{section strong}

Combining all the constructions of the (improved) global and the local parametrices we define $S^\infty(z)$ by
\begin{equation}\label{Sfinal}
S^\infty(z):=
\begin{cases}
{\cal R}(z)\Phi(z),&z\notin \cup_{j=1}^\nu D_{a_j},\\
Q_j(z)H_j(z){\cal F}_j(\zeta(z))Q_j(z)^{-1}{\cal U}(a_j,z)\Phi(z),& z\in D_{a_j},\quad j=1,\dots,\nu.
\end{cases}
\end{equation}
This will be the strong asymptotics of $S$ \eqref{srhp} as $N\to\infty$ as we prove now.

\subsection{Error analysis}
We define the error matrix by
\begin{equation}\label{def error}
    {\cal E}(z):=S^\infty(z)S(z)^{-1}.
\end{equation}

\begin{lemma}
Let ${\cal E}(z)$ be given above. Then ${\cal E}(z)= I_{\nu+1}+{\cal O}\left(1/{N^{\infty}}\right)$ as $N\to\infty$ uniformly over a compact set of the corresponding region. Here the error bound ${\cal O}(1/N^\infty)$ stands for ${\cal O}(1/N^m)$ for an arbitrary $m>0$.
\end{lemma}
\begin{proof}

When $z\in\partial D_{a_j}$, we have
\begin{align*}
     {\cal E}_+(z){\cal E}_-(z)^{-1}
         &=Q_j(z)H_j(z){\cal F}_j(\zeta(z))Q_j(z)^{-1}{\cal U}(a_j,z){\cal R}(z)^{-1}\\
         &={\cal R}(z)Q_j(z)\left({\cal F}_j^{(m)}(\zeta(z))\right)^{-1}{\cal F}_j(\zeta(z))Q_j(z)^{-1}{\cal U}(a_j,z){\cal R}(z)^{-1}\\
         &={\cal R}(z)Q_j(z){\widehat{\cal F}_j}(\zeta)Q_j(z)^{-1}\left(I_{\nu+1}+{\cal O}\left(\ee^{-C N}\right)\right){\cal R}(z)^{-1} \\
         &={\cal R}(z)\left(I_{\nu+1}+{\cal O}\left({N^{c_j-1-m}}\right)\right){\cal R}(z)^{-1}
\end{align*} for some $C>0$. The 2nd equality is obtained by \eqref{def hj}. The 3rd equality is obtained by \eqref{fffm} and the fact that, since $E_j(z)$ is dominant near $a_j$, $\|\,{\cal U}-I_{\nu+1}\|_\infty$ is exponentially small as $N$ grows on $\partial D_{a_j}$ from \eqref{def uaj1} and \eqref{def uaj2}.
By Lemma \ref{handr}, ${\cal R}(z)$ can be written as an upper triangular matrix by a reordering of $\{a_i\}_{i=1}^\nu$ such that the chain $j\to k$ always satisfies $j>k$, for instance by reordering $\{a_i\text{'s}\}$ by their levels.  Then 
${\cal R}(z)-I_{\nu+1}$ is a nilpotent matrix and we have 
$${\cal R}(z)^{-1}=I_{\nu+1}+\sum_{i=1}^{\nu-1}(-1)^i\left({\cal R}(z)-I_{\nu+1}\right)^i=I_{\nu+1}+{\cal O}\left(N^C\right)$$ for some fixed finite $C$ that does not depend on $m$ \eqref{eq rbdd}.  It means that, by increasing $m$, we can make 
$\|{\cal E}_+(z){\cal E}_-(z)^{-1}-I_{\nu+1}\|_\infty$ on $\partial D_{a_j}$ as small as we want.

When $z$ is on all the jump contours of $S\Phi^{-1}$ and $z\notin \cup_{j=1}^\nu D_{a_j}$ we have
\begin{equation}\label{error 4}
\begin{split}
      & {\cal E}_+(z){\cal E}_-(z)^{-1} =\left(S^\infty(z)S(z)^{-1}\right)_+\left(S^\infty(z)S(z)^{-1}\right)_-^{-1}\\
&={\cal R}(z)
\left(S(z)\Phi(z)^{-1}\right)_+^{-1}\left(S(z)\Phi(z)^{-1}\right)_-
{\cal R}(z)^{-1}.
\end{split}
\end{equation}
By the jump conditions of $S\Phi^{-1}$ in  \eqref{jump of sphi1},\eqref{jump of sphi2},\eqref{jump of sphi3} and \eqref{jump of sphi4}, $\|(S\Phi^{-1})_+^{-1}(S\Phi^{-1})_--I_{\nu+1}\|_\infty$ is exponentially small as $N$ grows we have
$${\cal E}_+(z){\cal E}_-^{-1}(z)=I_{\nu+1}+{\cal O}\left(\ee^{-C N}\right),\quad \mbox{when $z\notin \cup_{j=1}^\nu D_{a_j}$}$$
for some $C>0$. 

By Lemma \ref{lemma local1} and Lemma \ref{lemma loc2}, $S(z)\Phi(z)^{-1}$ and $S^\infty(z)\Phi(z)^{-1}$ have the same jump conditions in $D_{a_j}$. Therefore, ${\cal E}$ does not have any jump in $D_{a_j}$.

From the boundedness of $S$ in \eqref{srhp} and the definition of $S^{\infty}$ in \eqref{Sfinal}, as $z\to a_j$ from $\Omega_i$, both $S(z)M_i(z)={\cal O}(1)$ and $S^{\infty}(z)M_i(z)={\cal O}(1)$ are bounded for $i\in\{0,\dots,\nu\}$. Therefore ${\cal E}=S^\infty M_i(SM_i)^{-1}$ is bounded $z\to a_j$ in $\Omega_i$.

Since ${\cal E}_+(z){\cal E}_-(z)^{-1}=I_{\nu+1}+{\cal O}\left(1/{N^{\infty}}\right)$ when $z\in\partial D_{a_j}$ and $j=1,\dots,\nu$ and the other jump conditions of ${\cal E}$ in \eqref{error 4} are exponentially small in $N$ away from $\partial D_{a_j}$, by the small norm theorem (e.g. Theorem 7.171 in \cite{Deift 1999}, \cite{DKMVZ 1999} or \cite{Yang 2018}),  we obtain ${\cal E}(z)= I_{\nu+1}+{\cal O}\left(1/{N^{\infty}}\right)$.
\end{proof} 

By the definition of ${\cal E}(z)$ in \eqref{def error} and the above lemma, we have
 \begin{equation}\label{ssinf}
S(z)= \left(I_{\nu+1}+{\cal O}\left(1/{N^{\infty}}\right)\right)S^\infty(z).     
 \end{equation}

\subsection{Proof of Theorem \ref{thm01}}

\begin{lemma} \label{errorR}
On $z\notin D_{a_j}$ we have
\begin{equation}
\left[    \left(I_{\nu+1}+{\cal O}\left(\frac{1}{N^\infty}\right)\right) {\cal R}(z) \right]_{\mbox{1st row}}= \left(1+{\cal O}\left(\frac{1}{N^\infty}\right)\right) [{\cal R}(z)]_{\mbox{1st row}}.
\end{equation}
\end{lemma}
\begin{proof}
Recall that ${\cal O}(1/N^\infty)$ stands for ${\cal O}(1/N^m)$ for an arbitrary $m>0$.  From Lemma \ref{handr} all the entries of ${\cal R}(z)$ is growing at most polynomially in $N$ and, therefore, we have $\left(I+{\cal O}(1/N^\infty)\right) {\cal R}(z)  = {\cal R}(z)+{\cal O}(1/N^\infty)$. Furthermore, by \eqref{r chain} in Lemma \ref{handr}, all the entries in the 1st row of ${\cal R}(z)$ are nonvanishing away from $\{a_j\}_{j=1}^\nu$.   This leads to $[{\cal O}(1/N^\infty)]_\text{1st row}=(1+{\cal O}(1/N^\infty)) [{\cal R}(z)]_\text{1st row}$.
\end{proof}

\begin{cor}\label{errorH}
On $z\in D_{a_j}$ we have
\begin{equation}\label{errorhh}
\left[    \left(I_{\nu+1}+{\cal O}\left(\frac{1}{N^\infty}\right)\right) Q_j(z) H_j(z) \right]_{\mbox{1st row}}= \left(1+{\cal O}\left(\frac{1}{N^\infty}\right)\right) [Q_j(z) H_j(z)]_{\mbox{1st row}}.
\end{equation}
\end{cor}
\begin{proof}
\begin{equation*}
\begin{aligned}
\mbox{LHS of \eqref{errorhh}}&=\left[    \left(I_{\nu+1}+{\cal O}\left(\frac{1}{N^\infty}\right)\right) {\cal R}(z)Q_j(z)\left({\cal F}_j^{(m)}(\zeta(z))\right)^{-1} \right]_{\mbox{1st row}} \\
&=\left[    \left(I_{\nu+1}+{\cal O}\left(\frac{1}{N^\infty}\right)\right) {\cal R}(z) \right]_{\mbox{1st row}}Q_j(z)\left({\cal F}_j^{(m)}(\zeta(z))\right)^{-1}\\
&=\left(1+{\cal O}\left(\frac{1}{N^\infty}\right)\right) [{\cal R}(z)]_{\mbox{1st row}}Q_j(z)\left({\cal F}_j^{(m)}(\zeta(z))\right)^{-1}\\
&=\mbox{RHS of \eqref{errorhh}},
\end{aligned}
\end{equation*}
where the 1st equality is obtained by the definition of $H_j$ in \eqref{def hj} and the 3rd equality is obtained by the above lemma.
\end{proof}

In the following proofs, we will use the facts
\begin{equation}\label{eq phimg}
\displaystyle\left[\Phi(z)M_j(z)G_j(z)^{-1}\right]_{\mbox{1st column}} =\displaystyle\left(\frac{z^{n+\sum c}}{W_j(z)},-\frac{(z-a_1)^{c_1}E_1(z)}{W_1(z)},
\dots,-\frac{(z-a_\nu)^{c_\nu}E_\nu(z)}{W_\nu(z)}\right)^T,\mbox{when $z\in\Omega_j$},
\end{equation} where $j=0,1,\dots,\nu.$

\begin{proof}(Proof of Theorem \ref{thm01})
 Using \eqref{tildeY}, \eqref{T}, \eqref{def S} and \eqref{ssinf}, when $z\in\Omega_0\setminus U$, we have
\begin{equation}\label{invers trans1}
\begin{split}
Y(z)&=\begin{bmatrix}
1&{\bf 0}\\
{\bf 0}&{\bf C}
\end{bmatrix}S(z)G_0(z)^{-1}\begin{bmatrix}
1&{\bf 0}\\
{\bf 0}&{\bf C}^{-1}
\end{bmatrix}\begin{bmatrix}
1&{\bf 0}\\
{\bf 0}&\Psi_0(z)
\end{bmatrix}\\
&=\begin{bmatrix}
1&{\bf 0}\\
{\bf 0}&{\bf C}
\end{bmatrix}\left(I_{\nu+1}+{\cal O}\left(\frac{1}{N^{\infty}}\right)\right)S^\infty(z)G_0(z)^{-1}\begin{bmatrix}
1&{\bf 0}\\
{\bf 0}&{\bf C}^{-1}
\end{bmatrix}\begin{bmatrix}
1&{\bf 0}\\
{\bf 0}&\Psi_0(z)
\end{bmatrix}.\end{split}\end{equation}
Here the error bound is uniformly over a compact subset of $\Omega_0$ therefore one can always choose $U$, the neighbourhood of $\bigcup_j \Gamma_{j0}$,   small enough such that the compact subset in question sits in $\Omega_0\setminus U$.

For $z\in\Omega_0\setminus D_{a_j},$ we have
\begin{equation}\nonumber
    \begin{split}
        p_n(z)&=[Y(z)]_{11}=\left[\left(I_{\nu+1}+{\cal O}\left(1/{N^{\infty}}\right)\right)S^{\infty}(z)G_0(z)^{-1}\right]_{11}  \\
        &=\left[\left(I_{\nu+1}+{\cal O}\left(1/{N^{\infty}}\right)\right){\cal R}(z)\Phi(z)G_0(z)^{-1}\right]_{11}\\
&=\left[I_{\nu+1}+{\cal O}\left(1/{N^{\infty}}\right)\right]_{\mbox{1st row}} 
\left[{\cal R}(z)\Phi(z)G_0(z)^{-1}\right]_{\mbox{1st column}}
\\
&=\displaystyle\frac{z^{n+\sum c}}{W(z)}\left(1+{\cal O}\left(\frac{1}{N^{\infty}}\right)\right).
    \end{split}
\end{equation}
Above we have used that ${\cal R}(z)$ has no off-diagonal entries along its 1st column, i.e. $[{\cal R}(z)]_{j0}=0$ for all $j\neq 0$ by Lemma \ref{handr}. The 2nd equality is obtained by \eqref{invers trans1}. The 3rd equality is obtained by \eqref{Sfinal}. The 4th equality is obtained by \eqref{eq phimg}.

 Using \eqref{tildeY}, \eqref{T}, \eqref{def S} and \eqref{ssinf}, when $z\in\Omega_j$ we have
\begin{equation}\label{invers trans2}
Y(z)=\begin{bmatrix}
1&{\bf 0}\\
{\bf 0}&{\bf C}
\end{bmatrix}S(z)M_j(z)G_j(z)^{-1}\begin{bmatrix}
1&{\bf 0}\\
{\bf 0}&{\bf C}^{-1}
\end{bmatrix}\begin{bmatrix}
1&{\bf 0}\\
{\bf 0}&\Psi_j(z)
\end{bmatrix}\left[{\begin{array}{c:c}
\begin{matrix}
1
\end{matrix}
&\begin{matrix}W(z)\widetilde\psi_1(z)-W_j(z)\psi_j(z)\end{matrix} \\\hdashline
\begin{matrix}
{\bf 0}
\end{matrix}
&\begin{matrix}I_{\nu}\end{matrix}
\end{array}}\right].\end{equation}

For $z\in\Omega_{j}\setminus D_{a_j}$ we have
\begin{align*}
        p_n(z)&=[Y(z)]_{11}=\left[\left(I_{\nu+1}+{\cal O}\left(1/{N^{\infty}}\right)\right)S^{\infty}(z)M_j(z)G_j(z)^{-1}\right]_{11}  \\
 &=\displaystyle\left[\left(I_{\nu+1}+{\cal O}\left(1/{N^{\infty}}\right)\right){\cal R}(z)\right]_{\mbox{1st row}}\left[\Phi(z)M_j(z)G_j(z)^{-1}\right]_{\mbox{1st column}} \\
&=\displaystyle E_j(z)\left(\frac{z^{n+\sum c}}{W_j(z)E_j(z)}-\sum_{i=1}^\nu\frac{r_{0,i}(z)(z-a_i)^{c_i}E_i(z)}{W_i(z)E_j(z)}\right)\left(1+{\cal O}\left(\frac{1}{N^{\infty}}\right)\right)
\\
&=\displaystyle-\frac{r_{0,j}(z)E_j(z)(z-a_j)^{c_j}}{W_j(z)}\left(1+{\cal O}\left(\frac{1}{N^{\infty}}\right)\right) \\
&=\displaystyle-\frac{E_j(z)(z-a_j)^{c_j}}{W_j(z)}\frac{{\sf chain}(j)}{z-a_j}\left(1+{\cal O}\left(\frac{1}{N}\right)\right), 
\end{align*}
where $j=k_s\to k_{s-1}\to \dots \to k_1\to 0 $. The 2nd equality is obtained by \eqref{invers trans2}. The 3rd equality is obtained by \eqref{Sfinal}. The 4th equality is obtained by Lemma \ref{errorR} and \eqref{eq phimg}. The 5th equality is obtained by the fact that $E_j(z)$ is dominant in $\Omega_j$. The last equality is obtained by \eqref{r chain}.

 Using \eqref{tildeY}, \eqref{T}, \eqref{def S} and \eqref{ssinf}, when $z\in\Omega_0\cap U$, we have
\begin{equation}\label{invers trans3}
\begin{split}
Y(z)&=\begin{bmatrix}
1&{\bf 0}\\
{\bf 0}&{\bf C}
\end{bmatrix}S(z)M_0(z)G_0(z)^{-1}\begin{bmatrix}
1&{\bf 0}\\
{\bf 0}&{\bf C}^{-1}
\end{bmatrix}\begin{bmatrix}
1&{\bf 0}\\
{\bf 0}&\Psi_0(z)
\end{bmatrix}\\
&=\begin{bmatrix}
1&{\bf 0}\\
{\bf 0}&{\bf C}
\end{bmatrix}\left(I_{\nu+1}+{\cal O}\left(\frac{1}{N^{\infty}}\right)\right)S^\infty(z)M_0(z)G_0(z)^{-1}\begin{bmatrix}
1&{\bf 0}\\
{\bf 0}&{\bf C}^{-1}
\end{bmatrix}\begin{bmatrix}
1&{\bf 0}\\
{\bf 0}&\Psi_0(z)
\end{bmatrix}.\end{split}\end{equation}

For $z\in\Omega_{0}\cap U$ and near $\Gamma_{j0}$ we have
\begin{align*}
        p_n(z)&=[Y(z)]_{11}=\left[\left(I_{\nu+1}+{\cal O}\left(1/{N^{\infty}}\right)\right)S^{\infty}(z)M_0(z)G_0(z)^{-1}\right]_{11}  \\
 &=\left[\left(I_{\nu+1}+{\cal O}\left(1/{N^{\infty}}\right)\right){\cal R}(z)\right]_{\mbox{1st row}}\left[\Phi(z)M_0(z)G_0(z)^{-1}\right]_{\mbox{1st column}}\\
&= \displaystyle z^n\left(\frac{z^{\sum c}}{W(z)}-\frac{(z-a_j)^{c_j}r_{0,j}(z)}{W_j(z)}\frac{E_j(z)}{z^n}-\sum_{i\neq j}\frac{(z-a_i)^{c_i}r_{0,i}(z)}{W_i(z)}\frac{E_i(z)}{z^n}\right)\left(1+{\cal O}\left(\frac{1}{N^{\infty}}\right)\right)\\
&= \displaystyle \frac{z^{n+\sum c}}{W(z)}\left(1+{\cal O}\left(\frac{1}{N^{\infty}}\right)\right)-\frac{E_j(z)(z-a_j)^{c_j}}{W_j(z)}\frac{{\sf chain}(j)}{z-a_j}\left(1+{\cal O}\left(\frac{1}{N}\right)\right).
\end{align*}
The 2nd equality is obtained by \eqref{invers trans3}. The 3rd equality is obtained by \eqref{Sfinal}. The 4th equality is obtained by Lemma \ref{errorR} and \eqref{eq phimg}. The last equality is obtained by \eqref{r chain} and the fact that $z^n$ and $E_j(z)$ are the most dominant in the vicinity of $\Gamma_{j0}$. A similar calculation can be done for $z\in \Omega_j$ and near $\Gamma_{j0}$. 

Similar to the case of $z\in\Omega_{j}\setminus D_{a_j}$, when $z\in\Omega_{j}$ and near $\Gamma_{jk}$ we have
\begin{align*}
        p_n(z)&=\displaystyle E_j(z)\left(\frac{z^{n+\sum c}}{W_j(z)E_j(z)}-\sum_{i=1}^\nu\frac{r_{0,i}(z)(z-a_i)^{c_i}E_i(z)}{W_i(z)E_j(z)}\right)\left(1+{\cal O}\left(\frac{1}{N^{\infty}}\right)\right)
\\
&=\displaystyle-\frac{r_{0,j}(z)E_j(z)(z-a_j)^{c_j}}{W_j(z)}\left(1+{\cal O}\left(\frac{1}{N^{\infty}}\right)\right)-\frac{r_{0,k}(z)E_k(z)(z-a_k)^{c_k}}{W_k(z)}\left(1+{\cal O}\left(\frac{1}{N^{\infty}}\right)\right)\\
&=\displaystyle-\frac{E_j(z)(z-a_j)^{c_j}}{W_j(z)}\frac{{\sf chain}(j)}{z-a_j}\left(1+{\cal O}\left(\frac{1}{N}\right)\right)-\frac{E_k(z)(z-a_k)^{c_k}}{W_k(z)}\frac{{\sf chain}(k)}{z-a_k}\left(1+{\cal O}\left(\frac{1}{N}\right)\right).
\end{align*}
 The 2nd equality is obtained by the fact that $E_j(z)$ and $E_k(z)$ are dominant in the vicinity of $\Gamma_{jk}$. The last equality is obtained by \eqref{r chain}. A similar calculation can be done for $z\in \Omega_k$ and near $\Gamma_{jk}$. 
 
 This ends the proof of Theorem \ref{thm01}.
\end{proof}

\subsection{Proof of Theorem \ref{thm02}}

\begin{proof}(Proof of Theorem \ref{thm02})
For $z\in D_{a_j}\cap\Omega_j$ and $a_j\in\Gamma_{j0}$ we have
\begin{align*}
        p_n(z)&=[Y(z)]_{11}=\left[\left(I_{\nu+1}+{\cal O}\left(1/{N^{\infty}}\right)\right)S^{\infty}(z)M_j(z)G_j(z)^{-1}\right]_{11}  \\
 &=\Big[\left(I_{\nu+1}+{\cal O}\left(1/{N^{\infty}}\right)\right)Q_j(z)H_j(z)\Big]_{\mbox{1st row}}
 \\ & \qquad \times\Big[{\cal F}_j(\zeta(z))Q_j(z)^{-1}{\cal U}(a_j,z)\Phi(z)M_j(z)G_j(z)^{-1}\Big]_{\mbox{1st column}}\\
 &=\left(I_{\nu+1}+{\cal O}\left(1/{N^{\infty}}\right)\right)\Big[Q_j(z)H_j(z)\Big]_{\mbox{1st row}}
 \\ & \qquad \times{\cal F}_j(\zeta(z))Q_j(z)^{-1}{\cal U}(a_j,z)\Big[\Phi(z)M_j(z)G_j(z)^{-1}\Big]_{\mbox{1st column}}.
\end{align*}
The 2nd equality is obtained by \eqref{invers trans2}. The 3rd equality is obtained by \eqref{Sfinal}. The 4th equality is obtained by Lemma \ref{errorR} and Corollary \ref{errorH}. Moreover, by the definition of $H_j$ in \eqref{def hj}, we have
\begin{equation}\nonumber
\Big[H_j(z)\Big]_{\mbox{1st row}}
=\displaystyle\sum_{q\neq j}\big[Q_j(z)^{-1}{\cal R}(z)Q_j(z)\big]_{0q}{\bf e}_{0q}+\left(\frac{(z-a_j)^{c_j}r_{0,j}(z)}{\zeta^{c_j}z^{\sum c}}-\sum_{i=1}^{m}\frac{\alpha_i(c_j)}{\zeta^i}\right){\bf e}_{0j}.
\end{equation}
Above, by Lemma \ref{handr}, $r_{0,q}$ grows (or decays) algebraically in $N$ away from $a_q$ for $q\neq j$. Let $h_j(z)$ be defined by \begin{equation}\label{estimate h}
\begin{split}
    h_j(z)&=\frac{(z-a_j)^{c_j}}{\zeta^{c_j}z^{\sum c}}r_{0,j}(z)-\sum_{i=1}^{m}\frac{\alpha_i(c_j)}{\zeta^i}
 \\
    &= \frac{(z-a_j)^{c_j}}{\zeta^{c_j}z^{\sum c}}
    \frac{N^{c_j-1} a_j^{1+\sum_{i\neq j} c_i}}{\Gamma(c_j) (1-|a_j|^2)^{1-c_j}}\frac{1}{z-a_j}\left(1+{\cal O}\left(\frac{1}{N}\right)\right)-\sum_{i=1}^{m}\frac{\alpha_i(c_j)}{\zeta^i} \\&={\cal O}\left(\frac{1}{N}\right),\quad z\in\partial D_{a_j},    
    \end{split}
        \end{equation}
    where we used \eqref{124} at the second equality. 
    Since $h_j$ is holomorphic in $D_{a_j}$ the above bound holds in $D_{a_j}$.
Therefore,
\begin{align*}
        p_n(z)
&=\displaystyle\left(
\frac{z^{n+\sum c}}{W_j(z)}-\frac{z^{\sum c}\zeta^{c_j}E_j(z)}{W_j(z)}\left(f_{c_j}(z)+h_j(z)\right)
\right)\left(1+{\cal O}\left(\frac{1}{N^{\infty}}\right)\right)\\
&=\frac{{z^{n+\sum c}}}{W_j(z)}\left(1+{\cal O}\left(\frac{1}{N^{\infty}}\right)\right)-\frac{{z^{n+\sum c}}\zeta^{c_j}}{W_j(z)\ee^{\zeta}}\left(f_{c_j}(\zeta)+{\cal O}\left(\frac{1}{N}\right)\right)\\
&=\displaystyle z^{n+\sum c}\frac{\zeta^{c_j}}{W_j(z)\ee^{\zeta}}\left(\frac{\ee^{\zeta}}{\zeta^{c_j}}-f_{c_j}(\zeta)+{\cal O}\left(\frac{1}{N}\right)\right).
\end{align*}
The 1st equality is obtained by the facts that $r_{0,q}$ grows (or decays) algebraically in $N$ away from $a_q$ for $q\neq j$ and $E_j(z)$ is dominant in $\Omega_j$. The 2nd equality is obtained by the identity $z^n/E_j(z)=e^{\zeta}$ from the definition of $\zeta$ in \eqref{def zeta1} and the estimate in \eqref{estimate h}. A similar calculation can be done for $z\in D_{a_j}\cap\Omega_0$ and $a_j\in\Gamma_{j0}$.

For $z\in D_{a_j}\cap \Omega_j$ and $a_j\in\Gamma_{jk}$ we have
\begin{align*}
        p_n(z)&=[Y(z)]_{11}=\left[\left(I_{\nu+1}+{\cal O}\left(1/{N^{\infty}}\right)\right)S^{\infty}(z)M_j(z)G_j(z)^{-1}\right]_{11} \\
 &=\left[\left(I_{\nu+1}+{\cal O}\left(1/{N^{\infty}}\right)\right)Q_j(z)H_j(z)\right]_{\mbox{1st row}}
 \\&\qquad\times\left[{\cal F}_j(\zeta(z))Q_j(z)^{-1}{\cal U}(a_j,z)\Phi(z)M_j(z)G_j(z)^{-1}\right]_{\mbox{1st column}}\\ 
 &=\left(I_{\nu+1}+{\cal O}\left(1/{N^{\infty}}\right)\right)\left[Q_j(z)H_j(z)\right]_{\mbox{1st row}}
 \\&\qquad\times{\cal F}_j(\zeta(z))Q_j(z)^{-1}{\cal U}(a_j,z)\left[\Phi(z)M_j(z)G_j(z)^{-1}\right]_{\mbox{1st column}}.
\end{align*}
The 2nd equality is obtained by \eqref{invers trans2}. The 3rd equality is obtained by \eqref{Sfinal}. The 4th equality is obtained by Lemma \ref{errorR} and Corollary \ref{errorH}. Moreover, by the definition of $H_j$ in \eqref{def hj} and the relation in \eqref{eq r0j}, we have
\begin{equation}\nonumber
\begin{split}
\Big[H_j(z)\Big]_{\mbox{1st row}}
&=\displaystyle\sum_{q\neq j}^\nu\big[Q_j(z)^{-1}{\cal R}(z)Q_j(z)\big]_{0q}{\bf e}_{0q}
\\&+ \frac{\zeta^{c_j/2}{(z-a_k)^{c_k/2}}{r_{0,k}(z)}}{[(z-a_j)^{c_j/2}]_{{\bf B}[k]}}\left(\frac{[(z-a_j)^{c_j}]_{{\bf B}[k]} r_{k,j}(z)}{\zeta^{c_j}(z-a_k)^{c_k}}+\widetilde{\eta}_{kj}\sum_{i=1}^{m}\frac{\alpha_i(c_j)}{\zeta^i}\right){\bf e}_{0j}.
\end{split}\end{equation}
Above, by Lemma \ref{handr}, $r_{0,q}$ grows (or decays) algebraically in $N$ away from $a_q$ for $q\neq j$. Let  $h_j(z)$ be defined by
\begin{equation}\label{estimate h2}h_j(z)=\frac{[(z-a_j)^{c_j}]_{{\bf B}[k]}}{\zeta^{c_j}(z-a_k)^{c_k}}\frac{r_{k,j}(z)}{\widetilde{\eta}_{kj}}+\sum_{i=1}^{m}\frac{\alpha_i(c_j)}{\zeta^i}={\cal O}\left(\frac{1}{N}\right),\quad z\in\partial D_{a_j},\end{equation} which is obtained by a similar argument in \eqref{estimate h}. 
Therefore,

\begin{align*}
        p_n(z)
 &=\displaystyle-\left(\frac{r_{0,k}(z)E_k(z)(z-a_k)^{c_k}}{W_k(z)}
\right)\left(1+{\cal O}\left(\frac{1}{N^{\infty}}\right)\right)\\
 &\displaystyle\quad-\left(
\frac{r_{0,k}(z)E_j(z){\widetilde{\eta}_{kj}}(z-a_j)^{c_j}\left(-f_{c_j}(\zeta)+h_j(z)\right)\zeta^{c_j}(z-a_k)^{c_k}}{W_j(z)[(z-a_j)^{c_j}]_{{\bf B}_{[k]}}}
\right)\left(1+{\cal O}\left(\frac{1}{N^{\infty}}\right)\right)\\
&=
\displaystyle -E_k(z)r_{0,k}(z)(z-a_k)^{c_k}\left(\frac{1}{W_k(z)}\left(1+{\cal O}\left(\frac{1}{N^{\infty}}\right)\right)-\frac{\zeta^{c_j}\left(f_{c_j}(\zeta)+{\cal O}\left(1/{N}\right)\right)}{W_k(z)\ee^{\zeta}}\right)
\\
&=
\displaystyle -\frac{E_k(z)(z-a_k)^{c_k}{\sf chain}(k)}{z-a_k}\frac{\zeta^{c_j}}{W_k(z)\ee^{\zeta}}\left(\frac{\ee^{\zeta}}{\zeta^{c_j}}-f_{c_j}(\zeta)+{\cal O}\left(\frac{1}{N}\right)\right).
\end{align*}
The 1st equality is obtained by the facts that $r_{0,q}$ grows (or decays) algebraically in $N$ away from $a_q$ for $q\neq j$ and $E_j(z)$ is dominant in $\Omega_j$. The 2nd equality is obtained by \eqref{def etakjtilde}, the identity $E_k(z)/E_j(z)=e^{\zeta}$ from the definition of $\zeta$ in \eqref{def zeta2} and the estimate in \eqref{estimate h2}.
A similar calculation can be done for $z\in D_{a_j}\cap\Omega_k$ and $a_j\in\Gamma_{jk}$. 

 This ends the proof of Theorem \ref{thm02}.
\end{proof}

\appendix
\section{\texorpdfstring{$f_c(\zeta)$}{fczeta} and its properties}\label{appendix}

Let us define $f_{c}(\zeta)$ by the two conditions, $f_{c}(\zeta)\to 0$ as $|\zeta|\to 0$ and $\ee^\zeta/\zeta^c-f_{c}(\zeta)$ is entire. The integral representation of $f_c(\zeta)$ can be written by
\begin{equation}\label{def fc}
    f_{c}(\zeta)=-\frac{1}{2\pi \ii}\int_{\cal L}\frac{\ee^s}{s^{c}(s-\zeta)}\dd s,\quad \zeta\notin {\mathbb R}_-\cup\{0\}.
\end{equation}
The integration contour ${\cal L}$ is enclosing the negative real axis counterclockwise from $-\infty -\ii\epsilon$ to $-\infty+\ii\epsilon$ for an infinitesimal $\epsilon>0$ such that $\zeta$ is on the other side of ${\cal L}$ from the negative real axis.  When $c$ is a positive integer $\zeta^c f_c(\zeta)$ is exactly the first $c$ terms in the Taylor expansion of $\exp(\zeta)$. We take the principal branch for $\zeta^{c}$. 

As $|\zeta|\to \infty$ we have the expansion
\begin{equation}\label{exp fc}
f_{c}(\zeta)=
\sum_{i=1}^\infty\frac{\alpha_i(c)}{\zeta^i}
\quad \text{ where }\quad  \alpha_i(c)=\frac{1}{2\pi \ii}\int_{\cal L}\frac{s^{i-1}\ee^s}{s^{c_j}}\dd s=\frac{\sin(c\pi)\Gamma(i-c)}{\pi(-1)^{i-1}}.
\end{equation}
We also note that 
$\alpha_1(c) = 1/\Gamma(c)$.  When $c$ is integer we notice that $\alpha_i=0$ when $i>c$ and, therefore, $f_c(\zeta)$ is written in terms of a finite truncation of the Taylor series of the exponential function.

Let us show that $\ee^\zeta/\zeta^c-f_{c}(\zeta)$ is an entire function in $\zeta$ as follows.
\begin{align}\nonumber
\frac{\ee^{\zeta}}{\zeta^{c}}-f_{c}(\zeta)=\frac{1}{2\pi \ii}\oint_{\zeta}\frac{\ee^s}{s^{c}(s-\zeta)}\dd s+\frac{1}{2\pi \ii}\int_{\cal L}\frac{\ee^s}{s^{c}(s-\zeta)}\dd s,
\end{align}
where the first integration contour is the small circle around $\zeta$ directed counterclockwise.  The two integration contours can be deformed into a single contour that encloses the negative real axis and $\zeta$, hence the resulting integral has the analytic continuation onto $\zeta\in{\mathbb R}_-\cup\{0\}$.

\end{document}